\documentclass[acmlarge]{acmart}

\setcopyright{none}

\usepackage{booktabs}   
\usepackage{subcaption} 
\usepackage{amsmath}

\usepackage{amssymb}
\usepackage{tikz-cd}
\usepackage{ifthen}
\usepackage{xspace}
\usepackage{xparse}
\usepackage{xstring}
\usepackage{mathpartir}
\usepackage{cleveref}
\usepackage{enumitem}
\usepackage{stmaryrd}
\usepackage{multicol}
\usepackage{tikz-cd}
\usepackage{mathtools}
\usepackage{stackengine}
\usepackage{scalerel}
\usepackage{wasysym}
\usepackage{fontawesome}
\usepackage{xcolor}
\usepackage{tcolorbox}
\usepackage{accents}
\usepackage{ulem}
\usepackage{multirow}
\usepackage{array,longtable}

\usetikzlibrary{matrix}
\usetikzlibrary{calc}
\usetikzlibrary{positioning}
\usetikzlibrary{decorations.pathreplacing}


\DeclareFontFamily{OT1}{pzc}{}
\DeclareFontShape{OT1}{pzc}{m}{it}{<-> s * [1.10] pzcmi7t}{}
\DeclareMathAlphabet{\mathpzc}{OT1}{pzc}{m}{it}

\crefname{figure}{fig.}{fig.}
\Crefname{figure}{Fig.}{Fig.}
\Crefname{section}{Sec.}{Sec.}

\definecolor{light-gray}{gray}{0.85}

\allowdisplaybreaks

\DeclareUnicodeCharacter{21D2}{$\Rightarrow$}
\DeclareUnicodeCharacter{2113}{$\mathpzc{l}$}
\DeclareUnicodeCharacter{2248}{$\approx$}
\DeclareUnicodeCharacter{2200}{$\forall$}
\DeclareUnicodeCharacter{2218}{$\circ$}
\DeclareUnicodeCharacter{2294}{$\sqcup$}
\DeclareUnicodeCharacter{393}{$\Gamma$}
\DeclareUnicodeCharacter{3A8}{$\Psi$}
\DeclareUnicodeCharacter{3B1}{$\alpha$}
\DeclareUnicodeCharacter{3B2}{$\beta$}
\DeclareUnicodeCharacter{3B7}{$\eta$}
\DeclareUnicodeCharacter{3BB}{$\lambda$}
\DeclareUnicodeCharacter{3BC}{$\mu$}
\DeclareUnicodeCharacter{2080}{$_0$}
\DeclareUnicodeCharacter{2081}{$_1$}
\DeclareUnicodeCharacter{2082}{$_2$}
\DeclareUnicodeCharacter{2261}{$\equiv$}
\DeclareUnicodeCharacter{2E1}{$^l$}
\DeclareUnicodeCharacter{2B3}{$^r$}
\DeclareUnicodeCharacter{2032}{$'$}
\DeclareUnicodeCharacter{21D4}{$\Leftrightarrow$}
\DeclareUnicodeCharacter{224A}{$\approxeq$}
\DeclareUnicodeCharacter{27E8}{$\langle$}
\DeclareUnicodeCharacter{27E9}{$\rangle$}
\DeclareUnicodeCharacter{27F6}{$\longrightarrow$}
\DeclareUnicodeCharacter{220E}{$\blacksquare$}
\DeclareUnicodeCharacter{27FA}{$\Leftrightarrow$}
\DeclareUnicodeCharacter{22A4}{$\top$}
\DeclareUnicodeCharacter{22A5}{$\bot$}
\DeclareUnicodeCharacter{03A3}{$\Sigma$}
\DeclareUnicodeCharacter{03C0}{$\pi$}
\DeclareUnicodeCharacter{25A1}{$\square$}
\DeclareUnicodeCharacter{2237}{$::$}
\DeclareUnicodeCharacter{2308}{$\lceil$}
\DeclareUnicodeCharacter{2309}{$\rceil$}
\DeclareUnicodeCharacter{22A2}{$\vdash$}

\newcolumntype{L}{>{$}l<{$}}
\newcolumntype{R}{>{$}r<{$}}
\theoremstyle{remark}

\makeatletter
\def\namedlabel#1#2{\begingroup
   \def\@currentlabel{#2}%
   \label{#1}\endgroup
}
\makeatother

\makeatletter
\newcommand{\customlabel}[2]{%
   \protected@write \@auxout {}{\string \newlabel {#1}{{#2}{\thepage}{#2}{#1}{}} }%
   \hypertarget{#1}{}
}
\makeatother

\makeatletter
\DeclareFontFamily{U}{MnSymbolA}{}
\DeclareFontShape{U}{MnSymbolA}{m}{n}{
    <-6>  MnSymbolA5
   <6-7>  MnSymbolA6
   <7-8>  MnSymbolA7
   <8-9>  MnSymbolA8
   <9-10> MnSymbolA9
  <10-12> MnSymbolA10
  <12->   MnSymbolA12}{}
\DeclareFontShape{U}{MnSymbolA}{b}{n}{
    <-6>  MnSymbolA-Bold5
   <6-7>  MnSymbolA-Bold6
   <7-8>  MnSymbolA-Bold7
   <8-9>  MnSymbolA-Bold8
   <9-10> MnSymbolA-Bold9
  <10-12> MnSymbolA-Bold10
  <12->   MnSymbolA-Bold12}{}
\DeclareSymbolFont{MnSyA}{U}{MnSymbolA}{m}{n}
\SetSymbolFont{MnSyA}{bold}{U}{MnSymbolA}{b}{n}

\DeclareRobustCommand{\overleftharpoon}{\mathpalette{\overarrow@\leftharpoonfill@}}
\DeclareRobustCommand{\overrightharpoon}{\mathpalette{\overarrow@\rightharpoonfill@}}
\def\leftharpoonfill@{\arrowfill@\leftharpoondown\mn@relbar\mn@relbar}
\def\rightharpoonfill@{\arrowfill@\mn@relbar\mn@relbar\rightharpoonup}

\DeclareMathSymbol{\leftharpoondown}{\mathrel}{MnSyA}{'112}
\DeclareMathSymbol{\rightharpoonup}{\mathrel}{MnSyA}{'100}
\DeclareMathSymbol{\mn@relbar}{\mathrel}{MnSyA}{'320}
\makeatother

\DeclareFontFamily{U}{mathx}{}
\DeclareFontShape{U}{mathx}{m}{n}{<-> mathx10}{}
\DeclareSymbolFont{mathx}{U}{mathx}{m}{n}
\DeclareMathAccent{\widecheck}{0}{mathx}{"71}

\newcommand{\vect}[1]{\overrightarrow{#1}}

\newcommand{\vGamma}{\vect \Gamma}
\newcommand{\vDelta}{\vect \Delta}

\newcommand{\To}{\Longrightarrow}
\newcommand{\STo}{\Rightarrow}

\newcommand{\Nat}{\ensuremath{\texttt{Nat}}}

\newcommand{\tbox}{\ensuremath{\texttt{box}}\xspace}

\newcommand{\tletbox}{\ensuremath{\texttt{letbox}}}

\newcommand{\ze}{\textsf{zero}}
\newcommand{\tsucc}{\textsf{succ}}
\newcommand{\su}[1]{\tsucc~#1}
\newcommand{\tapp}{\textsf{app}}
\newcommand{\trec}{\textsf{rec}}
\newcommand\recn[4]{\trec_{#1}\ #2\ (#3)\ #4}
\newcommand{\Nf}{\textsf{Nf}}
\newcommand{\Ne}{\textsf{Ne}}

\newcommand{\Typ}{\textsf{Typ}}
\newcommand{\Exp}{\textsf{Exp}}
\newcommand{\CExp}{\textsf{CExp}}

\newcommand{\intp}[1]{\ensuremath{\llbracket #1 \rrbracket}}
\newcommand{\glu}[1]{\ensuremath{\llparenthesis #1 \rrparenthesis}}

\newcommand{\boxit}[1]{\tbox\ #1}

\newcommand{\letbox}[3]{\tletbox\ {#1} \shortleftarrow #2\ \texttt{in}\ #3}

\newcommand{\tisapp}{\textsf{isapp?}\xspace}
\newcommand{\isapp}[3]{\ensuremath{\tisapp\ #1 \ \textsf{then}\ #2\
    \textsf{else}\ #3}}
\newcommand{\tmatc}{\textsf{match}\xspace}
\newcommand{\matc}[1]{\ensuremath{\textsf{match}\ #1\ \textsf{with}\xspace}}
\newcommand{\tnfbranch}{\textsf{nfbranch}\xspace}

\newcommand{\sep}{\;|\;}
\newcommand{\func}{\longrightarrow}
\newcommand\hfunc{{\hat{\func}}}

\newcommand{\nbe}{\textsf{nbe}}

\newcommand{\inser}{\textsf{insert}}

\newcommand\id{\ensuremath{\textsf{id}}}

\newcommand{\Se}{\texttt{Ty}}

\newcommand{\Ty}[1]{\ensuremath{\Se_{#1}}}

\newcommand\WC{\ensuremath{\mathpzc{W}}\xspace}

\newcommand\SetC{\ensuremath{\mathpzc{Set}}\xspace}

\newcommand\El{\textbf{El}}

\newcommand{\tru}{\textsf{true}\xspace}
\newcommand{\fal}{\textsf{false}\xspace}
\newcommand{\branch}{b}
\newcommand{\nbranch}{r}
\newcommand{\var}[1]{\textsf{var}_{#1}}

\newcommand\tand{\text{ and }}
\newcommand\byIH{\tag{by IH}}

\newcommand\red[1]{{\color{red} #1}}

\newcommand{\vertrule}[1][1ex]{\rule{.45pt}{#1}}

\newcommand{\sttstile}{\ensuremath{\mathrel{\raisebox{0.01pt}{\vertrule[1.6ex]}{\hspace{.1em}\vDash}}}}

\DeclareDocumentCommand{\dSemjudgebf}{ o o m } {
  \IfNoValueTF {#1}
  {\IfNoValueTF {#2}
    {\Psi; \Gamma \sttstile #3}
    {\Psi; #2 \sttstile #3}}
  {\IfNoValueTF {#2}
    {#1; \Gamma \sttstile #3}
    {#1; #2 \sttstile #3}}
}

\DeclareDocumentCommand{\judge}{ o m } {
  \IfNoValueTF {#1}
  {\Gamma \vdash #2}
  {#1 \vdash #2}
}

\DeclareDocumentCommand{\gequiv}{ o m m } {
  \IfNoValueTF {#1}
  {L \vdash #2 \approx #3}
  {#1 \vdash #2 \approx #3}
}

\DeclareDocumentCommand{\cont}{ o m } {
  \square (\judge[#1]{#2})
}

\DeclareDocumentCommand{\rjudge}{ o m } {
  \IfNoValueTF {#1}
  {\vGamma \vdash_r #2}
  {#1 \vdash_r #2}
}

\DeclareDocumentCommand{\semjudge}{ o m } {
  \IfNoValueTF {#1}
  {\Gamma \vDash #2}
  {#1 \vDash #2}
}

\DeclareDocumentCommand{\Semjudge}{ o m } {
  \IfNoValueTF {#1}
  {\Gamma \Vdash #2}
  {#1 \Vdash #2}
}

\DeclareDocumentCommand{\mjudge}{ o m } {
  \IfNoValueTF {#1}
  {\vGamma \vdash #2}
  {#1 \vdash #2}
}

\DeclareDocumentCommand{\djudge}{ o o m } {
  \IfNoValueTF {#1}
  {\IfNoValueTF {#2}
    {\Psi; \Gamma \vdash #3}
    {\Psi; #2 \vdash #3}}
  {\IfNoValueTF {#2}
    {#1; \Gamma \vdash #3}
    {#1; #2 \vdash #3}}
}

\DeclareDocumentCommand{\ljudge}{ o m m } {
  \IfNoValueTF {#1}
  {\Psi \vdash_{#2} #3}
  {#1 \vdash_{#2} #3}
}

\DeclareDocumentCommand{\ljudgel}{ o m m } {
  \IfNoValueTF {#1}
  {\Gamma \vdash_{#2} #3}
  {#1 \vdash_{#2} #3}
}

\DeclareDocumentCommand{\DTyp}{o m m} {
  \ensuremath{(\ljudgel[#1]{#2}{\Ty{#3}})}
}

\DeclareDocumentCommand{\DTrm}{o m m m} {
  \ensuremath{(\ljudgel[#1]{#2}{#3 : \Ty{#4}})}
}

\DeclareDocumentCommand{\CTyp}{o m} {
  \ensuremath{\square \DTyp[#1]{0}{#2}}
}

\DeclareDocumentCommand{\CTrm}{o m m} {
  \ensuremath{\square \DTrm[#1]{0}{#2}{#3}}
}

\DeclareDocumentCommand{\TPI}{m o m m m}{
  \ensuremath{(#1 : \DTyp[#2]{h}{#3}) \STo^{#4} #5}
}

\DeclareDocumentCommand{\pjudge}{ o o m } {
  \IfNoValueTF {#1}
  {\IfNoValueTF {#2}
    {L \sep \Psi \vdash #3}
    {#2 \sep \Psi \vdash #3}}
  {\IfNoValueTF {#2}
    {L \sep #1 \vdash #3}
    {#2 \sep #1 \vdash #3}}
}

\DeclareDocumentCommand{\ptyping}{o o m m } {
  \pjudge[#1][#2]{#3 : #4}
}

\DeclareDocumentCommand{\ptyequiv}{o o m m m } {
  \pjudge[#1][#2]{#3 \approx #4 : #5}
}

\DeclareDocumentCommand{\lpjudge}{ o o m m } {
  \pjudge[#1][#2]{_{#3}#4}
}

\DeclareDocumentCommand{\lpequiv}{ o o m m m } {
  \lpjudge[#1][#2]{#3}{#4 \approx #5}
}

\DeclareDocumentCommand{\tjudge}{ o o o m } {
  \IfNoValueTF {#1}
  {\IfNoValueTF {#2}
    {\IfNoValueTF {#3}
      {L \sep \Psi; \Gamma \vdash #4}
      {#3 \sep \Psi; \Gamma \vdash #4}}
    {\IfNoValueTF {#3}
      {L \sep \Psi; #2 \vdash #4}
      {#3 \sep \Psi; #2 \vdash #4}}}
  {\IfNoValueTF {#2}
    {\IfNoValueTF {#3}
      {L \sep #1; \Gamma \vdash #4}
      {#3 \sep #1; \Gamma \vdash #4}}
    {\IfNoValueTF {#3}
      {L \sep #1; #2 \vdash #4}
      {#3 \sep #1; #2 \vdash #4}}}
}

\DeclareDocumentCommand{\ltjudge}{o o o m m} {
  \tjudge[#1][#2][#3]{_{#4} #5}
}

\DeclareDocumentCommand{\lttypwf}{o o o m m m} {
  \ltjudge[#1][#2][#3]{#4}{#5 : \Ty{#6}}
}

\DeclareDocumentCommand{\lttyping}{o o o m m m m} {
  \ltjudge[#1][#2][#3]{#4}{#5 : #6 : \Ty{#7}}
}

\DeclareDocumentCommand{\lttypingv}{o o o m m m m} {
  \ltjudge[#1][#2][#3]{#4}{^{\#} #5 : #6 : \Ty{#7}}
}

\DeclareDocumentCommand{\lttypingd}{o o o m m m m} {
  \ltjudge[#1][#2][#3]{#4}{#5 : \Ty{#6}\;(: \Ty{#7})}
}

\DeclareDocumentCommand{\ltsubst}{o o o m m m} {
  \ltjudge[#1][#2][#3]{#4}{#5 : #6}
}

\DeclareDocumentCommand{\ltsubstv}{o o o m m m} {
  \ltjudge[#1][#2][#3]{#4}{^{\#} #5 : #6}
}

\DeclareDocumentCommand{\ltsubeq}{o o o m m m m} {
  \ltjudge[#1][#2][#3]{#4}{#5 \approx #6 : #7}
}

\DeclareDocumentCommand{\lttypeq}{o o o m m m m} {
  \ltjudge[#1][#2][#3]{#4}{#5 \approx #6 : \Ty{#7}}
}

\DeclareDocumentCommand{\lttyequiv}{o o o m m m m m} {
  \ltjudge[#1][#2][#3]{#4}{#5 \approx #6 : #7 : \Ty{#8}}
}

\DeclareDocumentCommand{\lttyequivv}{o o o m m m m m} {
  \ltjudge[#1][#2][#3]{#4}{^{\#} #5 \approx #6 : #7 : \Ty{#8}}
}

\DeclareDocumentCommand{\lttyequivd}{o o o m m m m m} {
  \ltjudge[#1][#2][#3]{#4}{#5 \approx #6 : \Ty{#7}\;(: \Ty{#8})}
}

\DeclareDocumentCommand{\lsemjudge}{ o m m } {
  \IfNoValueTF {#1}
  {\Psi \vDash_{#2} #3}
  {#1 \vDash_{#2} #3}
}

\DeclareDocumentCommand{\lsemvjudge}{ o m m } {
  \lsemjudge[#1]{#2}{^v #3}
}

\DeclareDocumentCommand{\lmjudge}{ o o m m } {
  \djudge[#1][#2]{_{#3} #4}
}

\DeclareDocumentCommand{\dsemjudge}{ o o m } {
  \IfNoValueTF {#1}
  {\IfNoValueTF {#2}
    {\Psi; \Gamma \vDash #3}
    {\Psi; #2 \vDash #3}}
  {\IfNoValueTF {#2}
    {#1; \Gamma \vDash #3}
    {#1; #2 \vDash #3}}
}

\DeclareDocumentCommand{\dSemjudge}{ o o m } {
  \IfNoValueTF {#1}
  {\IfNoValueTF {#2}
    {\Psi; \Gamma \Vdash #3}
    {\Psi; #2 \Vdash #3}}
  {\IfNoValueTF {#2}
    {#1; \Gamma \Vdash #3}
    {#1; #2 \Vdash #3}}
}

\DeclareDocumentCommand{\wmjudge}{ o m } {
  \IfNoValueTF {#1}
  {\vGamma \vdash_{\! w} #2}
  {#1 \vdash_{\! w} #2}
}

\DeclareDocumentCommand{\msemjudge}{ o m } {
  \IfNoValueTF {#1}
  {\vGamma \vDash #2}
  {#1 \vDash #2}
}

\DeclareDocumentCommand{\mSemjudge}{ o m } {
  \IfNoValueTF {#1}
  {\vGamma \Vdash #2}
  {#1 \Vdash #2}
}

\DeclareDocumentCommand{\typing}{ o m m } {
  \judge[#1]{#2 : #3}
}

\DeclareDocumentCommand{\rtyping}{ o m m } {
  \rjudge[#1]{#2 : #3}
}

\DeclareDocumentCommand{\semtyp}{ o m m } {
  \semjudge[#1]{#2 : #3}
}

\DeclareDocumentCommand{\semvtyp}{ o m m } {
  \IfNoValueTF {#1}
  {\semtyp[\Psi]{^v #2}{#3}}
  {\semtyp[#1]{^v #2}{#3}}
}

\DeclareDocumentCommand{\Semtyp}{ o m m } {
  \Semjudge[#1]{#2 : #3}
}

\DeclareDocumentCommand{\tyequiv}{ o m m m } {
  \judge[#1]{#2 \approx #3 : #4}
}

\DeclareDocumentCommand{\semtyeq}{ o m m m } {
  \semjudge[#1]{#2 \approx #3 : #4}
}

\DeclareDocumentCommand{\semvtyeq}{ o m m m } {
  \IfNoValueTF {#1}
  {\semjudge[\Psi]{^v #2 \approx #3 : #4}}
  {\semjudge[#1]{^v #2 \approx #3 : #4}}
}

\DeclareDocumentCommand{\lSemtypPrime}{ o o m m m } {
  \dSemjudgebf[#1][#2]{_{#3} #4 : #5}
}

\DeclareDocumentCommand{\mtyping}{ o m m } {
  \mjudge[#1]{#2 : #3}
}

\DeclareDocumentCommand{\lmtyping}{ o m m m } {
  \mjudge[#1]{_{#2} #3 : #4}
}

\DeclareDocumentCommand{\dtyping}{ o o m m } {
  \djudge[#1][#2]{#3 : #4}
}

\DeclareDocumentCommand{\ltyping}{ o o m m m } {
  \djudge[#1][#2]{_{#3} #4 : #5}
}

\DeclareDocumentCommand{\wmtyping}{ o m m } {
  \wmjudge[#1]{#2 : #3}
}

\DeclareDocumentCommand{\mSemtyp}{ o m m } {
  \mSemjudge[#1]{#2 : #3}
}

\DeclareDocumentCommand{\lsemtyp}{ o o m m m } {
  \dsemjudge[#1][#2]{_{#3} #4 : #5}
}

\DeclareDocumentCommand{\lsemvtyp}{ o o m m m } {
  \dsemjudge[#1][#2]{^v_{#3} #4 : #5}
}

\DeclareDocumentCommand{\lSemtyp}{ o o m m m } {
  \dSemjudge[#1][#2]{_{#3} #4 : #5}
}

\DeclareDocumentCommand{\mrarr}{ o m m } {
  \mjudge[#1]{#2 > #3}
}

\DeclareDocumentCommand{\wmrarr}{ o m m } {
  \wmjudge[#1]{#2 > #3}
}

\DeclareDocumentCommand{\mtyequiv}{ o m m m } {
  \mjudge[#1]{#2 \approx #3 : #4}
}

\DeclareDocumentCommand{\dtyequiv}{ o o m m m } {
  \djudge[#1][#2]{#3 \approx #4 : #5}
}

\DeclareDocumentCommand{\ltyequiv}{ o o m m m m } {
  \djudge[#1][#2]{_{#3} #4 \approx #5 : #6}
}

\DeclareDocumentCommand{\ltygneeq}{ o o m m m } {
  \djudge[#1][#2]{ #3 \sim #4 : #5}
}

\DeclareDocumentCommand{\ltygteq}{ o o m m m } {
  \djudge[#1][#2]{ #3 \simeq #4 : #5}
}

\DeclareDocumentCommand{\ltyred}{ o o m m m } {
  \djudge[#1][#2]{ #3 \rightsquigarrow #4 : #5}
}

\DeclareDocumentCommand{\ltyreds}{ o o m m m } {
  \djudge[#1][#2]{ #3 \rightsquigarrow^\ast #4 : #5}
}

\DeclareDocumentCommand{\dtconv}{ o o m m m } {
  \djudge[#1][#2]{ #3 \ \hat{\Longleftrightarrow}\ #4 : #5}
}

\DeclareDocumentCommand{\dtconvnf}{ o o m m m } {
  \djudge[#1][#2]{ #3 \Longleftrightarrow #4 : #5}
}

\DeclareDocumentCommand{\dtconvne}{ o o m m m } {
  \djudge[#1][#2]{ #3 \longleftrightarrow #4 : #5}
}

\DeclareDocumentCommand{\lsemtyp}{ o o m m m } {
  \dsemjudge[#1][#2]{_{#3} #4 : #5}
}

\DeclareDocumentCommand{\dsemtyeq}{ o o m m m } {
  \dsemjudge[#1][#2]{#3 \approx #4 : #5}
}

\DeclareDocumentCommand{\lsemtyeq}{ o o m m m m } {
  \dsemjudge[#1][#2]{_{#3} #4 \approx #5 : #6}
}

\DeclareDocumentCommand{\lsemvtyeq}{ o o m m m m } {
  \dsemjudge[#1][#2]{^v_{#3} #4 \approx #5 : #6}
}

\DeclareDocumentCommand{\msemtyeq}{ o m m m } {
  \msemjudge[#1]{#2 \approx #3 : #4}
}
\DeclareDocumentCommand{\msemtyp}{ o m m } {
  \msemjudge[#1]{#2 : #3}
}

\DeclareDocumentCommand{\mgluty}{ o m o m } {
  \IfNoValueTF {#3}
  {\mjudge[#1]{#2 \; \circledR \; #4}}
  {\mjudge[#1]{#2 \; \circledR_{#3} \; #4}}
}

\DeclareDocumentCommand{\mglutm}{ o m m m o m } {
   \IfNoValueTF {#5}
   {\mjudge[#1]{#2 : #3 \; \circledR \; #4 \in \El(#6)}}
   {\mjudge[#1]{#2 : #3 \; \circledR_{#5} \; #4 \in \El_{#5}(#6)}}
 }

 \DeclareDocumentCommand{\mglunat}{ o m m } {
   \mjudge[#1]{#2 : N \; \circledR \; #3 \in Nat}
 }

\DeclareDocumentCommand{\mglutms}{ o m m m } {
  \mjudge[#1]{#2 : #3 \; \circledR \; #4}
}

\DeclareDocumentCommand{\mglutyu}{ o m o m } {
  \IfNoValueTF {#3}
  {\mjudge[#1]{#2 \; \overline{\circledR} \; #4}}
  {\mjudge[#1]{#2 \; \overline{\circledR}_{#3} \; #4}}
}

\DeclareDocumentCommand{\mglutmu}{ o m m m o m } {
  \IfNoValueTF {#5}
  {\mjudge[#1]{#2 : #3 \; \overline{\circledR} \; #4 \in \El(#6)}}
  {\mjudge[#1]{#2 : #3 \; \overline{\circledR}_{#5} \; #4 \in \El_{#5}(#6)}}
}

\DeclareDocumentCommand{\mglutmd}{ o m m m o m } {
  \IfNoValueTF {#5}
  {\mjudge[#1]{#2 : #3 \; \underline{\circledR} \; #4 \in \El(#6)}}
  {\mjudge[#1]{#2 : #3 \; \underline{\circledR}_{#5} \; #4 \in \El_{#5}(#6)}}
}

\newcommand{\labeledit}[1]{\label{#1}}

\newcommand{\mhighlight}[1]{\colorbox{light-gray}{\ensuremath{#1}}}

\newcommand{\iscore}[1]{\ensuremath{#1\ \texttt{core}}}
\newcommand{\istype}[1]{\ensuremath{#1\ \texttt{type}}}
\newcommand{\wf}[2]{\ensuremath{#2\ \texttt{wf}^{#1}}}

\newcommand{\JH}[1]{}

\AtBeginDocument{%
  \providecommand\BibTeX{{%
    \normalfont B\kern-0.5em{\scshape i\kern-0.25em b}\kern-0.8em\TeX}}}

\begin{document}

\title{Layered Modal Type Theories}         


\author{Jason Z. S. Hu}
\email{zhong.s.hu@mail.mcgill.ca}
\author{Brigitte Pientka}
\email{bpientka@cs.mcgill.ca}
\affiliation{%
  \department{School of Computer Science}
  \institution{McGill University}
  \streetaddress{McConnell Engineering Bldg. , 3480 University St.}
  \city{Montr\'eal}
  \state{Qu\'ebec}
  \country{Canada}
  \postcode{H3A 0E9}
}

\begin{abstract}
  We introduce layers to modal type theories, which subsequently enables type theories
  for pattern matching on code in meta-programming and clean and straightforward
  semantics.
\end{abstract}




\maketitle

\section{Introduction}

Under Curry-Howard correspondence, \citet{davies_modal_2001} discovered that the modal
logic $S4$ corresponds to staged or meta-programming. %
In their settings, what fundamentally characterizes meta-programming in $S4$ is the
reading of $\square A$ as the code of type $A$ and the following two axioms:
\begin{align*}
  T &: \square A \to A \\
  4 &: \square A \to \square \square A
\end{align*}
In these two axioms, $T$ means that we can extract a term of $A$ from its
representation. %
In other words, $T$ runs a meta-program and generates a program which evaluates to a
value of $A$. %
The Axiom $4$, on the other hand, allows the use of a macro inside of a macro of a
macro of a macro ... until an arbitrary depth. %
Though \citet{davies_modal_2001} have made a convincing argument that $S4$ does
provide a logical foundation for meta-programming, compared to practice, the facility
provided by $S4$ is hardly enough. %
In particular, tactics in proof assistants frequently refer to the current goal,
analyze the structure of the goal and apply various strategies to fill it in. %
However, $S4$ in itself only allows composition of meta-programs. It does not
provide a clear direction of how one should add support for case analyses of code
structures as we would usually expect from a reasonable tactic language in a proof
assistant. %
The most recent breakthrough in this direction is due to \citet{Jang:POPL22}. %
They propose Moebius, which extend System F with modal types so that it supports the
analysis or pattern matching on code. %
The crucial ideas of their work are:
\begin{enumerate}
\item contextual types due to \citet{nanevski_contextual_2008}, which generalize
  \citet{davies_modal_2001} and support representation of an open piece of code
  relative to a certain local context, and
  
\item levels, which ultimately are a technicality forced by the fact that pattern
  matching on code must be able to pattern match on the code that again does pattern
  matching.
\end{enumerate}
Though this work has provided ideas and shed light on how we should support pattern
matching on code, there are also problems:
\begin{enumerate}
\item Moebius, despite having preservation, does not have progress. %
  The crucial issue here is that Moebius does not provide guarantees on the covering
  of pattern matching on code. %
  It is possible for a pattern matching to find no matching case for a piece of code,
  in which case the evaluation gets stuck. %
  
\item Moebius does not have a normalization proof. %
  It is hard to say whether any well-typed program in Moebius terminates or not even
  if there is a fix for the covering problem. 
\end{enumerate}
These shortcomings are usually not a big deal if Moebius is used as a programming language. %
However, in type theory, we have to be more cautious, as our system must be
foundationally consistent. %
In this technical report, we develop another group of modal type theories based on one
simple idea of \emph{layers}. %
In this idea, a modal type theory is stratified into layers (primarily into two but
possibly more) but terms share one common syntax. %
The modalities are then responsible for distinguishing terms from different layers. %
It turns out that having layers is one powerful idea that not only enables a simpler
way to do pattern matching on code, but also finds a clear correspondence in the
semantics, which seems to work around the limitations found by
\citet{kavvos_intensionality_2021}.

In this technical report, we start with the simplest layered modal type theory that
can only be used to meta-program closed code. %
Then we work our way up to support open code, and next pattern matching. %
We are also interested in scaling the whole setup to dependent types, but we leave the
work in another technical report due to the size of the development. 

\section{Layered Dual-context Formulation}\label{sec:st}

In this section, we develop the layered modal type theory $S4$ in the dual-context
style with simple types. %
We first introduce the original formulation of $S4$
by~\citet{davies_modal_2001,pfenning_judgmental_2001} and compare it with our layered
formulation. %
It turns out that the original dual-context $S4$ is just a limit of our layered
version. %
As we will see, when we restrict the type theory to two layers, we obtain very clean
and intuitive semantics for the type theory. %

\subsection{Original Dual-context Formulation}

The original dual-context formulation of $S4$ is pioneered
by~\citet{davies_modal_2001,pfenning_judgmental_2001}, which not only gives an
intuitionistic formulation of the modal logic $S4$, but also shows that $S4$
corresponds to meta-programming under Curry-Howard correspondence. %
The core idea of the dual-context formulation is simple: there are two kinds of facts
in the system that are tracked by two different contexts. %
One kind stands for truth or local facts. %
The other kind is global facts. %
Effectively, that $\square A$ is true is equivalent to that $A$ is globally true. %
The type theory has the following syntax:
\begin{alignat*}{2}
  S, T &:=&&\ \Nat \sep \square T \sep S \func T
  \tag{Types, \Typ} \\
  x, y & && \tag{Local variables} \\
  u & && \tag{Global variables} \\
  s, t &:=&&\ x \sep u \tag{Terms, $\Exp$} \\
  & && \sep \ze \sep \su t
  \tag{natural numbers}\\
  & && \sep \boxit t \sep \letbox u s t
  \tag{box}\\
  & &&\sep \lambda x. t \sep s\ t \tag{functions}  \\
  \Gamma, \Delta &:= &&\ \cdot \sep \Gamma, x : T
  \tag{Local contexts}\\
  \Phi, \Psi &:= &&\ \cdot \sep \Phi, u : T
  \tag{Global contexts}
\end{alignat*}
We use natural numbers as our base type. %
We could have introduced a recursor for $\Nat$ but omit it to avoid clutter. %
Since there are two kinds of facts, we have two kinds of contexts and variables. %
$\boxit t$ introduces a $\square$ type while $\letbox u s t$ eliminates one. %
Their typing rules are:
\begin{mathpar}
  \inferrule
  {\dtyping[\Psi][\cdot]t T}
  {\dtyping{\boxit t}{\square T}}

  \inferrule
  {\dtyping{s}{\square T} \\ \dtyping[\Psi, u : T]{t}{T'}}
  {\dtyping{\letbox u s t} T'}
\end{mathpar}
According tho the introduction rule, $\boxit t$ can only depends on global facts, so
we can be sure that it is also a global fact. %
In the elimination rule, an additional global assumption of $T$ is added to the global
context when typing the body $t$. %
We can easily see that these two rules are coherent in the sense of local soundness:
\begin{mathpar}
  \inferrule
  {\dtyping[\Psi][\cdot]s T \\ \dtyping[\Psi, u : T]{t}{T'}}
  {\dtyping{\letbox u {\boxit s} t} T}

  \Rightarrow

  \dtyping{t[s/u]} T
\end{mathpar}
The right hand side holds due to a theorem of global substitutions.

\subsection{Adding Layers}

In this section, we add layers to the previous modal type theory. %
The ultimate motivation for adding layers comes from the desire of drawing more
distinctions to code and computation. %
In the original dual-context formulation, a judgment $\dtyping t T$ does not provide
information about whether we are treating $t$ as a piece of code or not. %
This distinction is made in programmers' mind, when they place $t$ inside or outside
of $\tbox$. %
However, this existing distinction is not made in the judgment, and we perceive that
this lack of information is the critical reason of why supporting pattern matching on
code is so difficult and cumbersome. %
By adding layers, we add a subscript $i \in [0, n]$ for a fixed natural number $n$
to the typing judgment so that the nested layers of $\square$ are accounted for. 
\begin{mathpar}
  \inferrule
  {\ltyping[\Psi][\cdot] i t T}
  {\ltyping{1 + i}{\boxit t}{\square T}}

  \inferrule
  {\ltyping i {s}{\square T} \\ \ltyping[\Psi, u : T] i {t}{T'}}
  {\ltyping i {\letbox u s t} T'}
\end{mathpar}
It is immediate to see that to recover the original $S4$, we simply take $n$ to
$\omega$ so that arbitrarily nested $\square$s are allowed. %
In this sense, the original formulation becomes a special case of the layered
version. %
In this technical report, we focus on the case of $n = 1$, in which typing of code and
computation is naturally distinguished, and subsequently we find that supporting
pattern matching on code becomes natural and the semantics of the type theory is
amazingly intuitive.  %
Under Curry-Howard correspondence, we intend to make layer $0$ the layer for code, and
layer $1$ the one for computation. %

When we take $n = 1$, the introduction rule is only possible when $i = 0$, namely
\begin{mathpar}
  \inferrule
  {\ltyping[\Psi][\cdot] 0 t T}
  {\ltyping{1}{\boxit t}{\square T}}\red{(?)}
\end{mathpar}
Since $\boxit t$ can only live at layer $1$, it is not possible to construct a
$\square T$ at layer $0$ without using the global assumptions. %
To actual ban $\square$ entirely at layer $0$, we must maintain an invariant:
\begin{mathpar}
  \inferrule
  { }
  {\iscore \Nat}

  \inferrule
  {\iscore S \\ \iscore T}
  {\iscore{S \func T}}
\end{mathpar}
$\iscore T$ holds whenever $T$ contains no $\square$. %
We require that all types from a valid global context satisfy this judgment by
generalizing this judgment to $\iscore \Psi$. %
Semantically, $\iscore T$ describes the types that belong to some core type theory
(simply typed $\lambda$-calculus in this case) where programming occurs, while
$\square$ is an add-on to extend this core type theory with
meta-programming.  %
Similarly, we require that all types from a valid local context at layer $1$ must have
at max one layer of $\square$. %
This is checked by the following judgment:
\begin{mathpar}
  \inferrule
  { }
  {\istype \Nat}

  \inferrule
  {\istype S \\ \istype T}
  {\istype{S \func T}}

  \inferrule
  {\iscore T}
  {\istype{\square T}}
\end{mathpar}
This judgment ensures that a meta-program never describes another meta-program,
i.e. meta-meta-programs don't exist. %
This might appear as a restriction at the first glance, but, lo and behold, most
widely used meta-programming systems have the same limitation and they have the
additional disadvantage of having a different meta-programming language\footnote{These
  systems include macro systems of practical languages like C and Rust, and tactic
  languages like Ltac, Ltac2 of Coq, and Eisbach of Isabelle}. %
Not to mention some of these systems are not even typed. %
Yet, they are still widely used and are pretty practical. %
We perceive this limitation as a blessing: practicality comes from simplicity, not
generality. %
Clearly, $\istype T$ subsumes $\iscore T$.
\begin{lemma}
  If $\iscore T$, then $\istype T$.
\end{lemma}
\begin{proof}
  Induction.
\end{proof}
The following theorem is part of the syntactic validity of the typing judgment:
\begin{theorem}\labeledit{thm:st:ctx-wf}
  If $\ltyping i t T$, then $\iscore \Psi$ and
  \begin{itemize}
  \item if $i = 0$, then $\iscore \Gamma$;
  \item if $i = 1$, then $\istype \Gamma$.
  \end{itemize}
\end{theorem}
Moreover, the layers are associated with the validity of types:
\begin{theorem}\labeledit{thm:st:typ-wf} $ $
  \begin{itemize}
  \item If $\ltyping 0 t T$, then $\iscore T$.
  \item If $\ltyping 1 t T$, then $\istype T$. 
  \end{itemize}
\end{theorem}

With these two intended theorems in mind, the rules become:
\begin{mathpar}
  \inferrule
  {\istype \Gamma \\ \ltyping[\Psi][\cdot] 0 t T}
  {\ltyping{1}{\boxit t}{\square T}}

  \inferrule
  {\ltyping 1 {s}{\square T} \\ \ltyping[\Psi, u : T] 1 {t}{T'}}
  {\ltyping 1 {\letbox u s t} T'}
\end{mathpar}
\Cref{thm:st:ctx-wf,thm:st:typ-wf} checks out in these two rules by seeing that the
induction hypotheses directly apply. %
Notice that the elimination rule can only occur at layer $1$. %
This is forced by the fact that $\square T$ is only possible at layer $1$. %
Eliminating a $\square$ at layer $0$ is automatically vacuous, so we immediately
exclude that from our rules. %
We finish the typing rules for other terms for completeness:
\begin{mathpar}
  \inferrule
  {\iscore \Psi \\ \iscore \Gamma}
  {\ltyping{0}{\ze}{\Nat}}

  \inferrule
  {\iscore \Psi \\ \istype \Gamma}
  {\ltyping{1}{\ze}{\Nat}}

  \inferrule
  {\ltyping{i}{t}{\Nat}}
  {\ltyping{i}{\su t}{\Nat}}

  \inferrule
  {\iscore \Psi \\ \iscore \Gamma \\ u : T \in \Psi}
  {\ltyping{0}{u}{T}}

  \inferrule
  {\iscore \Psi \\ \istype \Gamma \\ u : T \in \Psi}
  {\ltyping{1}{u}{T}}

  \inferrule
  {\iscore \Psi \\ \iscore \Gamma \\ x : T \in \Gamma}
  {\ltyping{0}{x}{T}}

  \inferrule
  {\iscore \Psi \\ \istype \Gamma \\ x : T \in \Gamma}
  {\ltyping{1}{x}{T}}
  \\
  
  \inferrule
  {\ltyping[\Psi][\Gamma, x : S]{i}{t}{T}}
  {\ltyping{i}{\lambda x. t}{S \func T}}

  \inferrule
  {\ltyping{i}{t}{S \func T} \\ \ltyping{i}{s}{S}}
  {\ltyping{i}{t\ s}{T}}
\end{mathpar}
\Cref{thm:st:ctx-wf,thm:st:typ-wf} are checked by doing induction on the typing
judgment. %
At last, layer $0$ is subsumed by layer $1$:
\begin{lemma}[Lifting]
  If $\ltyping 0 t T$, then $\ltyping 1 t T$. 
\end{lemma}
\begin{proof}
  Induction. For the case of local variables, apply subsumption.
\end{proof}

\subsection{Programming in Layered Formulation}

Before discussing more about the layered $S4$, let us see some programs written in
this system. %
It is possible to implement Axioms $K$ and $T$:
\begin{mathpar}
  \inferrule*
  {\inferrule*
    {\inferrule*
      {\inferrule*
        {\ltyping[\Psi, u : A \func B, u' : A][\Gamma, f : \square (A \func B), x :
          \square A] 0 {u}{A \func B}
          \\
          \ltyping[\Psi, u : A \func B, u' : A][\Gamma, f : \square (A \func B), x : \square A] 0 {u'}{A}}
      {\ltyping[\Psi, u : A \func B, u' : A][\Gamma, f : \square (A \func B), x : \square A] 0 {u\ u'}{B}}}
      {\ltyping[\Psi, u : A \func B, u' : A][\Gamma, f : \square (A \func B), x : \square A] 1 {\boxit{(u\ u')}}{\square B}}}
    {\ltyping[\Psi][\Gamma, f : \square (A \func B), x : \square A] 1 {\letbox u f {\letbox {u'} x {\boxit{(u\ u')}}}}{\square B}}}
  {\ltyping 1 {\lambda f\ x. \letbox u f {\letbox {u'} x {\boxit{(u\ u')}}}}{\square (A \func B) \func \square A \func \square B}}
\end{mathpar}
\begin{mathpar}
  \inferrule*
  {\inferrule*
    {\ltyping[\Psi, u : A][\Gamma, x : \square A] 1 {u}{A}}
    {\ltyping[\Psi][\Gamma, x : \square A] 1 {\letbox u x u}{A}}}
  {\ltyping 1 {\lambda x. \letbox u x u}{\square A \func A}}
\end{mathpar}
We omit the validity checking of contexts and types without loss of generality. %

Nevertheless, though we claim this system is a layered $S4$, it does not admit the
Axiom $4$:
\begin{align*}
  \square A \func \square \square A
\end{align*}
In fact, the type is not even valid. %
$\istype{\square \square A}$ requires $\iscore{\square A}$, which is not possible. %
This is intended because our goal is to exclude meta-meta-programs. %
If we let $n > 1$, then it is possible to express this axiom. %
Nevertheless, it is still possible to write reasonable meta-programs. %
The following composes two pieces of code of type $\Nat$ with addition (if we assume
the system supports it):
\begin{align*}
  \ltyping 1 {\lambda x\ y. \letbox u x {\letbox{u'}y{\boxit{(u + u')}}}}{\square \Nat
      \func \square \Nat \func \square \Nat}
\end{align*}
Let this function be $f$. Then we can have:
\begin{align*}
  \ltyping 1 {f\ (\boxit{(3 + 2)})\ (\boxit{(1 + 4)})}{\square \Nat}
\end{align*}
Executing this program generates the code $\boxit{((3 + 2) + (1 + 4))}$. %
We will see soon that this execution has already reached a normal form because code is
distinguished by its syntactic structure, so this code is different from $\boxit{10}$,
for example.

\subsection{Equivalence of Layered Formulation}

In this section, we consider the equivalence of our layered formulation. %
The equivalence relation is also layered. %
In principle, the $\beta$ and $\eta$ equivalence rules only apply for layer $1$,
i.e. the computation layer. %
Other rules like the PER rules and the congruence rules apply for both layers.

The PER rules:
\begin{mathpar}
  \inferrule*
  {\ltyequiv i s t T}
  {\ltyequiv i t s T}

  \inferrule*
  {\ltyequiv i s t T \\ \ltyequiv i t u T}
  {\ltyequiv i s u T}
\end{mathpar}

Congruence rules:
\begin{mathpar}
  \inferrule
  {\iscore \Psi \\ \iscore \Gamma}
  {\ltyequiv{0}{\ze}{\ze}{\Nat}}

  \inferrule
  {\iscore \Psi \\ \istype \Gamma}
  {\ltyequiv{1}{\ze}{\ze}{\Nat}}

  \inferrule
  {\ltyequiv{i}{t}{t'}{\Nat}}
  {\ltyequiv{i}{\su t}{\su t'}{\Nat}}
  
  \inferrule
  {\iscore \Psi \\ \iscore \Gamma \\ u : T \in \Psi}
  {\ltyequiv{0}{u}{u}{T}}

  \inferrule
  {\iscore \Psi \\ \istype \Gamma \\ u : T \in \Psi}
  {\ltyequiv{1}{u}{u}{T}}
  
  \inferrule*
  {\iscore \Psi \\ \iscore \Gamma \\  x : T \in \Gamma}
  {\ltyequiv 0 x x T}

  \inferrule*
  {\iscore \Psi \\ \istype \Gamma \\  x : T \in \Gamma}
  {\ltyequiv 1 x x T}
  \\
  
  \inferrule*
  {\ltyequiv[\Psi][\Gamma, x : S]{i}{t}{t'}{T}}
  {\ltyequiv{i}{\lambda x. t}{\lambda x. t'}{S \func T}}

  \inferrule*
  {\ltyequiv{i}{t}{t'}{S \func T} \\ \ltyequiv{i}{s}{s'}{S}}
  {\ltyequiv{i}{t\ s}{t'\ s'}{T}}

  \inferrule*
  {\iscore \Psi \\ \ltyequiv[\Psi][\cdot]{0}{t}{t'}T}
  {\ltyequiv{1}{\boxit t}{\boxit t'}{\square T}}

  \inferrule*
  {\ltyequiv 1 {s}{s'}{\square T} \\ \ltyequiv[\Psi, u : T] 1 {t}{t'}{T'}}
  {\ltyequiv 1 {\letbox u s t}{\letbox u{s'}{t'}} T'}
\end{mathpar}

$\beta$ equivalence:
\begin{mathpar}
  \inferrule*
  {\ltyping[\Psi][\Gamma, x : S] 1 t T \\ \ltyping 1 s S}
  {\ltyequiv 1{(\lambda x. t)\ s}{t[s/x]}{T}}

  \inferrule*
  {\ltyping[\Psi][\cdot] 0 s T \\ \ltyping[\Psi, u : T] 1 {t}{T'}}
  {\ltyequiv 1 {\letbox u {\boxit s} t}{t[s/u]}T}
\end{mathpar}

$\eta$ equivalence:
\begin{mathpar}
  \inferrule*
  {\ltyping 1 {t}{S \func T}}
  {\ltyequiv 1 t {\lambda x. (t\ x)}{S \func T}}
\end{mathpar}

Due to lack of dynamics at layer $0$, equivalence in fact is just syntactic equality:
\begin{lemma}
  If $\ltyequiv 0 t s T$, then $t = s$. 
\end{lemma}
\begin{proof}
  Induction, and notice that $\beta$ and $\eta$ equivalences are excluded from layer
  $0$.
\end{proof}
This lemma ensures that terms at layer $0$ are indeed treated as code, as they are
distinguished by their syntactic structures. 

The $\beta$ equivalence rules depend on the folklore substitution lemma:
\begin{theorem} $ $
  \begin{itemize}
  \item If $\ltyping[\Psi][\Gamma, x : S, \Gamma'] 1 t T$ and $\ltyping 1 s S$, then
    $\ltyping[\Psi][\Gamma, \Gamma']1{t[s/x]}T$. 
  \item If $\ltyping[\Psi, u : S, \Psi'] i t T$ and $\ltyping[\Psi][\cdot] 0 s S$, then
    $\ltyping[\Psi, \Psi'] i {t[s/u]}T$. 
  \end{itemize}
\end{theorem}
Notice that the case for local substitutions only occurs at layer $1$ while the case
for global substitutions occurs at both layers. %
This is because local substitutions do not propagate under $\tbox$ so the $\beta$ rule
for function will never affect layer $0$. %
Contrarily, global substitutions do go under $\tbox$, so we must consider all
layers. %
\begin{proof}
  Induction. In the case of global substitutions, notice layer $0$ is subsumed by
  layer $1$. 
\end{proof}

Presupposition then is a consequence of substitution lemma:
\begin{lemma}[Presupposition]\labeledit{lem:st:presup}
  If $\ltyequiv i{t}{t'}T$, then $\ltyping i t T$ and $\ltyping i {t'}T$.
\end{lemma}
\begin{proof}
  Induction on $\ltyequiv i{t}{t'}T$.
\end{proof}

\subsection{Simultaneous Substitutions}

In preparation for the semantic development, we first develop the theory of
simultaneous substitutions for the dual-context formulation. %
Though this section directly targets our layered formulation, nothing here cannot be
generalized to the original formulation. %
The development in this section will directly contribute to a deepened understanding
of the presheaf model developed later in this technical report.

Since there are two contexts in the dual-context style, we need two different kinds of
substitutions to handle those between global contexts and those between local
contexts. %
First, we consider the substitutions between global contexts. %
They are very similar to the substitutions in simply typed $\lambda$-calculus:
\begin{align*}
  \sigma := \cdot \sep \sigma, t/u \tag{Global substitutions}
\end{align*}
\begin{mathpar}
  \inferrule*
  {\iscore \Psi}
  {\typing[\Psi]{\cdot}{\cdot}}

  \inferrule*
  {\typing[\Psi]{\sigma}{\Phi} \\ \ltyping[\Psi][\cdot] 0 {t}{T}}
  {\typing[\Psi]{\sigma, t/u}{\Phi, u : T}}
\end{mathpar}
Notice that in the step case, the term $t$ is typed at layer $0$. %
This is because a global substitution is meant to store a list of codes of right
types. %
The equivalence between global substitutions is simply the congruence rules:
\begin{mathpar}
  \inferrule*
  {\iscore \Psi}
  {\tyequiv[\Psi]{\cdot}{\cdot}{\cdot}}

  \inferrule*
  {\tyequiv[\Psi]{\sigma}{\sigma'}{\Phi} \\ \ltyequiv[\Psi][\cdot] 0 {t}{t'}{T}}
  {\tyequiv[\Psi]{\sigma, t/u}{\sigma', t'/u}{\Phi, u : T}}
\end{mathpar}
The step case is equivalent to say $t = t'$ and that $t$ is well-typed. 

The local substitutions are similar. %
We quickly list out the definitions for completeness:
\begin{align*}
  \delta := \cdot \sep \delta, t/x \tag{Local substitutions}
\end{align*}
\begin{mathpar}
  \inferrule*
  {\iscore \Psi \\ \istype \Gamma}
  {\dtyping{\cdot}{\cdot}}

  \inferrule*
  {\dtyping{\delta}{\Delta} \\ \ltyping 1 {t}{T}}
  {\dtyping{\delta, t/x}{\Delta, x : T}}
\end{mathpar}
Notice that in the step case, the term $t$ is typed at layer $1$.
\begin{mathpar}
  \inferrule*
  {\iscore \Psi \\ \istype \Gamma}
  {\dtyequiv{\cdot}{\cdot}{\cdot}}

  \inferrule*
  {\dtyequiv{\delta}{\delta'}{\Delta} \\ \ltyequiv 1 {t}{t'}{T}}
  {\dtyequiv{\delta, t/x}{\delta', t'/x}{\Delta, x : T}}
\end{mathpar}
Since the term equivalence in the step case is at layer $1$, it is possible to have
$\beta$ and $\eta$ equivalence occurring in a local substitution. %
A simultaneous substitution between dual-contexts is simply a pair of global and local
substitutions:
\begin{mathpar}
  \inferrule
  {\typing[\Psi]{\sigma}{\Phi} \\ \dtyping{\delta}{\Delta}}
  {\dtyping{\sigma; \delta}{\Phi; \Delta}}

  \inferrule
  {\tyequiv[\Psi]{\sigma}{\sigma'}{\Phi} \\ \dtyequiv{\delta}{\delta'}{\Delta}}
  {\dtyequiv{\sigma; \delta}{\sigma', \delta'}{\Phi; \Delta}}
\end{mathpar}

To develop a substitution calculus, we first define the application of a global
substitution:
\begin{align*}
  x[\sigma] &:= x \\
  u[\sigma] &:= \sigma(u) \tag{lookup $u$ in $\sigma$} \\
  \ze[\sigma] &:= \ze \\
  \su t [\sigma] &:= \su{(t[\sigma])} \\
  \boxit t [\sigma] &:= \boxit{(t[\sigma])} \\
  \letbox u s t [\sigma] &:= \letbox u {s[\sigma]} {(t[\sigma, u/u])} \\
  \lambda x. t [\sigma] &:= \lambda x. (t[\sigma]) \\
  t\ s [\sigma] &:= (t[\sigma]) \ (s[\sigma])
\end{align*}
The following substitution lemma holds:
\begin{theorem}[Global substitution]
  If $\ltyping[\Phi] i t T$ and $\typing[\Psi] \sigma \Phi$, then $\ltyping i{t[\sigma]} T$.
\end{theorem}
\begin{proof}
  We do induction on $\ltyping[\Phi] i t T$ and only consider some important cases:
  \begin{itemize}[label=Case]
  \item $t = u$
    \begin{mathpar}
      \inferrule
      {\iscore \Phi \\ \istype \Gamma \\ u : T \in \Phi}
      {\ltyping[\Phi]{1}{u}{T}}
    \end{mathpar}
    \begin{align*}
      & \ltyping[\Psi][\cdot] 0 {\sigma(u)} T \tag{due to $u : T \in \Phi$} \\
      & \ltyping i {\sigma(u)} T \tag{weakening and subsumption}
    \end{align*}

  \item $t = \boxit t'$
    \begin{mathpar}
      \inferrule
      {\istype \Gamma \\ \ltyping[\Phi][\cdot] 0{t'} T}
      {\ltyping[\Phi]{1}{\boxit t'}{\square T}}
    \end{mathpar}
    \begin{align*}
      & \ltyping[\Psi][\cdot] 0{t'[\sigma]} T
        \byIH \\
      & \ltyping{1}{\boxit t'[\sigma]}{\square T}
    \end{align*}
    
  \item $t = \letbox u s t'$
    \begin{mathpar}
      \inferrule
      {\ltyping[\Phi] 1 {s}{\square T'} \\ \ltyping[\Phi, u : T'] 1 {t'}{T}}
      {\ltyping[\Phi] 1 {\letbox u s t'} T}
    \end{mathpar}
    \begin{align*}
      & \ltyping 1 {s[\sigma]}{\square T'}
        \byIH \\
      & \ltyping[\Psi, u : T'] 1 {t'[\sigma, u/u]}{T}
      \tag{by IH and weakening} \\
      & \ltyping{1}{\letbox u s t'[\sigma]}{T}
    \end{align*}
  \end{itemize}
\end{proof}

Applying a local substitution follows the same principle:
\begin{align*}
  x[\delta] &:= \delta(x) \tag{lookup $x$ in $\delta$} \\
  u[\delta] &:= u  \\
  \ze[\delta] &:= \ze \\
  \su t [\delta] &:= \su{(t[\delta])} \\
  \boxit t [\delta] &:= \boxit t \\
  \letbox u s t [\delta] &:= \letbox u {s[\delta]} {(t[\delta])} \\
  \lambda x. t [\delta] &:= \lambda x. (t[\delta, x/x]) \\
  t\ s [\delta] &:= (t[\delta]) \ (s[\delta])
\end{align*}

The local substitution lemma is only concerned about layer $1$:
\begin{theorem}[Local substitution]
  If $\ltyping[\Psi][\Delta] 1 t T$ and $\dtyping \delta \Delta$, then $\ltyping 1{t[\delta]} T$.
\end{theorem}
\begin{proof}
  We do induction on $\ltyping[\Psi][\Delta] 1 t T$. %
  The fact that we only operate at layer $1$ has eliminated several cases. %
  We only work on a few selected cases.
  \begin{itemize}[label=Case]
  \item $t = x$
    \begin{mathpar}
      \inferrule
      {\iscore \Psi \\ \istype \Delta \\ x : T \in \Delta}
      {\ltyping[\Psi][\Delta]{1}{x}{T}}
    \end{mathpar}
    \begin{align*}
      & \ltyping 1 {\delta(x)} T \tag{due to $x : T \in \Delta$} 
    \end{align*}

    \item $t = \boxit t'$
    \begin{mathpar}
      \inferrule
      {\istype \Delta \\ \ltyping[\Psi][\cdot] 0{t'} T}
      {\ltyping[\Psi][\Delta]{1}{\boxit t}{\square T}}
    \end{mathpar}
    \begin{align*}
      & \istype \Gamma \tag{by $\dtyping \delta \Delta$} \\
      & \ltyping{1}{\boxit t'}{\square T}
    \end{align*}
    
  \item $t = \letbox u s t'$
    \begin{mathpar}
      \inferrule
      {\ltyping[\Psi][\Delta] 1 {s}{\square T'} \\ \ltyping[\Psi, u : T'][\Delta] 1 {t'}{T}}
      {\ltyping[\Psi][\Delta] 1 {\letbox u s t'} T}
    \end{mathpar}
    \begin{align*}
      & \ltyping 1 {s[\delta]}{\square T'}
        \byIH \\
      & \ltyping[\Psi, u : T'] 1 {t'[\delta]}{T}
      \tag{by weakening of $\delta$ and IH} \\
      & \ltyping{1}{\letbox u s t'[\delta]}{T}
    \end{align*}
  \end{itemize}
\end{proof}

Then the application of a simultaneous substitution is just a chain of applying two
substitutions one after the other:
\begin{align*}
  t [\sigma; \delta] := t[\sigma][\delta]
\end{align*}
\begin{theorem}[Simultaneous substitution]
  If $\ltyping[\Phi][\Delta] 1 t T$ and $\dtyping{\sigma;\delta}{\Phi;\Delta}$, then
  $\ltyping 1 {t[\sigma;\delta]} T$.
\end{theorem}
Simultaneous substitutions are only applicable at layer $1$. %
At layer $0$, only we are only concerned about global substitutions. 
\begin{proof}
  Use the global substitution lemma and then the local one. 
\end{proof}

We can then generalize the application of global substitutions to local substitutions:
\begin{theorem}[Global substitution]
  If $\dtyping[\Phi] \delta \Delta$ and $\typing[\Psi] \sigma \Phi$, then
  $\dtyping{\delta[\sigma]} \Delta$.
\end{theorem}

This generalization is necessary to define the composition of two simultaneous substitutions:
\begin{align*}
  (\sigma'; \delta') \circ (\sigma; \delta) := \sigma' \circ \sigma; \delta' [\sigma]
  \circ \delta
\end{align*}
where the composition of global and local substitutions is just defined naturally.
\begin{theorem}
  If $\dtyping[\Psi'][\Gamma']{\sigma'; \delta'}{\Psi''; \Gamma''}$
  and $\dtyping{\sigma; \delta}{\Psi'; \Gamma'}$, then
  $\dtyping{(\sigma'; \delta') \circ (\sigma; \delta)}{\Psi''; \Gamma''}$
\end{theorem}
\begin{proof}
  By inversion, we have the following
  \begin{mathpar}
    \typing[\Psi']{\sigma'}{\Psi''}

    \dtyping[\Psi'][\Gamma']{\delta'}{\Gamma''}

    \typing[\Psi]{\sigma}{\Psi'}

    \dtyping{\delta}{\Gamma'}
  \end{mathpar}

  First it is clear that $\typing[\Psi]{\sigma' \circ \sigma'}{\Psi''}$ holds. %
  By the generalized global substitution lemma, we also have
  $\dtyping[\Psi][\Gamma']{\delta'[\sigma]}{\Gamma''}$ and hence
  $\dtyping{\delta'[\sigma] \circ \delta}{\Gamma''}$. %
  That justifies our definition. 
\end{proof}

With the easy definition of identity, we see that there are two categories to organize
for the dual-context formulation. %
First we have the category of global contexts and global substitutions. %
In this category, terms in global substitutions are identified by their syntactic
structures. %
The second category is the category of dual-contexts and simultaneous substitutions. %
Local substitutions in simultaneous substitutions do admit dynamics and it is similar
to typical type theories like the simply typed $\lambda$-calculus. %
In our presheaf model, we take their special subcategories of weakenings as bases and
use $\square$ to travel between two bases. %
Terms in the first category is interpret to programs in the second, so that dynamics
of syntactic structures are recovered. %

\section{A Presheaf Model and Normalization by Evaluation}\labeledit{sec:presheaf}

In this section, we follow the ideas we built at the end of the last section and build
a presheaf model for our 2-layered dual-context modal type theory. %
From this presheaf model, we extract a normalization by
evaluation~\citep{martin-lof_intuitionistic_1975,berger_inverse_1991,altenkirch_categorical_1995,abel_normalization_2013}
algorithm that is sound and complete w.r.t. the typing and equivalence judgments. %
Normalization by evaluation (or NbE for short) is a technique to establish the
normalization property of a type theory. %
An NbE proof usually is composed of two steps: first, we evaluate the terms of a type
theory in some mathematical domain which has a notion of computation; second, we
extract from the values in this mathematical domain normal forms in the original type
theory. %
Compared to the more traditional approach of reducibility candidates, An NbE proof has
the advantage of not having to establish the confluence property of certain reduction
relation. %
The confluence property of the type theory is ultimately piggybacked on the confluence
property of that chosen domain. %

An NbE proof using a presheaf model has become a standard technique in the community
to quickly and robustly establish the normalization property. %
A presheaf category maps from some chosen base category to the category of sets. %
By carefully choosing this base category, a straightforward interpretation from terms
to the presheaf category corresponds to a normalization algorithm. %
In this section, this base category is the category of weakenings. %
The presheaf model showed in this section is a moderate extension of the classic
presheaf model of simply typed $\lambda$-calculus~\citep{altenkirch_categorical_1995}.

\subsection{Weakenings}

As previously mentioned, we will use the category of weakenings as the base
category. %
Similar to substitutions, we need to have two different kinds of weakenings. %
In spite of that, weakenings are actually much simpler than substitutions, because we
only need to handle shifting of positions of variables in the contexts. %
First, we define the weakenings of global contexts:
\begin{align*}
  \gamma := \varepsilon \sep q(\gamma) \sep p(\gamma)
  \tag{Global weakenings}
\end{align*}
\begin{mathpar}
  \inferrule
  { }
  {\varepsilon: \cdot \To_g \cdot}

  \inferrule
  {\gamma : \Psi \To_g \Phi \\ \iscore T}
  {q(\gamma) : \Psi, u : T \To_g \Phi, u : T}

  \inferrule
  {\gamma : \Psi \To_g \Phi \\ \iscore T}
  {p(\gamma) : \Psi, u : T \To_g \Phi}
\end{mathpar}
This definition is identical to the one for simply typed $\lambda$-calculus. %
We let $\id_\Psi : \Psi \To_g \Psi$ be the identity global weakening. %
We omit the subscript whenever it is clear from the context. %
In fact, local substitutions are defined in the same way:
\begin{align*}
  \tau := \varepsilon \sep q(\tau) \sep p(\tau)
  \tag{Local weakenings}
\end{align*}
\begin{mathpar}
  \inferrule
  {\iscore \Psi}
  {\varepsilon: \Psi;\cdot \To_l \cdot}

  \inferrule
  {\tau : \Psi; \Gamma \To_l \Delta \\ \istype T}
  {q(\tau) : \Psi; \Gamma, x : T \To_l \Delta, x : T}

  \inferrule
  {\tau : \Psi; \Gamma \To_l \Delta \\ \istype T}
  {p(\tau) : \Psi; \Gamma, x : T \To_l \Delta}
\end{mathpar}
We overload $\id_{\Psi; \Gamma} : \Psi; \Gamma \To_l \Gamma$ be the identity local
weakening. %
Again, we omit the subscript when it can be inferred from the context. %
Both global and local weakenings have composition and form categories. %
Then the weakenings of dual-contexts are defined as pairs of global and local
weakenings:
\begin{mathpar}
  \inferrule
  {\gamma : \Psi \To_g \Phi \\ \tau : \Psi; \Gamma \To_l \Delta}
  {\gamma;\tau : \Psi; \Gamma \To \Phi; \Delta}
\end{mathpar}
Composition of weakenings is simply defined as composition of global and local
substitutions, due to the following lemma:
\begin{lemma}
  If $\tau : \Psi; \Gamma \To_l \Delta$ and $\iscore \Phi$, then $\tau : \Phi; \Gamma \To_l \Delta$.
\end{lemma}
Our base category is set to be the category of dual-contexts and weakenings:
\begin{theorem}
  Dual-contexts and weakenings form a category.
\end{theorem}
We call this category \WC.

\subsection{Presheaf Model}

In this section, we construct the presheaf model. %
Since the presheaf model effectively constructs a normalization algorithm, we first
define what it means to be normal and neutral:
\begin{alignat*}{2}
  t^c &:=&&\ x \sep u \sep \ze \sep \su t^c \sep \lambda x. t^c \sep t^c\ {t^c}'
            \tag{Core terms ($\CExp$)} \\
  w &:= &&\ v \sep \ze \sep \su w \sep \boxit t^c \sep \lambda x. w
  \tag{Normal form ($\Nf$)} \\
  v &:= &&\ x \sep u \sep v\ w \sep \letbox u v w \tag{Neutral form
    ($\Ne$)}
\end{alignat*}
In order to completely capture all normal forms by using grammar only, we introduce
\emph{core terms} which effectively remove $\tbox$ and $\tletbox$ from the terms. %
In other words, $t^c$ only represents programs that live at layer $0$. %
Moreover, our definition of normal forms indicates that we are not concerned about
extensional equivalence of $\square$, including commuting conversions. %
For example, the following is a well possible normal form:
\begin{align*}
  \lambda x. (\letbox u v {\lambda y. u\ y}) x
\end{align*}
even though in principle it can be further simplified to:
\begin{align*}
  \lambda x. \letbox u v {u\ x}
\end{align*}
This longer normal form is due to the lack of following commuting conversions:
\begin{align*}
  \letbox u s {\lambda x. t} \approx \lambda x. \letbox u s t
\end{align*}
Namely in our system, $\tletbox$ and $\lambda$ do not commute. %
In general, such commuting conversions can be introduced for sum types, i.e. those
types that use pattern matching instead of projections for elimination. %
We can replace $\tletbox$ for an \texttt{if} expression for example to obtain a
possible commuting conversion for booleans. %
In this report, we are not concerned about this because eventually our system will be
scaled to dependent types, where extensional equivalence is usually not considered
at all.

Furthermore, we let
\begin{align*}
  \Nf^T_{\;\Psi; \Gamma} &:= \{ w \sep \ltyping 1 w T \} \\
  \Ne^T_{\;\Psi; \Gamma} &:= \{ v \sep \ltyping 1 v T \}
\end{align*}
Notice that the sets only capture judgments at layer $1$, because at layer $0$, terms
do not compute so all well-typed terms are automatically normal. %
$\Nf^T$ and $\Ne^T$ are the mapping from dual-contexts to the sets of normal and
neutral forms, respectively. %
Evidently $\Nf^T$ and $\Ne^T$ are presheaves from the category of weakenings because
both normal and neutral forms respects weakenings. %


To give the interpretation of types, we need exponential objects in the presheaf
category. %
Given two presheaves $F$ and $G$, we form a presheaf exponential $F \hfunc G$ by
\begin{align*}
  F \hfunc G &: \WC^{op} \To \SetC \\
  (F \hfunc G)_{\;\Psi;\Gamma} &:= \forall\ \Phi;\Delta \To \Psi;\Gamma. F_{\;\Phi;\Delta} \to G_{\;\Phi;\Delta}
\end{align*}

We interpret types as follows:
\begin{align*}
  \intp{\_} &: \Typ \to \WC^{op} \To \SetC \\
  \intp{\Nat} &:= \Nf^\Nat \\
  \intp{\square T} &:= \Nf^{\square T} \\
  \intp{S \func T} &:= \intp{S} \hfunc \intp{T}
\end{align*}
Notice that the interpretation of $\square T$ is not even recursive. %
We directly interpret it as the presheaf of normal forms. %
According to the definition of $\Nf^{\square T}$, there are two possibilities:
\begin{enumerate}
\item either the term is already neutral, which there is nothing more we can learn, or
\item it is of the form $\boxit t^c$.
\end{enumerate}
In the latter case, we know that $t^c$ must be well-typed at layer $0$ as a closed
term, which gives us a ticket to its raw syntactic structure. %
This is more obvious when we interpret the terms to natural transformations. %
This interpretation does return a presheaf because all subcases return presheaves:
\begin{theorem} $ $
  \begin{itemize}
  \item If $\iscore T$, then $\intp{T}$ is a functor.
  \item If $\istype T$, then $\intp{T}$ is a functor.
  \end{itemize}
\end{theorem}
If $a \in \intp{T}_{\;\Phi; \Delta}$ and $\gamma; \tau : \Psi; \Gamma \To \Phi; \Delta$,
we write $a[\gamma;\tau]$ for the functorial action of $\gamma;\tau$ on $a$.

Next we give semantics for dual-contexts. %
The interpretation proceeds in two steps. %
In the first step, we interpret global contexts in the category of global contexts and
global weakenings:
\begin{align*}
  \intp{\Phi}^0_\Psi &:= \{ \gamma \sep \gamma : \Psi \To_g \Phi \} \\
  \intp{\Phi}^1_\Psi &:= \{ \sigma \sep \typing[\Psi]{\sigma}{\Phi}\}
\end{align*}
The interpretation of global contexts is layered because in the interpretation of
terms, we need to interpret global contexts differently at different layers to ensure
the interpretation of terms does form a natural transformation. %
This treatment will become clear when we consider terms. %
At layer $0$, $\Phi$ is mapped to its Hom set, while at layer $1$, $\Phi$ is
interpreted as a presheaf of global substitutions. %
We also use $\sigma$ to range over $\intp{\Phi}^i_{\Psi}$ when $i$ is unknown. %
Then we interpret the local contexts:
\begin{align*}
  \intp{\cdot}_{\;\Psi;\Gamma} &:= \{ * \} \\
  \intp{\Delta, x : T}_{\;\Psi;\Gamma} &:= \intp{\Delta}_{\;\Psi;\Gamma}
                                              \times \intp{T}_{\;\Psi;\Gamma}
\end{align*}
where $\{ * \}$ denotes a fixed singleton set. %
We let $\rho$ to range over $\intp{\Delta}_{\;\Psi;\Gamma}$. %
The interpretation of local contexts is not layered because computation only occurs at
layer $1$. %
The interpretation of dual-contexts is also layered due to the layered interpretation
of global contexts:
\begin{align*}
  \intp{\Phi;\Delta}^i_{\;\Psi;\Gamma} &:= \intp{\Phi}^i_\Psi \times
                                       \intp{\Delta}_{\;\Psi;\Gamma} 
\end{align*}
Functoriality of previous two interpretations are immediate:
\begin{theorem}
  If $\iscore \Phi$, then $\intp{\Phi}^i$ is a functor for both $i$.
\end{theorem}
\begin{theorem}
  If $\istype \Delta$, then $\intp{\Delta}$ is a functor.
\end{theorem}
\begin{theorem}
  If $\iscore \Phi$ and $\istype \Delta$, then $\intp{\Phi;\Delta}^i$ is a functor.
\end{theorem}
\begin{proof}
  The functorial action is defined as follows: given $\gamma; \tau : \Psi'; \Gamma'
  \To \Psi ; \Gamma$, we have
  \begin{align*}
    (\gamma'; \rho  : \intp{\Phi; \Delta}^0_{\;\Psi;\Gamma})[\gamma;\tau]
    &:= \gamma' \circ \gamma; \rho [\gamma; \tau] \\
    (\sigma; \rho  : \intp{\Phi; \Delta}^1_{\;\Psi;\Gamma})[\gamma;\tau]
    &:= \sigma[\gamma]; \rho [\gamma; \tau]
  \end{align*}
  Notice that regardless of which layer, the local weakening does not have effect on
  the interpretation of the global context. 
\end{proof}

Last we interpret terms to natural transformations between presheaves. %
We need two auxiliary natural transformations which map $\Ne^T$ to $\intp{T}$ and
$\intp{T}$ to $\Nf^T$, respectively:
\begin{align*}
  \downarrow^T &: \intp{T} \To \Nf^T \\
  \downarrow^\Nat_{\;\Psi;\Gamma}(a) &:= a \\
  \downarrow^{\square T}_{\;\Psi;\Gamma}(a)
               &:= a \\
  \downarrow^{S \func T}_{\;\Psi;\Gamma}(a)
               &:= \lambda x. \downarrow^T_{\;\Psi;\Gamma, x : S}(a~(\id; p(\id)~,~{\uparrow^S_{\;\Psi;\Gamma, x : S}\!(x)}))
                 \tag{where $\id; p(\id) : \Psi; \Gamma, x{:} S \To \Psi; \Gamma$}\\[5pt]
  \uparrow^T &: \Ne\ T \To \intp{T} \\
  \uparrow^B_{\;\Psi;\Gamma}(v) &:= v \\
  \uparrow^{\square T}_{\;\Psi;\Gamma}(v)
               &:= v \\
  \uparrow^{S \func T}_{\;\Psi; \Gamma}(v)
               &:= (\gamma;\tau : \Phi;\Delta \To \Psi;\Gamma)(a)
                 \mapsto \uparrow^T_{\;\Phi;\Delta}(v[\gamma;\tau]\ \downarrow^S_{\;\Phi;\Delta}(a))
\end{align*}
In the following, we interpret terms as natural transformations:
\begin{align*}
  \intp{\_}
  &: \forall i \to \ltyping[\Phi][\Delta] i t T \to \intp{\Phi;\Delta}^i \To \intp{T} \\
  \intp{t}^i_{\;\Psi;\Gamma}
  &: \intp{\Phi;\Delta}_{\;\Psi;\Gamma} \to \intp{T}_{\;\Psi;\Gamma}
  \tag{expanded form}\\
  \intp{\ze}^i_{\;\Psi;\Gamma}(\sigma; \rho)
  &:= \ze \\  
  \intp{\su t}^i_{\;\Psi;\Gamma}(\sigma; \rho)
  &:= \su (\intp{t}^i_{\;\Psi;\Gamma}(\sigma; \rho)) \\
  \intp{u}^0_{\;\Psi;\Gamma}(\gamma; \rho)
  &:= \uparrow^{T}_{\;\Psi; \Gamma}(u[\gamma]) \\
  \intp{u}^1_{\;\Psi;\Gamma}(\sigma; \rho)
  &:= \intp{\sigma(u)}^0_{\;\Psi; \Gamma}(\id; *) \\
  \intp{x}^i_{\;\Psi;\Gamma}(\sigma; \rho)
  &:= \rho(x) \tag{lookup $x$ in $\rho$} \\
  \intp{\boxit t}^1_{\;\Psi;\Gamma}(\sigma; \rho)
  &:= \boxit {(t[\sigma])} \\
  \intp{\letbox u s t}^1_{\;\Psi;\Gamma}(\sigma; \rho)
  &:= \intp{t}^1_{\;\Psi;\Gamma}(\sigma, t^c/u; \rho)
    \tag{if $\intp{s}^1_{\;\Psi;\Gamma}(\sigma; \rho) = \boxit {t^c}$} \\
  \intp{\letbox u s t}^1_{\;\Psi;\Gamma}(\sigma; \rho)
  &:= \uparrow^{T}_{\;\Psi;\Gamma}(\letbox u v {\downarrow^T_{\;\Psi, u: S;\Gamma}(\intp{t}^1_{\;\Psi, u:
      S;\Gamma}(\sigma', u/u; \rho'))})
    \tag{if $\intp{s}^1_{\;\Psi;\Gamma}(\sigma; \rho) = v : \square S$ and $p(\id);
      \id : \Psi, u: S;\Gamma \To \Psi;\Gamma$ and $(\sigma'; \rho') := (\sigma; \rho)[p(\id); \id]$} \\
  \intp{\lambda x : S. t}^i_{\;\Psi;\Gamma}(\sigma;\rho)
            &:= (\gamma;\tau : \Phi;\Delta \To \Psi;\Gamma)(a :
              \intp{S}_{\;\Phi;\Delta}) \mapsto
              \intp{t}^i_{\;\Phi;\Delta} (\sigma'; (\rho', a))
              \tag{where $(\sigma'; \rho') := \sigma;\rho[\gamma; \tau]$}  \\
  \intp{t\ s}^i_{\;\Psi;\Gamma}(\sigma;\rho) &:= \intp{t}^i_{\;\Psi;\Gamma} (\sigma;\rho, \id_{\;\Psi;\Gamma}~,~
                                     \intp{s}^i_{\;\Psi;\Gamma}(\sigma;\rho))    
\end{align*}
This interpretation extends the interpretation of simply typed $\lambda$-calculus. %
The interpretation is also layered because the type theory itself is layered. %
All cases that are not modal are identical to the simply typed $\lambda$-calculus
regardless of layers. %
For example, in the variable case, we use the variable $x$ to look up the current
local environment $\rho$ to find a $\intp{T}_{\;\Psi;\Gamma}$. %
What are more interesting are the modal cases. %
The $\boxit t$ case is only available at layer $1$, because $\tbox$ is only well-typed
at this layer. %
Since a $\tbox$'ed term should be identified with its syntactic structure, we directly
apply the global substitution $\typing[\Psi]\sigma\Phi$ and obtain a closed term
$t[\sigma]$. %
In the $\tletbox$ case, we must first evaluate $s$. %
Since
$\intp{s}^1_{\;\Psi;\Gamma} \in \intp{\square S}^1_{\;\Psi;\Gamma} =
\Nf^{\square S}_{\;\Psi;\Gamma}$, there are two possible outcomes of this evaluation: it is either a
$\boxit{t^c}$, or a neutral form. %
In the first case, we have obtained a piece of code, so we can just extend $t^c$ for $u$
to the global substitution. %
In this case, we can employ other interpretations. %
For example, it is possible to inspect the syntactic structure of $t^c$. %
If we have the following form in our syntax:
\begin{align*}
  \isapp s t t'
\end{align*}
which inspects the structure of $s$. %
If $s$ is a function application, then we branch to $t$, otherwise $t'$. %
In this interpretation, we can clearly interpret this code:
\begin{align*}
  \intp{\isapp s t t'}^1_{\;\Psi;\Gamma}(\sigma; \rho)
  &:= \intp{t}^1_{\;\Psi;\Gamma}(\sigma; \rho)
    \tag{if $\intp{s}^1_{\;\Psi;\Gamma}(\sigma; \rho) = \boxit{(s^c\ t^c)}$} \\
  \intp{\isapp s t t'}^1_{\;\Psi;\Gamma}(\sigma; \rho)
  &:= \intp{t'}^1_{\;\Psi;\Gamma}(\sigma; \rho)
    \tag{if $\intp{s}^1_{\;\Psi;\Gamma}(\sigma; \rho) = \boxit{t^c}$ but $t^c$ is not
    an application}
\end{align*}
The elimination form of a $\square$ can be more complex, e.g. case analysis on the code
structure. %
However, though the semantic side is ready for a more complex interpretation, the
syntactic side is not rich enough to express the same level of complexity. %
We will need a more powerful core theory, which exactly leads us to an adaptation to
dependent type theory. %
We will develop this theory once we finish the development of simple types. 

In the second case of $\tletbox$, on the other hand, we are blocked by some neutral
term, so we can only continue the evaluation of $t$ with $u$ as is, i.e. by applying
the $p(\id)$ weakening to the global substitution and then append the result with $u/u$. %
We need the overall term to be some $\intp{T}_{\;\Psi;\Gamma}$, so we first reify the
evaluation of $t$ into a normal form, and then reflect the neutral back to the
abstract data in $\intp{T}_{\;\Psi;\Gamma}$. %

Finally, we consider the interpretation of $u$. %
This case is the most interesting one and is the essence of this layered
interpretation. %
It is possible to refer to $u$ at both layers. %
When we refer to $u$ at layer $1$, this means that we are actually bringing a term at
layer $0$ to layer $1$, so we must bring the dynamics of that term back alive. %
This is why we interpret the lookup of $u$ in $\sigma$ at layer $0$. %
The layer goes from $1$ to $0$ ensures the well-foundedness of the interpretation
regardless of what $\sigma(u)$ is. %
The evaluation environment of the recursive interpretation is independent of the input
one. %
If we refer to $u$ at layer $0$, that means $u$ remains neutral even after the
evaluation in the core type theory. %
In this case, we want to eventually return something representing a neutral term by
calling reflection. %
In addition, we must ensure the naturality of the interpretation. %
Combining both requirements, we come to the conclusion of using global weakenings only
for the interpretation of global contexts at layer $0$. %
$u[\gamma]$ still returns a global variable. %
This interpretation maintains monotonocity, which eventually
leads to naturality.

Next, we define the identity environment $\intp{\Gamma}_{\;\Psi;\Gamma}$ as 
\begin{align*}
  \uparrow^{\Gamma} &:  \intp{\Gamma}^1_{\;\Psi;\Gamma} \\
  \uparrow^{\cdot} &:= * \\
  \uparrow^{\Gamma, x : T} &:= (\uparrow^{\Psi; \Gamma} [\id; p(\id)], \uparrow^{T}_{\Psi; \Gamma, x : T}\!\!(x))
\end{align*}
Then we let
\begin{align*}
  \uparrow^{\Psi; \Gamma} &:  \intp{\Psi; \Gamma}^1_{\;\Psi;\Gamma} \\
  \uparrow^{\Psi; \Gamma} &:= \id; \uparrow^\Gamma
\end{align*}

This allows us to finally write down the normalization algorithm:
\begin{definition}
  A normalization by evaluation algorithm given $\ltyping 1 t T$ is
  \begin{align*}
    \nbe^T_{\;\Psi;\Gamma}(t) &:= \downarrow^T_{\;\Psi;\Gamma} (\intp{t}^1_{\;\Psi;\Gamma}(\uparrow^{\Psi;\Gamma}))
  \end{align*}
\end{definition}

$\nbe^T_{\;\Psi;\Gamma}(t)$ returns a $\Nf^T_{\;\Psi;\Gamma}$. %
The correctness of this algorithm is formulated by the following two target theorems:
\begin{theorem}[Completeness]\labeledit{thm:st:compl}
  If $\ltyequiv 1 t {t'} T$, then $\nbe^T_{\;\Psi;\Gamma}(t) = \nbe^T_{\;\Psi;\Gamma}(t')$.
\end{theorem}
\begin{theorem}[Soundness]\labeledit{thm:st:sound}
  If $\ltyping 1 t T$, then $\ltyequiv 1 t {\nbe^T_{\;\Psi;\Gamma}(t)} T$.
\end{theorem}
The completeness theorem states that if two terms are equivalent, then they have equal
normal forms. %
The soundness theorem states that a well-typed term is equivalent to its normal
form. %
These two theorems are only concerned about layer $1$, because, again, this is the
only layer where computation actually occurs. %
Starting the next section, we will develop properties of the interpretations and
establish the target theorems. 

\subsection{Properties of Interpretations}

To establish the completeness and soundness theorems, we need the following important
lemmas:
\begin{lemma}[Naturality]\labeledit{lem:st:t-intp-nat}
  If $\ltyping[\Phi][\Delta] i t T$, $\gamma;\tau : \Psi';\Gamma' \To \Psi;\Gamma$ and
  $\sigma;\rho \in \intp{\Phi;\Delta}^i_{\;\Psi;\Gamma}$, then
  $\intp{t}^i_{\;\Psi;\Gamma}(\sigma;\rho)[\gamma;\tau] = \intp{t}^i_{\;\Psi';\Gamma'}((\sigma;\rho)[\gamma;\tau])$.
\end{lemma}
That is, the order of applying a weakening does not change the result.

We also need two lemmas that describe how the interpretation of terms interacts with
global and local substitutions:
\begin{lemma}[Global substitutions]\labeledit{lem:st:glob-subst}
  $\intp{t}^i_{\;\Psi;\Gamma}(\sigma \circ \sigma';\rho) = \intp{t[\sigma]}^i_{\;\Psi;\Gamma}(\sigma';\rho)$
\end{lemma}
\begin{lemma}[Local substitutions]\labeledit{lem:st:loc-subst}
  Assume $x$ is the topmost local variable, \\
  then
  $\intp{t[s/x]}^i_{\;\Psi;\Gamma}(\sigma;\rho) =
  \intp{t}^i_{\;\Psi;\Gamma}(\sigma;(\rho, \intp{s}^i_{\;\Psi;\Gamma}(\sigma;\rho)))$
\end{lemma}

Let us first tackle the naturality of the interpretation. %
The naturality of the interpretation relies on the naturality of both reification and
reflection.
\begin{lemma}
  Both $\uparrow^T$ and $\downarrow^T$ are natural.
\end{lemma}
\begin{proof}
  We perform induction on $T$. %
  We only need to check the function case because the other two cases are identity. %
  Let $T = S \func T'$. %
  First assume $\gamma;\tau : \Phi; \Delta \To \Psi; \Gamma$ and $a \in \intp{S \func
    T'}_{\Psi; \Gamma}$, then we need to prove
  \begin{align*}
    \downarrow^{S \func T'}_{\;\Psi;\Gamma}(a)[\gamma;\tau]
    = \downarrow^{S \func T'}_{\;\Phi;\Delta}(a[\gamma;\tau])
  \end{align*}
  We compute the left hand side:
  \begin{align*}
    \downarrow^{S \func T'}_{\;\Psi;\Gamma}(a)[\gamma;\tau]
    &= (\lambda x. \downarrow^{T'}_{\;\Psi;\Gamma, x : S}(a(\id;
      p(\id)~,~{\uparrow^S_{\;\Psi;\Gamma, x : S}\!(x)})))[\gamma; \tau] \\
    &= \lambda x. \downarrow^{T'}_{\;\Psi;\Gamma, x : S}(a(\id;
      p(\id)~,~{\uparrow^S_{\;\Psi;\Gamma, x : S}\!(x)}))[\gamma; q(\tau)] \\
    &= \lambda x. \downarrow^{T'}_{\;\Psi;\Gamma, x : S}(a(\id;
      p(\id)~,~{\uparrow^S_{\;\Psi;\Gamma, x : S}\!(x)})[\gamma; q(\tau)])
      \byIH \\
    &= \lambda x. \downarrow^{T'}_{\;\Phi;\Delta, x : S}(a(\gamma;
      p(\tau)~,~{\uparrow^S_{\;\Phi;\Delta, x : S}\!(x)}))
      \tag{Notice by IH $\uparrow^S_{\;\Psi;\Gamma, x : S}\!(x)[\gamma; q(\tau)] =
      \uparrow^S_{\;\Psi;\Gamma, x : S}\!(x[\gamma; q(\tau)]) =
      \uparrow^S_{\;\Phi;\Delta, x : S}\!(x)$} \\
    &= \lambda x. \downarrow^{T'}_{\;\Phi;\Delta, x : S}(a((\gamma; \tau) \circ (\id;
      p(\id))~,~{\uparrow^S_{\;\Phi;\Delta, x : S}\!(x)}))
      \tag{$\gamma; p(\tau) = (\gamma \circ \id); (\tau \circ p(\id)) = (\gamma; \tau)
      \circ (\id; p(\id))$} \\
    &= \lambda x. \downarrow^{T'}_{\;\Phi;\Delta, x : S}(a[\gamma; \tau](\id;
      p(\id)~,~{\uparrow^S_{\;\Phi;\Delta, x : S}\!(x)}))
      \tag{functoriality of $\intp{S \func T'}$} \\
    &= \downarrow^{S \func T'}_{\;\Phi;\Delta}(a[\gamma;\tau])
  \end{align*}

  Now we check reflection. %
  Assume $v \in \Ne^{S \func T'}_{\;\Psi;\Gamma}$, then we compute
  \begin{align*}
    \uparrow^{S \func T'}_{\;\Psi;\Gamma}(v)[\gamma;\tau]
    &= (\gamma';\tau' : \Phi';\Delta' \To \Phi;\Delta)(a)
                 \mapsto \uparrow^{T'}_{\;\Phi';\Delta'}(v[\gamma;\tau \circ
      \gamma';\tau']\ \downarrow^S_{\;\Phi';\Delta'}(a)) \\
    &= (\gamma';\tau' : \Phi';\Delta' \To \Phi;\Delta)(a)
      \mapsto \uparrow^{T'}_{\;\Phi';\Delta'}(v[\gamma;\tau][\gamma';\tau']\
      \downarrow^S_{\;\Phi';\Delta'}(a))
      \tag{by composition} \\
    &= \uparrow^{S \func T'}_{\;\Phi;\Delta}(v[\gamma;\tau])
\end{align*}
\end{proof}

Next we show that at layer $0$, the component of global weakenings itself is natural.
\begin{lemma}\labeledit{lem:st:gsubst-nat-gen}
  If $\ltyping[\Phi][\Delta] 0 t T$, $\gamma' : \Psi' \To_g \Psi$, $\gamma : \Psi \To_g \Phi$ and
  $\rho \in \intp{\Delta}_{\;\Psi';\Gamma}$, then
  $\intp{t}^0_{\;\Psi';\Gamma}(\gamma \circ \gamma';\rho) = \intp{t[\gamma]}^0_{\;\Psi';\Gamma}(\gamma';\rho)$.
\end{lemma}
\begin{proof}
  Since we are only concerned about the global weakenings and we are at layer $0$, so
  the only relevant case is the global variable case. %
  The other cases are immediate by IH. %
  Let $t = u$. %
  \begin{align*}
    \intp{u}^0_{\;\Psi';\Gamma}(\gamma \circ \gamma';\rho)
    &= \uparrow_{\;\Psi';\Gamma}(u[\gamma \circ \gamma']) \\
    &= \uparrow_{\;\Psi';\Gamma}(u[\gamma][\gamma']) \\
    &= \intp{u[\gamma]}^0_{\;\Psi';\Gamma}(\gamma';\rho)
  \end{align*}
\end{proof}
\begin{corollary}\labeledit{lem:st:gsubst-nat}
  If $\ltyping[\Phi][\Delta] 0 t T$, $\gamma : \Psi \To_g \Phi$ and
  $\rho \in \intp{\Delta}_{\;\Psi;\Gamma}$, then
  $\intp{t}^0_{\;\Psi;\Gamma}(\gamma;\rho) = \intp{t[\gamma]}^0_{\;\Psi;\Gamma}(\id;\rho)$.
\end{corollary}
\begin{proof}
  This is because $\gamma \circ \id = \gamma$. 
\end{proof}
This corollary allows us to always set the global weakening to $\id$ by pushing the
non-$\id$ one eagerly into the term. 

Now we show that $\intp{t}^i$ is also natural.
\begin{proof}[Proof of \Cref{lem:st:t-intp-nat}]
  We do induction on $\ltyping[\Phi][\Delta] i t T$. %
  We only focus on the modal cases:
  \begin{itemize}[label=Case]
  \item $i = 0$ and $t = u$. In this case,
    \begin{align*}
      \intp{u}^0_{\;\Psi;\Gamma}(\gamma';\rho)[\gamma;\tau]
      &= \uparrow_{\;\Psi;\Gamma}(u[\gamma'])[\gamma;\tau] \\
      &= \uparrow_{\;\Psi';\Gamma'}(u[\gamma'][\gamma;\tau])
      \tag{naturality of reflection} \\
      &= \uparrow_{\;\Psi';\Gamma'}(u[\gamma'][\gamma]) \\
      &= \uparrow_{\;\Psi';\Gamma'}(u[\gamma' \circ \gamma]) \\
      &= \intp{u}^0_{\;\Psi';\Gamma'}(\gamma' \circ \gamma;\rho[\gamma;\tau]) \\
      &= \intp{u}^0_{\;\Psi';\Gamma'}((\gamma';\rho)[\gamma;\tau])
    \end{align*}

  \item $i = 1$ and $t = u$. In this case,
    \begin{align*}
      \intp{u}^1_{\;\Psi;\Gamma}(\sigma;\rho)[\gamma;\tau]
      &= \intp{\sigma(u)}^0_{\;\Psi;\Gamma}(\id; *)[\gamma;\tau] \\
      &= \intp{\sigma(u)}^0_{\;\Psi';\Gamma'}((\id; *)[\gamma;\tau])
        \byIH \\
      &= \intp{\sigma(u)}^0_{\;\Psi';\Gamma'}(\gamma; *) \\
      &= \intp{\sigma(u)[\gamma]}^0_{\;\Psi';\Gamma'}(\id; *)
      \tag{by \Cref{lem:st:gsubst-nat}} \\
      &= \intp{\sigma[\gamma](u)}^0_{\;\Psi';\Gamma'}(\id; *) \\
      &= \intp{u}^1_{\;\Psi';\Gamma'}((\sigma;\rho)[\gamma;\tau])
    \end{align*}
    
  \item
    $i = 1$ and $t = \boxit t'$. 
    \begin{align*}
      \intp{\boxit t'}^1_{\;\Psi;\Gamma}(\sigma;\rho)[\gamma;\tau]
      &= \boxit {(t'[\sigma])}[\gamma;\tau] \\
      &= \boxit {(t'[\sigma][\gamma])} \\
      &= \boxit {(t'[\sigma[\gamma]])} \\
      &= \intp{\boxit t'}^1_{\;\Psi';\Gamma'}(\sigma[\gamma]; \rho[\gamma;\tau]) \\
      &= \intp{\boxit t'}^1_{\;\Psi';\Gamma'}((\sigma;\rho)[\gamma;\tau])
    \end{align*}
    
  \item 
    $i = 1$ and $t = \letbox u s t'$ and $\intp{s}^1_{\;\Psi;\Gamma}(\sigma; \rho) =
    \boxit {t^c}$. %
    Then we have
    \begin{align*}
      \intp{s}^1_{\;\Psi;\Gamma}(\sigma; \rho)[\gamma;\tau]
      = \intp{s}^1_{\;\Psi';\Gamma'}((\sigma; \rho)[\gamma;\tau])
      = \boxit {(t^c[\gamma])}
      \byIH
    \end{align*}
    Then
    \begin{align*}
      \intp{\letbox u s t'}^1_{\;\Psi;\Gamma}(\sigma;\rho)[\gamma;\tau]
      &= \intp{t'}^1_{\;\Psi;\Gamma}(\sigma,t^c/u;\rho)[\gamma;\tau] \\
      &= \intp{t'}^1_{\;\Psi';\Gamma'}((\sigma,t^c/u;\rho)[\gamma;\tau])
      \byIH \\
      &= \intp{t'}^1_{\;\Psi';\Gamma'}(\sigma[\gamma],t^c[\gamma]/u;\rho[\gamma;\tau])
      \\
      &= \intp{\letbox u s t'}^1_{\;\Psi;\Gamma}((\sigma;\rho)[\gamma;\tau])
    \end{align*}

  \item 
    $i = 1$ and $t = \letbox u s t'$ and $\intp{s}^1_{\;\Psi;\Gamma}(\sigma; \rho) = v$. %
    Then we have
    \begin{align*}
      \intp{s}^1_{\;\Psi;\Gamma}(\sigma; \rho)[\gamma;\tau]
      = \intp{s}^1_{\;\Psi';\Gamma'}((\sigma; \rho)[\gamma;\tau])
      = v[\gamma;\tau]
      \byIH
    \end{align*}
    And $(\sigma'; \rho') := (\sigma; \rho)[p(\id); \id]$:
    \begin{align*}
      & \intp{\letbox u s t'}^1_{\;\Psi;\Gamma}(\sigma;\rho)[\gamma;\tau] \\
      =~& \uparrow^{T}_{\;\Psi;\Gamma}(\letbox u v {\downarrow^T_{\;\Psi, u: S;\Gamma}(\intp{t'}^1_{\;\Psi, u:
          S;\Gamma}(\sigma', u/u; \rho'))})[\gamma;\tau] \\
      =~& \uparrow^{T}_{\;\Psi';\Gamma'}(\letbox u v {\downarrow^T_{\;\Psi, u: S;\Gamma}(\intp{t'}^1_{\;\Psi, u:
          S;\Gamma}(\sigma', u/u; \rho'))}[\gamma;\tau])
          \tag{naturality of reflection} \\
      =~& \uparrow^{T}_{\;\Psi';\Gamma'}(\letbox u {v[\gamma;\tau]} {\downarrow^T_{\;\Psi, u: S;\Gamma}(\intp{t'}^1_{\;\Psi, u:
          S;\Gamma}(\sigma', u/u; \rho'))[q(\gamma);\tau]}) \\
      =~& \uparrow^{T}_{\;\Psi';\Gamma'}(\letbox u {v[\gamma;\tau]}
          {\downarrow^T_{\;\Psi', u: S;\Gamma'}(\intp{t'}^1_{\;\Psi, u:
          S;\Gamma}(\sigma', u/u; \rho')[q(\gamma);\tau])})
          \tag{naturality of reification} \\
      =~& \uparrow^{T}_{\;\Psi';\Gamma'}(\letbox u {v[\gamma;\tau]}
          {\downarrow^T_{\;\Psi', u: S;\Gamma'}(\intp{t'}^1_{\;\Psi', u:
          S;\Gamma'}((\sigma', u/u; \rho')[q(\gamma);\tau]))})
          \byIH \\
      =~& \uparrow^{T}_{\;\Psi';\Gamma'}(\letbox u {v[\gamma;\tau]}
          {\downarrow^T_{\;\Psi', u: S;\Gamma'}(\intp{t'}^1_{\;\Psi', u:
          S;\Gamma'}((\sigma[p(\gamma)], u/u; \rho[p(\gamma);\tau])))}) \\
      =~& \uparrow^{T}_{\;\Psi';\Gamma'}(\letbox u {v[\gamma;\tau]}
          {\downarrow^T_{\;\Psi', u: S;\Gamma'}(\intp{t'}^1_{\;\Psi', u:
          S;\Gamma'}((\sigma[\gamma \circ p(\id)], u/u; \rho[\gamma \circ p(\id);\tau
          \circ \id])))}) \\
      =~& \uparrow^{T}_{\;\Psi';\Gamma'}(\letbox u {v[\gamma;\tau]}
          {\downarrow^T_{\;\Psi', u: S;\Gamma'}(\intp{t'}^1_{\;\Psi', u:
          S;\Gamma'}((\sigma[\gamma][p(\id)], u/u; \rho[\gamma; \tau][p(\id);\id])))}) \\
      =~& \intp{\letbox u s t'}^1_{\;\Psi';\Gamma'}((\sigma;\rho)[\gamma;\tau])
    \end{align*}
  \end{itemize}
\end{proof}

We have established the naturality of interpretation. %
Next we establish \Cref{lem:st:glob-subst}. %
\Cref{lem:st:gsubst-nat-gen} has shown the case where $i = 0$, so we are only
concerned about the case where $i = 1$. %
We first work on the following lemma which bridges the interpretation at both layers:
\begin{lemma}\labeledit{lem:st:intp-0-1}
  If $\ltyping[\Phi][\Delta] 0 t T$, $\typing[\Psi]\sigma\Phi$, $\gamma : \Psi' \To_g \Psi$ and $\rho \in
  \intp{\Delta}_{\;\Psi';\Gamma}$, then
  $\intp{t[\sigma]}^0_{\;\Psi';\Gamma}(\gamma; \rho) = \intp{t}^1_{\;\Psi';\Gamma}(\sigma[\gamma]; \rho)$.
\end{lemma}
\begin{proof}
  We proceed by induction on $\ltyping[\Phi][\Delta] 0 t T$. %
  We only consider significant cases:
  \begin{itemize}[label=Case]
  \item $t = u$, then
    \begin{align*}
      \intp{u[\sigma]}^0_{\;\Psi';\Gamma}(\gamma; \rho)
      &= \intp{\sigma(u)}^0_{\;\Psi';\Gamma}(\gamma; \rho) \\
      &= \intp{\sigma(u)[\gamma]}^0_{\;\Psi';\Gamma}(\id; \rho)
        \tag{by \Cref{lem:st:gsubst-nat}} \\
      &= \intp{\sigma(u)[\gamma]}^0_{\;\Psi';\Gamma}(\id; *)
        \tag{we know $\sigma(u)[\gamma]$ is a closed term} \\
      &= \intp{u}^1_{\;\Psi';\Gamma}(\sigma[\gamma]; \rho)
    \end{align*}
  \item $t = x$, then it is immediate because they both look up the same $\rho$. 
  \item $t = \lambda x. t'$, then
    \begin{align*}
      & \intp{\lambda x. t'[\sigma]}^0_{\;\Psi';\Gamma}(\gamma;\rho) \\
      =~& (\gamma';\tau : \Psi'';\Gamma'' \To \Psi';\Gamma)(a) \mapsto
          \intp{t'[\sigma]}^0_{\;\Psi'';\Gamma''} (\gamma \circ \gamma';
          (\rho[\gamma'; \tau], a)) \\
      =~& (\gamma';\tau : \Psi'';\Gamma'' \To \Psi';\Gamma)(a) \mapsto
          \intp{t'}^1_{\;\Psi'';\Gamma''} (\sigma[\gamma \circ \gamma'];
          (\rho[\gamma'; \tau], a))
          \byIH \\
      =~& (\gamma';\tau : \Psi'';\Gamma'' \To \Psi';\Gamma)(a) \mapsto
          \intp{t'}^1_{\;\Psi'';\Gamma''} (\sigma[\gamma][\gamma'];
          (\rho[\gamma'; \tau], a)) \\
      =~& \intp{\lambda x. t'}^1_{\;\Psi';\Gamma}(\sigma[\gamma];\rho) 
    \end{align*}
  \item $t = t'\ s$, immediate by IH. 
  \end{itemize}
\end{proof}

\begin{lemma}
  If $\ltyping[\Phi][\Delta] 1 t T$, $\typing[\Psi']{\sigma'}\Psi$, $\typing[\Psi]{\sigma}\Phi$ and
  $\rho \in \intp{\Delta}_{\;\Psi';\Gamma}$, then
  $\intp{t}^1_{\;\Psi';\Gamma}(\sigma \circ \sigma';\rho) = \intp{t[\sigma]}^1_{\;\Psi';\Gamma}(\sigma';\rho)$.
\end{lemma}
\begin{proof}
  We proceed by induction on $\ltyping[\Phi][\Delta] 1 t T$. %
  We only consider the cases that are not immediate.
  \begin{itemize}[label=Case]
  \item $t = u$, we compute
    \begin{align*}
      \intp{u}^1_{\;\Psi';\Gamma}(\sigma \circ \sigma';\rho)
      &= \intp{(\sigma \circ \sigma')(u)}^0_{\;\Psi';\Gamma}(\id; *) \\
      &= \intp{\sigma(u)[\sigma']}^0_{\;\Psi';\Gamma}(\id; *) \\
      &= \intp{\sigma(u)}^1_{\;\Psi';\Gamma}(\sigma';\rho)
        \tag{by \Cref{lem:st:intp-0-1}}
    \end{align*}
    
  \item $t = \boxit t'$, we also compute
    \begin{align*}
      \intp{\boxit t'}^1_{\;\Psi';\Gamma}(\sigma \circ \sigma';\rho)
      &= \boxit {(t'[\sigma \circ \sigma'])} \\
      &= \boxit {(t'[\sigma][\sigma'])} \\
      &= \intp{\boxit t'[\sigma]}^1_{\;\Psi';\Gamma}(\sigma';\rho)
    \end{align*}

  \item $t = \letbox u s t'$ and assume
    $\intp{s}^1_{\;\Psi';\Gamma}(\sigma \circ \sigma'; \rho) = \boxit {t^c}$. %
    By IH, we have
    \begin{align*}
      \intp{s}^1_{\;\Psi';\Gamma}(\sigma \circ \sigma'; \rho) = \intp{s[\sigma]}^1_{\;\Psi';\Gamma}(\sigma'; \rho) = \boxit {t^c}
    \end{align*}
    Then we compute
    \begin{align*}
      \intp{\letbox u s t'}^1_{\;\Psi';\Gamma}(\sigma \circ \sigma';\rho)
      &= \intp{t'}^1_{\;\Psi';\Gamma}((\sigma \circ \sigma'), t^c/u;\rho) \\
      &= \intp{t'}^1_{\;\Psi';\Gamma}((\sigma[p(\id)], u/u) \circ (\sigma', t^c/u);\rho) \\
      &= \intp{t'[\sigma[p(\id)], u/u]}^1_{\;\Psi';\Gamma}(\sigma', t^c/u;\rho)
        \byIH \\
      &= \intp{\letbox u s t'[\sigma]}^1_{\;\Psi';\Gamma}(\sigma';\rho)
    \end{align*}
    We notice that
    \begin{align*}
      \letbox u s t'[\sigma]
      = \letbox u {s[\sigma]} {(t'[\sigma[p(\id)], u/u])}
    \end{align*}
    
  \item $t = \letbox u s t'$ and assume
    $\intp{s}^1_{\;\Psi';\Gamma}(\sigma \circ \sigma'; \rho) = v : \square S$. %
    By IH, we have
    \begin{align*}
      \intp{s}^1_{\;\Psi';\Gamma}(\sigma \circ \sigma'; \rho) = \intp{s[\sigma]}^1_{\;\Psi';\Gamma}(\sigma'; \rho) = v : \square S
    \end{align*}
    Then we compute
    \begin{align*}
      & \intp{\letbox u s t'}^1_{\;\Psi';\Gamma}(\sigma \circ \sigma';\rho) \\
      =~& \uparrow^{T}_{\;\Psi';\Gamma}(\letbox u v {\downarrow^T_{\;\Psi', u: S;\Gamma}(\intp{t}^1_{\;\Psi', u:
      S;\Gamma}((\sigma \circ \sigma')[p(\id)], u/u; \rho[p(\id); \id]))}) \\
      =~& \uparrow^{T}_{\;\Psi';\Gamma}(\letbox u v {\downarrow^T_{\;\Psi', u: S;\Gamma}(\intp{t}^1_{\;\Psi', u:
      S;\Gamma}(\sigma \circ (\sigma'[p(\id)]), u/u; \rho[p(\id); \id]))}) \\
      =~& \uparrow^{T}_{\;\Psi';\Gamma}(\letbox u v {\downarrow^T_{\;\Psi', u: S;\Gamma}(\intp{t}^1_{\;\Psi', u:
      S;\Gamma}((\sigma, u/u) \circ (\sigma'[p(\id)], u/u); \rho[p(\id); \id]))}) \\
      =~& \uparrow^{T}_{\;\Psi';\Gamma}(\letbox u v {\downarrow^T_{\;\Psi', u: S;\Gamma}(\intp{t[\sigma, u/u]}^1_{\;\Psi', u:
          S;\Gamma}(\sigma'[p(\id)], u/u; \rho[p(\id); \id]))})
          \byIH \\
      =~& \intp{\letbox u s t'[\sigma]}^1_{\;\Psi';\Gamma}(\sigma';\rho)
    \end{align*}
    
  \item $t = \lambda x. t'$, we compute
    \begin{align*}
      & \intp{\lambda x. t'}^i_{\;\Psi';\Gamma}(\sigma \circ \sigma';\rho) \\
      =~& (\gamma;\tau : \Psi'';\Gamma'' \To \Psi';\Gamma)(a) \mapsto
          \intp{t'}^i_{\;\Psi'';\Gamma''} ((\sigma \circ \sigma')[\gamma];
          (\rho[\gamma; \tau], a)) \\
      =~& (\gamma;\tau : \Psi'';\Gamma'' \To \Psi';\Gamma)(a) \mapsto
          \intp{t'}^i_{\;\Psi'';\Gamma''} ((\sigma \circ \sigma'[\gamma]);
          (\rho[\gamma; \tau], a))
          \tag{notice $(\sigma \circ \sigma')[\gamma] = \sigma \circ \sigma'[\gamma]$} \\
      =~& (\gamma;\tau : \Psi'';\Gamma'' \To \Psi';\Gamma)(a) \mapsto
          \intp{t'[\sigma]}^i_{\;\Psi'';\Gamma''} (\sigma'[\gamma];
          (\rho[\gamma; \tau], a))
          \byIH \\
      =~& \intp{\lambda x. t'[\sigma]}^i_{\;\Psi';\Gamma}(\sigma';\rho)
    \end{align*}
  \end{itemize}
\end{proof}

\subsection{Interactions with Local Substitutions}

We have one last theorem remains, which describes the interactions between evaluation
and local substitutions. %
A similar proof is done for the presheaf model
in~\citet{DBLP:journals/corr/abs-2206-07823}, so we only sketch the essence in this
section. 
\begin{definition}
  Given $\rho \in \intp{\Delta, \Delta'}_{\;\Psi;\Gamma}$ and
  $a \in \intp{T}_{\;\Psi;\Gamma}$, then
  $\inser(\rho, \Delta', x, a) \in \intp{\Delta, x : T, \Delta'}_{\;\Psi;\Gamma}$ is defined as
  follows:
  \begin{align*}
    \inser(\rho, \cdot, x, a) &:= \rho, a \\
    \inser((\rho, b), (\Delta'', y : S), x, a) &:= \inser(\rho, \Delta'', x, a), b
  \end{align*}
\end{definition}

The $\inser$ function is monotonic w.r.t. weakenings.
\begin{lemma}
  $\inser(\rho, \Delta', x, a)[\gamma; \tau] = \inser(\rho[\gamma; \tau], \Delta', x,
  a[\gamma; \tau])$
\end{lemma}
\begin{proof}
  Induction on $\Delta'$.
  \begin{itemize}[label=Case]
  \item $\Delta' = \cdot$, immediate.
  \item $\Delta' = \Delta'', y : S$,
    \begin{align*}
      \inser((\rho, b), (\Delta'', y : S), x, a)[\gamma;\tau]
      &= (\inser(\rho, \Delta'', x, a), b)[\gamma;\tau] \\
      &= \inser(\rho, \Delta'', x, a)[\gamma;\tau], b[\gamma;\tau] \\
      &= \inser(\rho[\gamma;\tau], \Delta'', x, a[\gamma;\tau]), b[\gamma;\tau]
        \byIH \\
      &= \inser(\rho[\gamma;\tau], (\Delta'', y : S), x, a[\gamma;\tau])
    \end{align*}
  \end{itemize}
\end{proof}

\begin{lemma}
  If $\ltyping[\Phi][\Delta, x : S, \Delta'] i t T$ and $\ltyping[\Phi][\Delta,
  \Delta'] i s S$, then
  $\intp{t[s/x]}^i_{\;\Psi;\Gamma}(\sigma; \rho) = \intp{t}^i_{\;\Psi;\Gamma}(\sigma; \inser(\rho, \Delta', x,
  \intp{s}^i_{\;\Psi;\Gamma}(\sigma; \rho)))$
\end{lemma}
\begin{proof}
  We proceed by induction on $\ltyping[\Phi][\Delta, x : S, \Delta'] i t T$.
  \begin{itemize}[label=Case]
  \item $t = x$, immediate.
  \item $t = y$ and $x \neq y$, immediate.
  \item $t = \lambda y : S'. t'$, then
    \begin{align*}
      & \intp{\lambda y. t'[s/x]}^i_{\;\Psi';\Gamma}(\sigma;\rho) \\
      =~& (\gamma;\tau : \Psi';\Gamma' \To \Psi;\Gamma)(a) \mapsto
          \intp{t'[s/x]}^i_{\;\Psi';\Gamma'} (\sigma[\gamma];
          (\rho[\gamma; \tau], a)) \\
      =~& (\gamma;\tau : \Psi';\Gamma' \To \Psi;\Gamma)(a) \mapsto
          \intp{t'}^i_{\;\Psi';\Gamma'} (\sigma[\gamma];
          \inser((\rho[\gamma; \tau], a), (\Delta', y : S'), x, \intp{s}^i_{\;\Psi';\Gamma}(\sigma[\gamma];
          (\rho[\gamma; \tau], a))))
          \byIH \\
      =~& (\gamma;\tau : \Psi';\Gamma' \To \Psi;\Gamma)(a) \mapsto
          \intp{t'}^i_{\;\Psi';\Gamma'} (\sigma[\gamma];
          \inser(\rho[\gamma; \tau], \Delta', x, \intp{s}^i_{\;\Psi';\Gamma}(\sigma[\gamma];
          (\rho[\gamma; \tau], a))), a) \\
      =~& (\gamma;\tau : \Psi';\Gamma' \To \Psi;\Gamma)(a) \mapsto
          \intp{t'}^i_{\;\Psi';\Gamma'} (\sigma[\gamma];
          \inser(\rho[\gamma; \tau], \Delta', x, \intp{s}^i_{\;\Psi';\Gamma}((\sigma;
          \rho)[\gamma; \tau])), a)
          \tag{$s$ was locally weakening when inside of a $\lambda$} \\
      =~& (\gamma;\tau : \Psi';\Gamma' \To \Psi;\Gamma)(a) \mapsto
          \intp{t'}^i_{\;\Psi';\Gamma'} (\sigma[\gamma];
          \inser(\rho, \Delta', x, \intp{s}^i_{\;\Psi';\Gamma}(\sigma;
          \rho))[\gamma; \tau], a)
          \tag{monotonicity of $\inser$} \\
      =~& \intp{\lambda y. t'}^i_{\;\Psi;\Gamma}(\sigma; \inser(\rho, \Delta', x,
          \intp{s}^i(\sigma; \rho)))
    \end{align*}
  \end{itemize}
\end{proof}

Now we can conclude our target lemma:
\begin{proof}[Proof of \Cref{lem:st:loc-subst}]
  Let $\Delta'$ be empty and immediate from the previous lemma.
\end{proof}

\subsection{Semantic Judgments and Completeness}

Now we have established all desired properties of the interpretations. %
In this section, we establish \Cref{thm:st:compl} i.e. the completeness theorem. %
In the proof of the completeness theorem, we must establish the relation between
evaluations of two equivalent terms. %
Luckily, in the presheaf model, this is very easy and this relation can just be
defined as point-wise equality in the semantics as we will see very soon. %
More specifically, the semantic judgments must also be layered due to the behavioral
differences between layers.
\begin{mathpar}
  \inferrule*
  {\ltyping 0 s T \\ \ltyping 0 t T \\\\ s = t}
  {\lsemtyeq 0{s}{t}{T}}

  \inferrule*
  {\ltyping 1 s T \\ \ltyping 1 t T \\\\
    \forall (\sigma; \rho) \in \intp{\Psi; \Gamma}^1_{\;\Phi;\Delta}.
    \intp{s}^1_{\;\Phi;\Delta}(\sigma;\rho) = \intp{t}^1_{\;\Phi;\Delta}(\sigma;\rho)}
  {\lsemtyeq 1{s}{t}{T}}

  \inferrule*
  {\lsemtyeq i{t}{t}{T}}
  {\lsemtyp i t T}
\end{mathpar}

The semantic equivalence at layer $0$ is straightforward: it simply requires
the terms are syntactically equal. %
This makes much sense because terms at layer $0$ are just code and they do not have
dynamics. %
At layer $1$, due to dynamics, the semantic equivalence is more typical. %
It requires two terms are equal when evaluated in the same environment. %
Our target theorem is
\begin{theorem}[Fundamental]\labeledit{thm:st:comp-fund} $ $
  \begin{itemize}
  \item If $\ltyping i t T$, then $\lsemtyp i t T$.
  \item If $\ltyequiv i s t T$, then $\lsemtyeq i s t T$.
  \end{itemize}
\end{theorem}

The completeness theorem is then a corollary of the fundamental theorem at layer $1$. %
Notice that at layer $0$, we do not really have much to do: the layer $0$ only
requires syntactic equality, which we have shown is implied by the syntactic
equivalence judgment. %
Therefore, we only focus on layer $1$.

Most PER and congruence rules are immediate because our underlying equivalence in the
semantics is equality. %
We will examine a few cases.

\subsubsection{Congruence Rules}

\begin{lemma}
  \begin{mathpar}
    \inferrule*
    {\iscore \Psi \\ \lsemtyeq[\Psi][\cdot]{0}{t}{t'}T}
    {\lsemtyeq{1}{\boxit t}{\boxit t'}{\square T}}
  \end{mathpar}
\end{lemma}
\begin{proof}
  Assume $(\sigma; \rho) \in \intp{\Psi; \Gamma}^1_{\;\Phi;\Delta}$,
  we compute
  \begin{align*}
    \intp{\boxit t}^1_{\;\Phi;\Delta}(\sigma; \rho)
    &= \boxit {(t[\sigma])} \\
    \intp{\boxit t'}^1_{\;\Phi;\Delta}(\sigma; \rho)
    &= \boxit {(t'[\sigma])} 
  \end{align*}
  Notice $t = t'$ by $\lsemtyeq[\Psi][\cdot]{0}{t}{t'}T$. %
  Then
  \begin{align*}
    \intp{\boxit t}^1_{\;\Phi;\Delta}(\sigma; \rho)
    = \intp{\boxit t'}^1_{\;\Phi;\Delta}(\sigma; \rho)
  \end{align*}
\end{proof}

\begin{lemma}
  \begin{mathpar}
    \inferrule*
    {\lsemtyeq 1 {s}{s'}{\square T} \\ \lsemtyeq[\Psi, u : T] 1 {t}{t'}{T'}}
    {\lsemtyeq 1 {\letbox u s t}{\letbox u{s'}{t'}} T'}
  \end{mathpar}
\end{lemma}
\begin{proof}
  Assume $(\sigma; \rho) \in \intp{\Psi; \Gamma}^1_{\;\Phi;\Delta}$,
  we have 
  \begin{align*}
    \intp{s}^1_{\;\Phi;\Delta}(\sigma; \rho)
    = \intp{s'}^1_{\;\Phi;\Delta}(\sigma; \rho)
    \in \intp{\square T}_{\;\Phi;\Delta}
    = \Nf^{\square T}_{\;\Phi;\Delta}
  \end{align*}
  There are two possibilities: either $s$ and $s'$ evaluate to some $\boxit{t^c}$ or
  they evaluate to a neutral term $v$.
  \begin{itemize}[label=Case]
  \item $\intp{s}^1_{\;\Phi;\Delta}(\sigma; \rho)
    = \intp{s'}^1_{\;\Phi;\Delta}(\sigma; \rho) = \boxit {t^c}$.

    Moreover, we have
    \begin{align*}
      (\sigma, t^c/u; \rho) \in \intp{\Psi, u : T; \Gamma}^1_{\;\Phi;\Delta}
    \end{align*}
    due to which we have
    \begin{align*}
      \intp{t}^1_{\;\Phi;\Delta}(\sigma, t^c/u; \rho)
      = \intp{t'}^1_{\;\Phi;\Delta}(\sigma, t^c/u; \rho)
    \end{align*}
    This concludes our first case.
    
  \item $\intp{s}^1_{\;\Phi;\Delta}(\sigma; \rho)
    = \intp{s'}^1_{\;\Phi;\Delta}(\sigma; \rho) = v$.

    Then
    \begin{align*}
      (\sigma, u/u; \rho) \in \intp{\Psi, u : T; \Gamma}^1_{\;\Phi, u : T;\Delta}
    \end{align*}
    due to which we have
    \begin{align*}
      \intp{t}^1_{\;\Phi, u : T;\Delta}(\sigma, u/u; \rho)
      = \intp{t'}^1_{\;\Phi, u : T;\Delta}(\sigma, u/u; \rho)
    \end{align*}
    By congruence, we have
    \begin{align*}
      \uparrow^{T'}_{\;\Phi;\Delta}(\letbox u v {\downarrow^{T'}_{\;\Phi, u: T;\Delta}(\intp{t}^1_{\;\Phi, u:
      T;\Delta}(\sigma, u/u; \rho))})
      = \uparrow^{T'}_{\;\Phi;\Delta}(\letbox u v {\downarrow^{T'}_{\;\Phi, u: T;\Delta}(\intp{t'}^1_{\;\Phi, u:
      T;\Delta}(\sigma, u/u; \rho))})
    \end{align*}
    This concludes our goal.
  \end{itemize}
\end{proof}

\subsubsection{$\beta$ Rules}

\begin{lemma}
  \begin{mathpar}
    \inferrule*
    {\lsemtyp[\Psi][\Gamma, x : S] 1 t T \\ \lsemtyp 1 s S}
    {\lsemtyeq 1{(\lambda x. t)\ s}{t[s/x]}{T}}
  \end{mathpar}
\end{lemma}
\begin{proof}
  Assume $(\sigma; \rho) \in \intp{\Psi; \Gamma}^1_{\;\Phi;\Delta}$,
  we have 
  \begin{align*}
    \intp{t}^1_{\;\Phi;\Delta}(\sigma; (\rho, \intp{s}^1_{\;\Phi;\Delta}))
    = \intp{t[s/x]}^1_{\;\Phi;\Delta}(\sigma; \rho)
  \end{align*}
  by~\Cref{lem:st:loc-subst}.
\end{proof}

\begin{lemma}
  \begin{mathpar}
    \inferrule*
    {\lsemtyp[\Psi][\cdot] 0 s T \\ \lsemtyp[\Psi, u : T] 1 {t}{T'}}
    {\lsemtyeq 1 {\letbox u {\boxit s} t}{t[s/u]}T}
  \end{mathpar}  
\end{lemma}
\begin{proof}
  Assume $(\sigma; \rho) \in \intp{\Psi; \Gamma}^1_{\;\Phi;\Delta}$,
  we have 
  \begin{align*}
    \intp{\letbox u {\boxit s} t}^1_{\;\Phi;\Delta}(\sigma; \rho)
    &= \intp{t}^1_{\;\Phi;\Delta}(\sigma, s[\sigma]/u; \rho)
      \tag{$s$'s typing suggests that it is a core term} \\
    &= \intp{t}^1_{\;\Phi;\Delta}((\id, s/u) \circ \sigma; \rho) \\
    &= \intp{t[s/u]}^1_{\;\Phi;\Delta}(\sigma; \rho)
    \tag{by~\Cref{lem:st:glob-subst}} 
  \end{align*}
\end{proof}

\subsubsection{$\eta$ Rules}

\begin{lemma}
  \begin{mathpar}
    \inferrule*
    {\lsemtyp 1 {t}{S \func T}}
    {\lsemtyeq 1 t {\lambda x. (t\ x)}{S \func T}}
  \end{mathpar}
\end{lemma}
\begin{proof}
  Assume $(\sigma; \rho) \in \intp{\Psi; \Gamma}^1_{\;\Phi;\Delta}$,
  we have 
  \begin{align*}
    \intp{\lambda x. t\ x}^1_{\;\Phi;\Delta}(\sigma; \rho)
    &= (\gamma;\tau : \Psi';\Gamma' \To \Psi;\Gamma)(a) \mapsto
          \intp{t\ x}^1_{\;\Psi';\Gamma'} (\sigma[\gamma];
          (\rho[\gamma; \tau], a)) \\
    &= (\gamma;\tau : \Psi';\Gamma' \To \Psi;\Gamma)(a) \mapsto
          \intp{t}^1_{\;\Psi';\Gamma'} (\sigma[\gamma];
          (\rho[\gamma; \tau], a), \id, a) \\
    &= (\gamma;\tau : \Psi';\Gamma' \To \Psi;\Gamma)(a) \mapsto
          \intp{t}^1_{\;\Psi';\Gamma'} ((\sigma; \rho)[\gamma; \tau], \id, a)
      \tag{due to weakening in $t$} \\
    &= (\gamma;\tau : \Psi';\Gamma' \To \Psi;\Gamma)(a) \mapsto
      \intp{t}^1_{\;\Psi';\Gamma'} (\sigma; \rho, \gamma; \tau \circ \id, a)
      \tag{functoriality of $\intp{t}$} \\
    &= (\gamma;\tau : \Psi';\Gamma' \To \Psi;\Gamma)(a) \mapsto
      \intp{t}^1_{\;\Psi';\Gamma'} (\sigma; \rho, \gamma; \tau, a) \\
    &= \intp{t}^1_{\;\Phi;\Delta}(\sigma; \rho)
  \end{align*}
\end{proof}

\subsubsection{Completeness}

With the semantic judgments verified, we establish the fundamental theorems
(\Cref{thm:st:comp-fund}). %
This subsequently implies the completeness theorem (\Cref{thm:st:compl}) because
$\uparrow^{\Psi; \Gamma} \in \intp{\Psi; \Gamma}^1_{\;\Psi;\Gamma}$. 

\subsection{Soundness}

In this section, we show the soundness theorem (\Cref{thm:st:sound}) of the
normalization algorithm. %
In the soundness proof, we use a \emph{gluing} model, which relates a syntactic term
with a natural transformation in the presheaf model. %
For a gluing relation $R$, we write $a \sim b \in R$ to denote $(a, b) \in R$. %

It is probably unsurprising that layers play some part in the definition of the gluing
model. %
However, it is quite surprising for us how much layers need to participate:
effectively, we have two sets of gluing models to expose the fact that this 2-layered
system effectively contains two languages. %
We will see more concretely once we look into the semantic judgments.

\subsubsection{Layer-$0$ Gluing Model} \labeledit{sec:st:glue-0}

In this section, we define the gluing model at layer $0$. %
At layer $0$, we only have access to the core types. %
This fact is very important when we interpret the types. %
We first define the following relation which glues terms and values of natural
numbers:
\begin{mathpar}
  \inferrule*
  {\ltyequiv 1 t \ze \Nat}
  {t \sim \ze \in \Nat_{\;\Psi;\Gamma}}

  \inferrule*
  {\ltyequiv 1 t {\su t'} \Nat \\ t' \sim w \in \Nat_{\;\Psi;\Gamma}}
  {t \sim \su w \in \Nat_{\;\Psi;\Gamma}}

  \inferrule*
  {\ltyequiv 1 t v \Nat}
  {t \sim v \in \Nat_{\;\Psi;\Gamma}}
\end{mathpar}
$\Nat_{\;\Psi;\Gamma}$ relates a term with a normal form that is of type $\Nat$. %
The gluing model at layer $0$ is defined as
\begin{align*}
  \glu{T}^0_{\;\Psi;\Gamma} &\subseteq \Exp \times \intp{T}_{\;\Psi;\Gamma} \\
  \glu{\Nat}^0_{\;\Psi;\Gamma} &:= \Nat_{\;\Psi;\Gamma} \\
  \glu{S \func T}^0_{\;\Psi;\Gamma} &:= \{(t, a) \sep \forall
                              \gamma; \tau : \Phi; \Delta \To \Psi;\Gamma, s \sim b \in
                              \glu{S}^0_{\;\Phi;\Delta}. t[\gamma; \tau]\ s \sim a(\gamma; \tau, b) \in
                              \glu{T}^0_{\;\Phi;\Delta} \} 
\end{align*}
Notice that $\glu{T}^0$ does not have a case for $\square$. %
We can make sure that $\glu{T}^0$ is well-defined by assuming $\iscore T$ whenever we
refer to the layer-$0$ gluing model.

\subsubsection{Properties of Gluing Model at Layer $0$}

We need to examine a few properties of the gluing model before proceeding to the
semantic judgments. %
Most these proofs are in fact simpler than those
in~\citet{DBLP:journals/corr/abs-2206-07823} because the $\Nat$ case is defined
through equivalence, so the only case that needs to look into is the function case. %

\begin{lemma}[Monotonicity]\labeledit{lem:st:glue-mon}
  If $t \sim a \in \glu{T}^0_{\;\Psi;\Gamma}$, given $\gamma;\tau : \Phi;\Delta \To \Psi;\Gamma$, then
  $t[\gamma;\tau] \sim a[\gamma;\tau] \in \glu{T}^0_{\;\Phi;\Delta}$. 
\end{lemma}
\begin{proof}
  We proceed by induction on $T$.
  \begin{itemize}[label=Case]
  \item $T = \Nat$, immediate by nested induction.
    
  \item $T = S \func T'$, then
    \begin{align*}
      H_1: &\ \ltyping 1{t}{S \func T'} \tag{by assumption} \\
      H_2: &\ \forall \gamma'; \tau' : \Phi';\Delta' \To \Psi;\Gamma, s \sim b \in
             \glu{S}^0_{\;\Phi';\Delta'}. t[\gamma';\tau']\ s \sim a(\gamma';\tau', b) \in \glu{T'}^0_{\;\vDelta'}
             \tag{by assumption} \\
      &\ \text{assume }\gamma';\tau' : \Phi';\Delta' \To \Phi;\Delta, s \sim b \in
        \glu{S}^0_{\;\Phi';\Delta'} \\
      H_3: &\ t[(\gamma;\tau) \circ (\gamma';\tau')]\ s \sim a((\gamma;\tau) \circ (\gamma';\tau'), b) \in
             \glu{T'}^0_{\;\Phi';\Delta'}
             \tag{by $H_2$} \\
      &\ t[\gamma;\tau][\gamma';\tau']\ s \sim a[\gamma;\tau](\gamma';\tau', b) \in
        \glu{T'}^0_{\;\Phi';\Delta'} \\
      &\ t[\gamma;\tau] \sim a[\gamma;\tau] \in \glu{S \func T'}^0_{\;\Phi;\Delta}
      \tag{by abstraction} 
    \end{align*}
  \end{itemize}
\end{proof}

\begin{lemma}\labeledit{lem:st:glue-resp-equiv}
  If $t \sim a \in \glu{T}^0_{\;\Psi;\Gamma}$ and
  $\ltyequiv 1{t}{t'}{T}$, then $t' \sim a \in \glu{T}^0_{\;\Psi;\Gamma}$.
\end{lemma}
\begin{proof}
  We prove by induction on $T$.
  \begin{itemize}[label=Case]
  \item $T = \Nat$, immediate by case analysis and then transitivity.
    
  \item $T = S \func T'$, then
    \begin{align*}
      &\ \forall \gamma;\tau : \Phi;\Delta \To \Psi;\Gamma,
        s \sim b \in \glu{S}^0_{\;\Phi;\Delta}. t[\gamma;\tau]\ s \sim a(\gamma;\tau,
        b) \in \glu{T'}^0_{\;\Phi;\Delta} \tag{by assumption} \\
      &\ \text{assume }\gamma;\tau : \Phi;\Delta \To \Psi;\Gamma,
        s \sim b \in \glu{S}^0_{\;\Phi;\Delta} \\
      &\ \ltyequiv[\Phi][\Delta] 1{t[\gamma;\tau]\ s}{t'[\gamma;\tau]\ s}{T'}
        \tag{by congruence} \\
      &\ t'[\gamma;\tau]\ s \sim a(\gamma;\tau, b) \in \glu{T'}^0_{\;\Phi;\Delta}
        \byIH \\
      &\ t' \sim a \in \glu{S \func T'}^0_{\;\Psi;\Gamma}
    \end{align*}
  \end{itemize}
\end{proof}

Next we show that if $t$ and $a$ are related, then $t$ is equivalent to $a$'s
reification. %
This lemma needs to be proved mutually with two lemmas:
\begin{lemma}\labeledit{lem:st:ne-glue}
  If $\ltyping 0 v T$, then $v \sim \uparrow^T_{\;\Psi;\Gamma}(v) \in
  \glu{T}^0_{\;\Psi;\Gamma}$. 
\end{lemma}
\begin{lemma}\labeledit{lem:st:glue-nf}
  If $t \sim a \in \glu{T}^0_{\;\Psi;\Gamma}$, then
  $\ltyequiv 1{t}{\downarrow^T_{\;\Psi;\Gamma}(a)}{T}$. 
\end{lemma}
\begin{proof}[Proof of \Cref{lem:st:ne-glue}]
  We proceed by induction on $T$.
  \begin{itemize}[label=Case]
  \item $T= \Nat$, immediate.
    
  \item $T = S \func T'$, then assume $\gamma;\tau : \Phi;\Delta \To \Psi;\Gamma$ and $s \sim b \in
    \glu{S}^0_{\;\Phi;\Delta}$,
    \begin{align*}
      \uparrow^{S \func T'}_{\;\Psi;\Gamma}(v)(\gamma;\tau, b)
      &= \uparrow^{T'}_{\;\Phi;\Delta}(v[\gamma;\tau]\ \downarrow^S_{\;\Phi;\Delta}(b)) 
    \end{align*}
    By IH, we have
    \begin{align*}
      v[\gamma;\tau]\ \downarrow^S_{\;\Phi;\Delta}(b) \sim \uparrow^{T'}_{\;\Phi;\Delta}(v[\gamma;\tau]\
      \downarrow^S_{\;\Phi;\Delta}(b))
      \in \glu{T'}^0_{\;\Phi;\Delta}
    \end{align*}
    We further obtain from \Cref{lem:st:glue-resp-equiv,lem:st:glue-nf} (IH)
    \begin{align*}
      v[\gamma;\tau]\ s \sim \uparrow^{S \func T'}_{\;\Psi;\Gamma}(v)(\gamma;\tau, b)
      \in \glu{T'}^0_{\;\Phi;\Delta}
    \end{align*}
    Therefore $v \sim \uparrow^{S \func T'}_{\;\Psi;\Gamma}(v) \in \glu{S \func T'}^0_{\;\Psi;\Gamma}$.
  \end{itemize}
\end{proof}

\begin{proof}[Proof of \Cref{lem:st:glue-nf}]
  We proceed by induction on $T$.
  \begin{itemize}[label=Case]
  \item $T = \Nat$, immediate.

  \item $T = S \func T'$, then
    \begin{align*}
      H_2: &\ \forall \gamma;\tau : \Phi;\Delta \To \Psi;\Gamma,
        s \sim b \in \glu{S}^0_{\;\Phi;\Delta}. t[\gamma;\tau]\ s \sim a(\gamma;\tau,
             b) \in \glu{T'}^0_{\;\Phi;\Delta} \tag{by assumption} \\
           &\ x \sim \uparrow^S_{\Psi; (\Gamma, x : S)}(x) \in \intp{S}_{\Psi; (\Gamma, x : S)}
             \tag{by \Cref{lem:st:ne-glue} (IH)} \\
           &\ t[\id;p(\id)]\ x \sim a(\id;p(\id), \uparrow^S_{\Psi; (\Gamma, x : S)}(x)) \in
             \glu{T'}^0_{\Psi; (\Gamma, x : S)}
             \tag{by $H_2$ and let $\gamma;\tau$ be $\id;p(\id) : \Psi; (\Gamma, x
             : S) \To \Psi; \Gamma$} \\
           &\ \ltyequiv[\Psi][(\Gamma, x : S)] 1{t\ x}{\downarrow^{T'}_{\Psi;
             (\Gamma, x : S)}(a(\id;p(\id), \uparrow^S_{\Psi; (\Gamma, x
             : S)}(x)))}{T'}
             \byIH \\
           &\ \ltyequiv 1{\lambda x.t\ x}
             {\lambda x.\downarrow^{T'}_{\Psi;
             (\Gamma, x : S)}(a(\id;p(\id), \uparrow^S_{\Psi; (\Gamma, x
             : S)}(x)))}{S \func T'}
             \tag{by congruence} \\
           &\ \ltyequiv 1{t}
             {\lambda x.\downarrow^{T'}_{\Psi;
             (\Gamma, x : S)}(a(\id;p(\id), \uparrow^S_{\Psi; (\Gamma, x
             : S)}(x)))}{S \func T'}
             \tag{by $\eta$ and transitivity}
    \end{align*}
  \end{itemize}
\end{proof}

\subsubsection{Semantic Judgments at Layer $0$}

We are about to define the semantic judgments at layer $0$. %
Before that, we first generalize the gluing of types to
local contexts. %
Notice that this generalization is also layered and assumes $\iscore \Delta$. 
\begin{align*}
  \glu{\Delta}^0_{\;\Psi;\Gamma} &\subseteq (\dtyping \delta \Delta) \times \intp{\Delta}_{\;\Psi;\Gamma} \\
  \glu{\cdot}^0_{\;\Psi;\Gamma} &:= \{(\cdot, *)\} \\
  \glu{\Delta, x : T}^0_{\;\Psi;\Gamma} &:= \{ ((\delta, t/x), (\rho, a)) \sep \delta \sim \rho \in \glu{\Delta}^0_{\;\Psi;\Gamma}
                                  \tand t \sim a \in \glu{T}^0_{\;\Psi;\Gamma} \}
\end{align*}

Then monotonicity also generalizes to local contexts:
\begin{lemma}[Monotonicity]\labeledit{lem:st:glue-context-mon}
  If $\delta \sim \rho \in \glu{\Delta}^0_{\;\Psi;\Gamma}$, given $\gamma;\tau : \Phi;\Delta' \To
  \Psi;\Gamma$,  then
  $\delta[\gamma;\tau] \sim \rho[\gamma;\tau] \in \glu{\Delta}^0_{\;\Phi;\Delta'}$. 
\end{lemma}
\begin{proof}
  Immediate by induction on $\Delta$ and \Cref{lem:st:glue-mon}.
\end{proof}

Finally, we can define the judgments:
\begin{definition}
  We define the semantic judgment for typing and global substitutions:
  \begin{align*}
    \lSemtyp 0 t T &:= \forall \gamma : \Phi \To_g \Psi \tand \delta \sim \rho \in
                     \glu{\Gamma}^0_{\;\Phi;\Delta}. t[\gamma][\delta] \sim
                     \intp{t}^0_{\;\Phi;\Delta}(\gamma;\rho) \in
                     \glu{T}^0_{\;\Phi;\Delta}
  \end{align*}
  \begin{mathpar}
    \inferrule*
    {\iscore \Psi}
    {\Semtyp[\Psi]{\cdot}{\cdot}}

    \inferrule*
    {\Semtyp[\Psi]{\sigma}{\Phi} \\ \lSemtyp[\Psi][\cdot] 0 {t}{T}}
    {\Semtyp[\Psi]{\sigma, t/u}{\Phi, u : T}}
  \end{mathpar}
\end{definition}

We can recover the syntactic judgment for global substitutions from the semantic judgment:
\begin{lemma}
  If $\Semtyp[\Psi]{\sigma}{\Phi}$, then $\typing[\Psi]{\sigma}{\Phi}$.
\end{lemma}
\begin{proof}
  Induction. 
\end{proof}

\begin{lemma}[Monotonicity]
  If $\lSemtyp[\Psi] 0 {t}{T}$ and $\gamma : \Psi' \To_g \Psi$, then
  $\lSemtyp[\Psi'] 0 {t[\gamma]}{T}$.
\end{lemma}
\begin{proof}
  immediate by composition of global weakenings. 
\end{proof}

\begin{lemma}[Monotonicity]
  If $\Semtyp[\Psi]{\sigma}{\Phi}$ and $\gamma : \Psi' \To_g \Psi$, then
  $\Semtyp[\Psi']{\sigma[\gamma]}{\Phi}$.
\end{lemma}
\begin{proof}
  Induction and use monotonicity of the semantic typing at layer $0$.
\end{proof}

The semantic typing judgment is also split into two layers. %
The global substitution judgment is a generalization of the semantic judgment at layer
$0$, which will be used to define the semantic judgment at layer $1$. %
However, before doing that, we must first give the layer-$1$ gluing model.

\subsubsection{Layer-$1$ Gluing Model}

Unlike all previous layered definitions, where we can more or less let the definition
parameterized by the layer, in this layered gluing model, we can see that this 2-layered
modal type theory effectively operates on two different languages. %
Specifically our layer-$1$ model actually relies on the layer-$0$ semantic judgment,
and thus we have this incredibly long detour to arrive at the final definition. %
We refer to the layer-$0$ semantic judgment when we define the semantic for $\square$:
\begin{mathpar}
  \inferrule*
  {\ltyequiv 1 t{\boxit {t^c}}{\square T} \\ \lSemtyp[\Psi][\cdot] 0 {t^c} T}
  {t \sim {\boxit {t^c}} \in \square T_{\;\Psi;\Gamma}}

  \inferrule*
  {\ltyequiv 1 t v{\square T}}
  {t \sim v \in {\square T}_{\;\Psi;\Gamma}}
\end{mathpar}
There are two ways for $\square T$ to relate a term and a value. %
The first case is the more interesting one. %
When we know $t$ is related to $\boxit {t^c}$, we must be aware of the fact that it is
also a well-typed term in the semantics. %
This information is necessary when we prove the semantic rule for $\tletbox$. %
Now we give the gluing model at layer $1$:
\begin{align*}
  \glu{T}^1_{\;\Psi;\Gamma} &\subseteq \Exp \times \intp{T}_{\;\Psi;\Gamma} \\
  \glu{\Nat}^1_{\;\Psi;\Gamma} &:= \Nat_{\;\Psi;\Gamma} \\
  \glu{\square T}^1_{\;\Psi;\Gamma} &:= {\square T}_{\;\Psi;\Gamma} \\ 
  \glu{S \func T}^1_{\;\Psi;\Gamma} &:= \{(t, a) \sep \forall
                                      \gamma; \tau : \Phi; \Delta \To \Psi;\Gamma, s \sim b \in
                                      \glu{S}^1_{\;\Phi;\Delta}. t[\gamma; \tau]\ s \sim a(\gamma; \tau, b) \in
                                      \glu{T}^1_{\;\Phi;\Delta} \} 
\end{align*}
A few previous properties for the layer-$0$ gluing model continue to hold for the
layer-$1$ model,
e.g. \Cref{lem:st:glue-mon,lem:st:glue-resp-equiv,lem:st:ne-glue,lem:st:glue-nf}, with
layer changed from $0$ to $1$ appropriately. %
In fact, $\glu{T}^0$ is subsumed by $\glu{T}^1$:
\begin{lemma}\labeledit{lem:st:glu-0-to-1}
  If $\iscore T$, then $\glu{T}^0_{\;\Psi;\Gamma} = \glu{T}^1_{\;\Psi;\Gamma}$. 
\end{lemma}
\begin{proof}
  Induction on $\iscore T$.
\end{proof}
This lemma states that if we do not use $\square$ types, then our programs essentially
are just running in simply typed $\lambda$-calculus, so $\square$ types are the
\emph{only} addition we get by operating at layer $1$. 

With the gluing model, we can define the generalization of the gluing model to local
contexts and the semantic judgment:
\begin{align*}
  \glu{\Delta}^1_{\;\Psi;\Gamma} &\subseteq (\dtyping \delta \Delta) \times \intp{\Delta}_{\;\Psi;\Gamma} \\
  \glu{\cdot}^1_{\;\Psi;\Gamma} &:= \{(\cdot, *)\} \\
  \glu{\Delta, x : T}^1_{\;\Psi;\Gamma} &:= \{ ((\delta, t/x), (\rho, a)) \sep \delta \sim \rho \in \glu{\Delta}_{\;\Psi;\Gamma}
                                  \tand t \sim a \in \glu{T}^1_{\;\Psi;\Gamma} \}
\end{align*}

\begin{definition}
  \begin{align*}
    \lSemtyp 1 t T &:= \forall \Semtyp[\Phi]\sigma\Psi \tand \delta \sim \rho \in
                     \glu{\Gamma}^1_{\;\Phi;\Delta}. t[\sigma;\delta] \sim \intp{t}^1_{\;\Phi;\Delta}(\sigma;\rho) \in \glu{T}^1_{\;\Phi;\Delta}
  \end{align*}
\end{definition}
Notice that the semantic judgment at layer $1$ universally quantifies over a semantic
judgment of global substitutions. %
This quantification informs us that terms in the substitutions are semantically
well-typed at layer $0$, hence necessary for the rule of global variables.

\subsubsection{Semantic Rules}

We examine all the semantic judgments:

\begin{lemma}
  \begin{mathpar}
    \inferrule
    {\iscore \Psi \\ \iscore \Gamma}
    {\lSemtyp 0{\ze}{\Nat}}

    \inferrule
    {\iscore \Psi \\ \istype \Gamma}
    {\lSemtyp 1{\ze}{\Nat}}

    \inferrule
    {\lSemtyp i{t}{\Nat}}
    {\lSemtyp i{\su t}{\Nat}}
  \end{mathpar}
\end{lemma}
\begin{proof}
  Immediate.
\end{proof}

\begin{lemma}
  \begin{mathpar}
    \inferrule
    {\iscore \Psi \\ \iscore \Gamma \\ u : T \in \Psi}
    {\lSemtyp 0{u}{T}}
  \end{mathpar}

  \begin{mathpar}
    \inferrule
    {\iscore \Psi \\ \istype \Gamma \\ u : T \in \Psi}
    {\lSemtyp 1{u}{T}}
  \end{mathpar}
\end{lemma}
\begin{proof}
  We consider each rule by $i$.
  \begin{itemize}[label=Case]
  \item $i = 0$, then we assume $\gamma : \Phi \To_g \Psi$ and
    $\delta \sim \rho \in \glu{\Gamma}^0_{\;\Phi;\Delta}$. %
    We should prove
    \begin{align*}
      u[\gamma] \sim \intp{u}^0_{\;\Phi;\Delta}(\gamma; \rho) \in \glu{T}^0_{\;\Phi;\Delta}
    \end{align*}
    But we know
    \begin{align*}
      \intp{u}^0_{\;\Phi;\Delta}(\gamma; \rho)
      = \uparrow^{T}_{\;\Phi; \Delta}(u[\gamma])
    \end{align*}
    and thus the goal holds by \Cref{lem:st:ne-glue}.

  \item $i = 1$. then we assume $\Semtyp[\Phi]\sigma\Psi$ and
    $\delta \sim \rho \in \glu{\Gamma}^1_{\;\Phi;\Delta}$. %
    Then we should prove
    \begin{align*}
      \sigma(u) \sim \intp{\sigma(u)}^0_{\;\Phi;\Delta}(\id; *) \in
      \glu{T}^0_{\;\Phi;\Delta} = \glu{T}^1_{\;\Phi;\Delta}
      \tag{by \Cref{lem:st:glu-0-to-1}}
    \end{align*}
    By $\sigma(u)$'s typing judgment, we know it is closed
    \begin{align*}
      \sigma(u) \sim \intp{\sigma(u)}^1_{\;\Phi;\cdot}(\id; *) \in \glu{T}^0_{\;\Phi;\cdot}
    \end{align*}
    We get this from the semantic judgment $\lSemtyp[\Phi][\cdot] 0{\sigma(u)}{T}$.
  \end{itemize}  
\end{proof}

\begin{lemma}
  \begin{mathpar}
    \inferrule
    {\iscore \Psi \\ \iscore \Gamma \\ x : T \in \Gamma}
    {\lSemtyp{0}{x}{T}}

    \inferrule
    {\iscore \Psi \\ \istype \Gamma \\ x : T \in \Gamma}
    {\lSemtyp{1}{x}{T}}
  \end{mathpar}
\end{lemma}
\begin{proof}
  These two rules are immediate because they simply lookup the evaluation environment
  which maintains the gluing relation. 
\end{proof}

\begin{lemma}
  \begin{mathpar}  
    \inferrule
    {\lSemtyp[\Psi][\Gamma, x : S]{i}{t}{T}}
    {\lSemtyp i{\lambda x. t}{S \func T}}
  \end{mathpar}  
\end{lemma}
\begin{proof}
  We only consider the case where $i = 1$ because the case of $i = 0$ is the same. %
  Assume $\Semtyp[\Phi]\sigma\Psi$,
  $\delta \sim \rho \in \glu{\Gamma}^1_{\;\Phi;\Delta}$,
  $\gamma; \tau : \Phi'; \Delta' \To \Phi; \Delta$ and
  $s \sim b \in \glu{S}^1_{\;\Phi';\Delta'}$ %
  Then we should prove that
  \begin{align*}
    (\lambda x. (t[\sigma; (\delta[p(\id)], x/x)][\gamma; q(\tau)]))\ s \sim
    \intp{t}^1_{\;\Phi';\Delta'}(\sigma[\gamma];(\rho[\gamma;\tau], b)) \in \glu{T}^1_{\;\Phi';\Delta'}
  \end{align*}

  We need to construct $\glu{\Gamma, x : S}^1_{\;\Phi';\Delta'}$ from what we have. We
  claim
  \begin{align*}
    (\delta[\gamma;\tau], s/x) \sim (\rho[\gamma;\tau], b) \in \glu{\Gamma, x : S}^1_{\;\Phi';\Delta'}
  \end{align*}
  because $\delta[\gamma;\tau] \sim \rho[\gamma;\tau] \in \glu{\Gamma}^1_{\;\Phi';\Delta'}$
  by \Cref{lem:st:glue-context-mon}. Therefore
  \begin{align*}
    t[\sigma; (\delta[\gamma;\tau], s/x)] \sim \intp{t}^1_{\;\Phi';\Delta'}(\sigma;
    (\rho[\gamma;\tau], b)) \in \glu{T}^1_{\;\Phi';\Delta'}
  \end{align*}
  Comparing the goal and what we have, we use \Cref{lem:st:glue-resp-equiv} at layer
  $1$ to $\beta$ reduce the left hand side and achieve our goal.
\end{proof}

\begin{lemma}
  \begin{mathpar}
    \inferrule
    {\lSemtyp i{t}{S \func T} \\ \lSemtyp i{s}{S}}
    {\lSemtyp i{t\ s}{T}}
  \end{mathpar}
\end{lemma}
\begin{proof}
  Immediate by applying the premises.
\end{proof}

Finally, we consider the modal components. %
Recall that they only live at layer $1$.
\begin{lemma}
  \begin{mathpar}
    \inferrule
    {\istype \Gamma \\ \lSemtyp[\Psi][\cdot] 0 t T}
    {\lSemtyp{1}{\boxit t}{\square T}}
  \end{mathpar}
\end{lemma}
\begin{proof}
  Assume $\Semtyp[\Phi]\sigma\Psi$,
  $\delta \sim \rho \in \glu{\Gamma}_{\;\Phi;\Delta}$. %
  Then we should prove
  \begin{align*}
    \boxit {(t[\sigma])} \sim \boxit {(t[\sigma])} \in \glu{\square T}^1_{\;\Phi;\Delta}
  \end{align*}
  This holds by reflexivity and the definition of the gluing model. %
  Notice that $\lSemtyp[\Psi][\cdot] 0 t T$ is remembered by the gluing model of
  $\glu{\square T}$.
\end{proof}

\begin{lemma}  
  \begin{mathpar}
    \inferrule
    {\lSemtyp 1 {s}{\square T} \\ \lSemtyp[\Psi, u : T] 1 {t}{T'}}
    {\lSemtyp 1 {\letbox u s t} T'}
  \end{mathpar}
\end{lemma}
\begin{proof}
  Assume $\Semtyp[\Phi]\sigma\Psi$,
  $\delta \sim \rho \in \glu{\Gamma}_{\;\Phi;\Delta}$. %
  Then we know
  \begin{align*}
    s[\sigma;\delta] \sim \intp{s}^1_{\;\Phi;\Delta}(\sigma;\rho) \in \glu{\square T}^1_{\;\Phi;\Delta}
  \end{align*}
  There are two ways that this can hold:
  \begin{itemize}[label=Case]
  \item $\intp{s}^1_{\;\Phi;\Delta}(\sigma;\rho) = \boxit{t^c}$ and we know $\lSemtyp[\Psi][\cdot]
    0 {t^c} T$. %
    Then we have
    \begin{align*}
      & \Semtyp[\Phi]{\sigma, t^c/u}{\Psi, u : T} \\
      & t[\sigma,t^c/u;\delta] \sim \intp{t}^1_{\;\Phi;\Delta}((\sigma,t^c/u);\rho)
        \in \glu{T'}^1_{\;\Phi;\Delta}
        \tag{from premise} 
    \end{align*}
    We have the goal by \Cref{lem:st:glue-resp-equiv} and perform $\beta$ reduction on
    the left. 
    
  \item $\intp{s}^1_{\;\Phi;\Delta}(\sigma;\rho) = v$ for some neutral $v$ and we know
    $H1: \ltyequiv[\Phi][\Delta] 1{s[\sigma;\delta]}{v}{\square T}$. %
    Then
    \begin{align*}
      & u \sim \uparrow^T_{\;\Psi;u:T;\cdot}(u) \in \glu{T}^0_{\;\Psi;u:T;\cdot}
        \tag{by \Cref{lem:st:ne-glue}} \\
      & \lSemtyp[\Phi, u : T][\cdot] 0 {u} T
        \tag{by monotonicity} \\
      & \Semtyp[\Phi, u : T]{\sigma[p(\id)], u/u}{\Psi, u : T}
        \tag{by monotonicity} \\
      & t[(\sigma[p(\id)], u/u); \delta[p(\id);\id]]
        \sim \intp{t}^1_{\;\Phi, u : T;\Delta}((\sigma[p(\id)], u/u);
        \rho[p(\id);\id])
        \in \glu{T'}^1_{\;\Phi, u : T;\Delta}
        \tag{by premise} \\
      H2:\ & \ltyequiv[\Phi; u : T] [\Delta] 1 {t[(\sigma[p(\id)], u/u);
        \delta[p(\id);\id]]}
        {\downarrow^{T'}_{\;\Phi, u : T;\Delta}(\intp{t}^1_{\;\Phi, u : T;\Delta}((\sigma[p(\id)], u/u);
        \rho[p(\id);\id]))}{T'}
        \tag{by \Cref{lem:st:glue-nf}}
    \end{align*}
    We let $E := \letbox u v {\downarrow^{T'}_{\;\Phi, u : T;\Delta}(\intp{t}^1_{\;\Phi, u : T;\Delta}((\sigma[p(\id)], u/u);
      \rho[p(\id);\id]))}$ and notice that $E$ is neutral and well-typed. %
    By \Cref{lem:st:ne-glue}, we have
    \begin{align*}
      E \sim \uparrow^{T'}_{\;\Phi;\Delta}(E) \in \glu{T'}^1_{\;\Phi;\Delta}
    \end{align*}
    From $H1$ and $H2$, we have
    \begin{align*}
      \ltyequiv[\Phi][\Delta] 1 {\letbox u s t [\sigma;\delta]}{E}{T'}
    \end{align*}
    from which we obtain our goal via \Cref{lem:st:glue-resp-equiv}.
  \end{itemize}
  
\end{proof}

Now we have finished all the semantic rules.

\subsubsection{Fundamental Theorems and Soundness}

\begin{theorem}[Fundamental] $ $
  If $\ltyping i t T$, then $\lSemtyp i t T$.
\end{theorem}
\begin{proof}
  Induction.
\end{proof}

We would like to extract the soundness proof by setting $i = 1$. %
After that we need to also make sure we can eliminate the universal quantifications.
\begin{lemma}
  $\Semtyp[\Phi]\id\Phi$
\end{lemma}
\begin{proof}
  Use the semantic rule for global variables at layer $0$. 
\end{proof}

\begin{lemma}
  $\id \sim \uparrow^{\Gamma} \in \glu{\Gamma}^1_{\;\Psi;\Gamma}$
\end{lemma}
\begin{proof}
  Apply \Cref{lem:st:ne-glue} for local variables.
\end{proof}

This concludes our soundness proof.
\begin{proof}[Proof of \Cref{thm:st:sound}]
  From the fundamental theorems, we have $\lSemtyp i t T$, and then due to the
  previous two lemmas
  $t \sim \intp{t}^1_{\;\Psi;\Gamma}(\uparrow^{\Psi;\Gamma}) \in
  \glu{T}^1_{\;\Psi;\Gamma}$. %
  We get the goal by \Cref{lem:st:glue-nf}.
\end{proof}

At this point, we have shown the completeness and soundness of the NbE algorithm.

\section{Layered Contextual Modal Type Theory}\labeledit{sec:contextual}

In the previous sections, we have developed layered modal type theory and its
normalization by evaluation proof using a presheaf model. %
Though our semantics already supports certain code inspection operations, these
operations can only handle closed code. %
This problem has already been addressed by \emph{contextual types} proposed
by~\citet{nanevski_contextual_2008}. %
Instead of clearing out all local assumptions in the $\square$ introduction rule,
contextual types allow to specify local assumptions that can be depended on. %
To support full case analysis on code, we introduce contextual types to our 2-layered
modal type theory, forming layered contextual modal type theory. %
In this section, we adapt our previous development to contextual types and pattern
matching, and in the next section, we adapt our previous NbE proof to show that with
contextual types, the type theory remains normalizing.

\subsection{Examples}

As briefly discussed in the previous section, we can support an \tisapp construct,
which branches based on whether the received argument is a code of function
application or not. %
This construct can be implemented by a more general pattern matching construct, which
branches into multiple cases based on a code like this:
\begin{align*}
  & \matc t \\
  & |\; ?u~?u' \STo \cdots \\
  & |\; \su ?u \STo \cdots \\
  & |\; \_ \STo \cdots
\end{align*}
A pattern variable $u$ following a question mark captures a code in a pattern. %
This style might have received some influence from Ltac, the tactic language in Coq. %
In this example, the first pattern matches a function application. %
We use $u$ for the code of the function and $u'$ for the code of the argument. %
If $t$ is a code of function application, then we fall into this case and execute the
body, which potentially uses $u$ and $u'$. %
In the second branch, the pattern matches a successor and use $u$ to capture the code
of the predecessor. %
At last, we use a wildcard to capture all remaining cases, which includes cases like
$\ze$ or variable $x$. %

Depending on the type of $t$, we know ahead of time that certain branches are not
possible. %
For example, if $t$ has type $\square \Nat$, then we know it cannot be introduced by a
$\lambda$ abstraction. %
In this case, we cannot match it against a pattern of $\lambda$. %
Similarly, if $t$ is a code of a function type, then both $\ze$ and $\su$ cases are
not possible and therefore these cases cannot exist. %
Though it is possible to know ahead of time when an introduction is not possible,
elimination forms are always possible. %
Function applications can return any type as long as it is the return type of a
function. %

To confess, in the previous discussion, we omit one detail: if we allow pattern of the form
$\lambda x. ?u$, what should be the type of $u$? %
Clearly it is not closed because $u$ can at least refer to $x$. %
This is when contextual types becomes useful. %
If $t$ is a closed code of a function and of type $\square (S \func T)$. %
Then its function body behind a $\lambda$ abstraction should have type $\cont[x :
S]{T}$. %
The $x : S$ to the left of the turnstile denotes the open context the underlying term
is open with respect to. %
In general, this context can be arbitrarily long. %
Consider the following pattern matching, where $t$ has type $\cont{S \func T}$:
\begin{align*}
  & \matc t \\
  & |\; \lambda x. ?u \STo (\matc {\boxit u} \sep ?u'~?u'' \STo \cdots \sep \_ \STo \cdots) \\
  & |\; \_ \STo \cdots
\end{align*}
In this pattern matching, we first test whether $t$ is a $\lambda$ abstraction. %
If it is the case, we use $u$ to capture the code of the body, which has type
$\cont[\Gamma, x : S] T$. %
Then we use another pattern matching to test whether $u$ is a function application or
not. %
The context in the type grows due to $\lambda$, while function applications do not
introduce new binder, so the context of $u'$ and $u''$ does not grow. %

\subsection{Syntax and Typing Judgments}

Following the same principle as in~\Cref{sec:st}, we define the syntax of the type
theory as follows:
\begin{alignat*}{2}
  S, T &:=&&\ \Nat \sep \cont T \sep S \func T
  \tag{Types, \Typ} \\
  \Gamma, \Delta &:= &&\ \cdot \sep \Gamma, x : T
                        \tag{Local contexts}\\
  \delta &:=&& \cdot \sep \delta, t/x \tag{Local substitutions} \\
  s, t &:=&&\ x \sep u^\delta \tag{Terms, $\Exp$} \\
  & && \sep \ze \sep \su t
  \tag{natural numbers} \\
  & && \sep \boxit t \sep \letbox u s t \sep \matc t\ \vect{\branch}
       \tag{box}\\
  & &&\sep \lambda x. t \sep s\ t \tag{functions}  \\
  \branch & := &&\ \var x \STo t \sep \ze \STo t \sep \su{?u} \STo t \sep \lambda x. ?u \STo t
               \sep ?u~?u' \STo t
               \tag{Branches} \\
  \Phi, \Psi &:= &&\ \cdot \sep \Phi, u : (\judge T)
  \tag{Global contexts}
\end{alignat*}
Instead of $\square T$, we now use $\cont T$, contextual types, to represent the type
of code. %
The local context $\Gamma$ in $\cont T$ is the open variables that can be used to
construct a code of $T$. %
We let $\square T := \cont[\cdot] T$ to recover the previous $\square$ type. %
To introduce a contextual type, we use $\boxit t$ where $t$ is the potentially open
term. %
For elimination, in addition to $\tletbox$, we use the pattern matching construct that
we introduced in the previous example section. %
There are only five possible patterns in a pattern matching, each representing one
possible syntax in the core type theory. %
These five cases are exhaustive already due to the layered nature of our type theory:
it is not possible to refer to modal components at layer $0$, which is guaranteed by
our typing judgments. %
In reality, the branches of a well-typed pattern matching expression must be less than
five because we forbid incompatible introduction forms according to the type
information as explained in the previous example section. %
Finally, we do not support nested patterns like $\lambda x. ?u~?u'$ in our language. %
Fortunately, all nested patterns can be converted into nested pattern matchings as we
did in a previous example.

Due to the introduction of local contexts in the $\square$ types, our global contexts
deviate from the previous form, having to also keep track of a local context for each
binding. %
Because of this extra local context, we can no longer simply refer to a global
variable $u$. %
Instead, we must specify how this local context is filled in using a local
substitution $\delta$. %
When $u$ is being substituted, we also need to apply this local substitution to obtain
a well-typed term in the current context. %
Before capturing the intuitions formally using typing judgments, we also need two
predicates on types as in \Cref{sec:st} to classify types in the core type theory and
those in the extended type theory. %
Since our core type theory remains the same, the predicate $\iscore T$ is also
unchanged. %
When we write the generalization $\iscore \Phi$, we mean that all types occurring in
$\Phi$ including those in respective local contexts are core types. %
We only need to amend the predicate $\istype T$:
\begin{mathpar}
  \inferrule
  { }
  {\istype \Nat}

  \inferrule
  {\istype S \\ \istype T}
  {\istype{S \func T}}

  \inferrule
  {\iscore \Gamma \\ \iscore T}
  {\istype{\cont T}}
\end{mathpar}
In the case of contextual types, we require $\iscore \Gamma$, because $\cont T$
represents an open code in the \emph{core type theory}, so $\Gamma$ must not have
access to another contextual type. %
$\istype \Gamma$ is the natural generalization of $\istype T$, which is used at layer
$1$ to describe the local contexts. 

Next we give the typing judgments. %
There are four kinds of typing judgments in the contextual system. %
$\ltyping i t T$ is the usual typing judgment for terms. %
$\ltyping i \delta \Delta$ is the typing judgment for local substitutions. %
$\ltyping 1 \branch {\judge[\Delta] T \STo T'}$ is the typing judgment for pattern $\branch$
knowing that the type of the scrutinee is $\cont[\Delta] T$ and the return type is
$T'$. %
Last, $\ltyping 1{\vect\branch}{\judge[\Delta] T \STo T'}$ is a generalization of
$\ltyping 1 \branch {\judge[\Delta] T \STo T'}$ and the \emph{covering} judgment for all branches of a
pattern matching expression. %
We only show the modal rules because they are the only changed rules:
\begin{mathpar}
  \inferrule
  {\istype \Gamma \\ \ltyping[\Psi][\Delta] 0 t T}
  {\ltyping{1}{\boxit t}{\cont[\Delta] T}}

  \inferrule
  {\ltyping 1 {s}{\cont[\Delta] T} \\ \ltyping 1 {\vect \branch}{\judge[\Delta] T \STo T'}}
  {\ltyping 1 {\matc s\ \vect\branch} T'}

  \inferrule
  {\iscore \Psi \\ \iscore \Gamma}
  {\ltyping 0 {\cdot}{\cdot}}

  \inferrule
  {\iscore \Psi \\ \istype \Gamma}
  {\ltyping 1 {\cdot}{\cdot}}

  \inferrule
  {\ltyping i {\delta}{\Delta} \\ \ltyping i {t}{T}}
  {\ltyping i {\delta, t/x}{\Delta, x : T}}

  \inferrule
  {\ltyping i \delta \Delta \\ u : (\judge[\Delta] T) \in \Psi}
  {\ltyping{i}{u^\delta}{T}}
\end{mathpar}

In the rule for pattern matching, we need typing rules for all branches:
\begin{mathpar}
  \inferrule
  {\iscore \Delta \\ \ltyping 1 t T'}
  {\ltyping 1 {\var x \STo t}{\judge[\Delta] T \STo T'}}
  
  \inferrule
  {\iscore \Delta \\ \ltyping 1 t T'}
  {\ltyping 1 {\ze \STo t}{\judge[\Delta] \Nat \STo T'}}

  \inferrule
  {\ltyping[\Psi, u : (\judge[\Delta]\Nat)] 1 t T'}
  {\ltyping 1 {\su ?u \STo t}{\judge[\Delta] \Nat \STo T'}}

  \inferrule
  {\ltyping[\Psi, u : (\judge[\Delta, x : S]T)] 1 t T'}
  {\ltyping 1 {\lambda x. ?u \STo t}{\judge[\Delta] S \func T \STo T'}}

  \inferrule
  {\forall \iscore S.~ \ltyping[\Psi, u : (\judge[\Delta]S \func T), u' : (\judge[\Delta]S)] 1 t T'}
  {\ltyping 1 {?u~?u' \STo t}{\judge[\Delta] T \STo T'}}
\end{mathpar}
At last, $\ltyping 1 {\vect \branch}{\judge[\Delta] T \STo T'}$ includes \emph{all
  compatible} rules from the $\ltyping 1 \branch {\judge[\Delta] T \STo T'}$ judgment
above according to $T$. %
We write $\vect\branch(t)$ for the lookup of a branch. %
For example, $\vect\branch(t~s)$ returns a branch $?u~?u' \STo t'$ in $\vect\branch$ because
$t~s$ is a function application. %
We use the typing information to ensure this lookup is only performed legally.

Many rules just generalize their previous corresponding rules. %
In the introduction rule for $\square$, $t$ is typed with $\Delta$ instead of an empty
local context. %
When we refer to $u$, we need to specify a local substitution, so we must also move
the judgment for local substitutions up earlier than before. %
Notice that now we also distinguish local substitutions by layers, as opposed to
before, where local substitutions only live at layer $1$. %
When we refer to $u$ at layer $0$, $u$ effectively is a slot waiting to be substituted
by an open code at layer $0$. %
Therefore, the local substitution $\delta$ for the open variables must only contain
code at layer $0$ as well. %
When we refer to $u$ at layer $1$, we are evaluating or embedding the code represented
by $u$ in the extended type theory. %
Since we are operating at layer $1$, it does no harm to let the terms in $\delta$ to
also live at layer $1$. %
In this case, due to the confluence property of the system, contextual types behave no
differently from functions. %
In the special case where $\delta$ is the identity local substitution, we very often
omit it for brevity, i.e. $u^\id = u$. 

At last we have the elimination rule using pattern matching. %
After typing the scrutinee $s$, we also need to type all branches. %
What branches are possible depends on $T$. %
As explained before, there are impossible cases in $\vect\branch$ based
on $T$. %
The typing judgment requires $\vect\branch$ to contain all possible branches. %
For each branch, we proceed as expected, and push global variables to the global
context as needed when code is captured by pattern variables. %
The judgment for all branches $\ltyping 1 {\vect \branch}{\judge[\Delta] T \STo T'}$
is defined as:
\begin{mathpar}
  \inferrule
  {\forall b \in \vect \branch ~.~ \ltyping 1 {\branch}{\judge[\Delta] \Nat \STo T'}
    \\
    \branch_\ze = \ze \STo t \text{ for some } t \\
    \branch_\tsucc = \su ?u \STo t \text{ for some } t \\
    \branch_\tapp = ?u~?u' \STo t \text{ for some } t \\
    \forall x : \Nat \in \Delta ~.~ b_x = \var x \STo t \text{ for some } t \\
  \vect\branch \text{ is a permutation of } \{ \branch_\ze, \branch_\tsucc,
  \branch_\tapp, \branch_x \text{ for all } x : \Nat \in \Delta \}}
  {\ltyping 1{\vect\branch}{\judge[\Delta] \Nat \STo T'}}

  \inferrule
  {\forall b \in \vect \branch ~.~ \ltyping 1 {\branch}{\judge[\Delta] S \func T \STo T'}
    \\
    \branch_\lambda = \lambda x. ?u \STo t \text{ for some } t \\
    \branch_\tapp = ?u~?u' \STo t \text{ for some } t \\
    \forall x : S \func T \in \Delta ~.~ b_x = \var x \STo t \text{ for some } t \\
  \vect\branch \text{ is a permutation of } \{ \branch_\lambda,
  \branch_\tapp, \branch_x \text{ for all } x : S \func T \in \Delta \}}
  {\ltyping 1{\vect\branch}{\judge[\Delta] S \func T \STo T'}}
\end{mathpar}

There are two issues worth noting in the typing of pattern matching. %
The first one is the case of function applications. %
The premise is a universal quantification over all possible $\iscore S$. %
This is necessary due to the lack of expressive power of our simple type system. %
In a function application, the argument type is hidden and we can not learn it from
the return type. %
In order to make sure the body is well-typed for all argument types, we effectively
introduce certain parametricity in the meta-theory to ensure that $t$ is oblivious to
the concrete argument type $S$ given $\iscore S$. %
This treatment only exists due to simple types. %
With strong type systems, e.g. System F-style type variables or dependent types, this
issue can be handled by unification or other methods, c.f. \citet{Jang:POPL22}. %

The second problem is the lack of recursion on code. %
Indeed, the type theory we are presenting now cannot define a recursive function on
code structure. %
This is, again, because of our simple type system. %
Consider a possible recursive operation which inspects a code and recurses on the
code of the function if the scrutinee is a function application. %
This operation is not definable in our current type theory because the type of code
has changed. %
If the original scrutinee has type $\cont[\Delta] T$, then the recursive call is on a
term of type $\cont[\Delta]{S \func T}$ for some $S$. %
Nevertheless, the semantics is ready for recursion on code. %
With stronger types, this feature becomes possible, and indeed we will support it when
we adapt our layered modal type theory to dependent types. 

\subsection{Neutral Forms and Equivalence Rules}

Inheriting the previous development, our equivalence rules are also layered. %
The same as before, we have PER rules, congruence rules, $\beta$ rules and $\eta$
rules. %
The PER and $\eta$ rules stay the same so we omit them here. %
The congruence rules are simply derived from the typing judgments above, so we also
omit most of them here. %
The congruence rules for local substitutions are:
\begin{mathpar}
  \inferrule
  {\wf 0 \Psi \\ \wf i \Gamma}
  {\ltyequiv i {\cdot}{\cdot}{\cdot}}

  \inferrule
  {\ltyequiv i {\delta}{\delta'}{\Delta} \\ \ltyequiv i {t}{t'}{T}}
  {\ltyequiv i {\delta, t/x}{\delta', t'/x}{\Delta, x : T}}
\end{mathpar}
In this rule, the equivalence relation simply propagates down recursively. %
Yet, depending on $i$, this rule behaves differently. %
At layer $0$, since all terms are identified by their syntactic structures,
equivalence between local substitutions also implies identical syntax. %
However, at layer $1$, due to dynamics, we allow all terms inside to reduce or expand
freely.

The following are the $\beta$ rules for pattern matching on contextual types. %
The rules are very intuitive, as they simply look up the branches according to
the code and perform reduction accordingly:
\begin{mathpar}
  \inferrule
  {\ltyequiv i \delta{\delta'} \Delta \\ u : (\judge[\Delta] T) \in \Psi}
  {\ltyequiv{i}{u^\delta}{u^{\delta'}}{T}}
  
  \inferrule
  {x : T \in \Delta \\ \ltyping 1 {\vect \branch}{\judge[\Delta] T \STo T'} \\
    \vect\branch(x) = \var x \STo t}
  {\ltyequiv 1 {\matc {\boxit x}\ \vect\branch}{t} T'}  

  \inferrule
  {\ltyping 1 {\vect \branch}{\judge[\Delta] \Nat \STo T'} \\
    \vect\branch(\ze) = \ze \STo t}
  {\ltyequiv 1 {\matc {\boxit \ze}\ \vect\branch}{t} T'}  

  \inferrule
  {\ltyping[\Psi][\Delta] 0 s\Nat \\ \ltyping 1 {\vect \branch}{\judge[\Delta] \Nat \STo T'}
    \\ \vect\branch(\su s) = \su ?u \STo t}
  {\ltyequiv 1 {\matc {\boxit {(\su s)}}\ \vect\branch}{t[s/u]}{T'}}

  \inferrule
  {\ltyping[\Psi][\Delta, x : S] 0 s T \\ \ltyping 1 {\vect \branch}{\judge[\Delta] S \func T \STo
      T'} \\ \vect\branch(\lambda x. s) = \lambda x. ?u \STo t}
  {\ltyequiv 1 {\matc {\boxit {(\lambda x. s)}}\ \vect\branch}{t[s/u]}{T'}}

  \inferrule
  {\ltyping[\Psi][\Delta] 0 t{S \func T} \\ \ltyping[\Psi][\Delta] 0 s S \\
    \ltyping 1 {\vect \branch}{\judge[\Delta] T \STo T'} \\
    \vect\branch(t~s) = ?u~?u' \STo t}
  {\ltyequiv 1 {\matc {\boxit {(t~s)}}\ \vect\branch}{t[t/u, s/u']}{T'}}
\end{mathpar}
Here we rely on global substitutions which are about to be defined, but otherwise the
$\beta$ rules should follow intuition very closely. %
However, there is one tiny detail lying in this set of rules which might be difficult
to notice at the first glance: what if we match against a global variable, i.e. what
$\matc {\boxit{u^\delta}}\ \vect\branch$ should be reduced to?

This problem is raised by~\citet{kavvos_intensionality_2021}, where the author points
out that one example for the $\tisapp$ function by~\citet{gabbay_denotation_2013}
seems to break confluence:
\begin{align*}
  \letbox u {\boxit{(t~s)}} {\tisapp~(\boxit{u})}
\end{align*}
In our settings, we would use pattern matching to implement the $\tisapp$ function. %
Then \citet{kavvos_intensionality_2021} points out two possible equivalence chains:
\begin{align*}
  \letbox u {\boxit{(t~s)}} {\tisapp~(\boxit{u})}
  \approx \tisapp~(\boxit{(t~s)})
  \approx \tru
\end{align*}
There is another chain:
\begin{align*}
  \letbox u {\boxit{(t~s)}} {\tisapp~(\boxit{u})}
  \approx \letbox u {\boxit{(t~s)}}\fal
  \approx \fal
\end{align*}
The first equivalence makes more sense and is expected, so let us look into why the
second case goes wrong. %
It appears that in the second case, $\tisapp~(\boxit u)$ reduces to $\fal$, which
voids the later substitution of $t~s$ for $u$. %
Though there could be problems in~\citet{gabbay_denotation_2013} that exhibit this
issue, in our settings, $\tisapp~(\boxit{u}) \approx \fal$ would correspond to that
global variables have their own branch in a pattern matching, which returns \fal. %
This, however, is not the case in our system. %
In fact, in our system, $\matc {\boxit {u^\delta}} \ \vect\branch$ is blocked and
considered neutral. %
This choice is the same as considering $\matc x \ \vect\branch$ neutral, as they are both
waiting for a substitution to supply an actual input for the computation to continue. %
The only difference is that $\matc x \ \vect\branch$ is waiting for a local substitution
that could give $\tbox$ of some open code, while $\matc {\boxit {u^\delta}} \ \vect\branch$
is waiting for a global substitution that actually supplies the code for $u$. %
Therefore, by adding pattern matching as an elimination form, we in
fact need to change our definition of neutral forms in a slightly surprising way:
\begin{alignat*}{2}
    v &:= &&\ x \sep u^\theta \sep v\ w \sep \letbox u v w \sep \matc v \ \vect\nbranch \sep \matc {\boxit{u^\delta}} \ \vect\nbranch \tag{Neutral form
             ($\Ne$)} \\
  \nbranch &:= && \ \var x \STo w \sep \ze \STo w \sep \su{?u} \STo w \sep \lambda x. ?u \STo w
               \sep ?u~?u' \STo w
              \tag{Normal branches} \\
  \theta & := && \cdot \sep \theta, w/x
                 \tag{Normal local substitutions}
\end{alignat*}
The definition of normal form stays the same as before. %
This definition of neutral forms also implies that $\tletbox$ in fact cannot be
faithfully encoded by pattern matching. %
If we define
\begin{align*}
  \letbox u s t := \matc {s} \sep ?u \STo t
\end{align*}
namely letting bodies of all branches be $t$, this encoding will not behave the same as the
original $\tletbox$. %
When $s$ is some $\boxit{u'}$, the original $\tletbox$ will reduce and substitute $u'$
for $u$, while with this encoding, the evaluation will get stuck. %

Due to this difference in computational behavior, we let both elimination forms to
coexist. %
In our layered view of type theory, this is perfectly fine because our core theory is
unaffected. %
Coexistence only enriches the extended type theory. %
We use $\tletbox$ primarily for the execution and the composition of code and use
pattern matching for intensional analysis. %
Without layers, both $\tletbox$ and pattern matching are also available inside of
$\tbox$, which makes the behavior of the overall type theory overwhelmingly difficult
to control. %

\subsection{Full Set of Rules for Completeness}

In this section, we simply list the rest of the equivalence rules for the sake of
complete definition of the system. %
We have the following PER rules:
\begin{mathpar}
  \inferrule*
  {\ltyequiv i s t T}
  {\ltyequiv i t s T}

  \inferrule*
  {\ltyequiv i s t T \\ \ltyequiv i t u T}
  {\ltyequiv i s u T}

  \inferrule*
  {\ltyequiv i {\delta}{\delta'}\Delta}
  {\ltyequiv i {\delta'}{\delta} \Delta}

  \inferrule*
  {\ltyequiv[\Psi][\Gamma_0] i \delta \delta' {\Gamma_1} \\ \ltyequiv[\Psi][\Gamma_1] i {\delta'}{\delta''}{\Gamma_2}}
  {\ltyequiv[\Psi][\Gamma_0] i \delta {\delta''} {\Gamma_1}}
\end{mathpar}

Then we have the following congruence rules:
\begin{mathpar}
  \inferrule*
  {\wf 0 \Psi \\ \wf i \Gamma}
  {\ltyequiv i {\cdot}{\cdot}{\cdot}}

  \inferrule*
  {\ltyequiv i {\delta}{\delta'}{\Delta} \\ \ltyequiv i {t}{t'}{T}}
  {\ltyequiv i {\delta, t/x}{\delta', t'/x}{\Delta, x : T}}

  \inferrule*
  {\ltyequiv i \delta{\delta'} \Delta \\ u : (\judge[\Delta] T) \in \Psi}
  {\ltyequiv{i}{u^\delta}{u^{\delta'}}{T}}
  
  \inferrule*
  {\wf 0 \Psi \\ \wf i \Gamma \\  x : T \in \Gamma}
  {\ltyequiv i x x T}
  
  \inferrule*
  {\wf 0 \Psi \\ \wf i \Gamma}
  {\ltyequiv{i}{\ze}{\ze}{\Nat}}

  \inferrule*
  {\ltyequiv{i}{t}{t'}{\Nat}}
  {\ltyequiv{i}{\su t}{\su t'}{\Nat}}

  \inferrule
  {\ltyequiv i {s_0}{s_0'} T \\ \ltyequiv[\Psi][\Gamma, x : \Nat, y : T]i{s_1}{s_1'}T
    \\ \ltyequiv i t {t'}\Nat}
  {\ltyequiv i{\recn T{s_0}{x~y. s_1}t}{\recn T{s_0'}{x~y. s_1'}{t'}}T}

  \inferrule*
  {\ltyequiv[\Psi][\Gamma, x : S]{i}{t}{t'}{T}}
  {\ltyequiv{i}{\lambda x. t}{\lambda x. t'}{S \func T}}

  \inferrule*
  {\ltyequiv{i}{t}{t'}{S \func T} \\ \ltyequiv{i}{s}{s'}{S}}
  {\ltyequiv{i}{t\ s}{t'\ s'}{T}}

  \inferrule*
  {\wf 0 \Psi \\ \ltyequiv[\Psi][\cdot]{0}{t}{t'}T}
  {\ltyequiv{1}{\boxit t}{\boxit t'}{\square T}}

  \inferrule*
  {\ltyequiv 1 {s}{s'}{\square T} \\ \ltyequiv[\Psi, u : T] 1 {t}{t'}{T'}}
  {\ltyequiv 1 {\letbox u s t}{\letbox u{s'}{t'}} T'}
\end{mathpar}

\subsection{Reexamination of Syntactic Properties}

After establishing the 2-layered contextual modal type theory, we now reexamine the
useful properties of the judgments in \Cref{sec:st}. %
Most of these properties pleasantly remain stable in our contextual system. %
We simply enumerate them and sketch their proofs.

\begin{lemma}
  If $\iscore T$ ($\iscore \Gamma$, resp.), then $\istype T$ ($\istype \Gamma$, resp.).
\end{lemma}
\begin{proof}
  Induction.
\end{proof}

\begin{theorem}[Validity]\labeledit{thm:ct:validity}
  If $\ltyping i t T$, then $\iscore \Psi$ and
  \begin{itemize}
  \item if $i = 0$, then $\iscore \Gamma$ and $\iscore T$;
  \item if $i = 1$, then $\istype \Gamma$ and $\istype T$.
  \end{itemize}
\end{theorem}
\begin{theorem}[Validity]
  If $\ltyping i \delta \Delta$, then $\iscore \Psi$ and
  \begin{itemize}
  \item if $i = 0$, then $\iscore \Gamma$ and $\iscore \Delta$;
  \item if $i = 1$, then $\istype \Gamma$ and $\istype \Delta$.
  \end{itemize}
\end{theorem}
\begin{theorem}[Validity]
  If $\ltyping 1 \branch {\judge[\Delta] T \STo T'}$, then
  $\iscore \Psi$, $\istype \Gamma$, $\iscore \Delta$, $\iscore T$ and $\istype{T'}$. 
\end{theorem}
\begin{proof}
  Mutual induction.
\end{proof}

The lifting of layer from $0$ to $1$ is mutually proved between terms and local substitutions.
\begin{lemma}[Lifting] $ $
  \begin{itemize}
  \item If $\ltyping 0 t T$, then $\ltyping 1 t T$. 
  \item If $\ltyping 0 \delta \Delta$, then $\ltyping 1 \delta \Delta$. 
  \end{itemize}
\end{lemma}
\begin{proof}
  Both lifting lemmas are proved by mutual induction. 
\end{proof}
Similarly, the following lemmas are also proved mutually:
\begin{lemma} $ $
  \begin{itemize}
  \item If $\ltyequiv 0 t s T$, then $t = s$.
  \item If $\ltyequiv 0 {\delta}{\delta'} \Delta$, then $\delta = \delta'$. 
  \end{itemize}
\end{lemma}

\begin{lemma}[Presupposition]\labeledit{lem:ct:presup} $ $
  \begin{itemize}
  \item If $\ltyequiv i{t}{t'}T$, then $\ltyping i t T$ and $\ltyping i {t'}T$.
  \item If $\ltyequiv i{\delta}{\delta'}\Delta$, then $\ltyping i \delta {\Delta}$ and
    $\ltyping i {\delta'}\Delta$. 
  \item If $\ltyequiv 1 {\branch}{\branch'}{\judge[\Delta] T \STo T'}$, then$\ltyping 1 \branch
    {\judge[\Delta] T \STo T'}$ and $\ltyping 1{\branch'}{\judge[\Delta] T \STo T'}$.
  \end{itemize}
\end{lemma}
\begin{proof}
  Induction. Here we will need the substitution lemmas. 
\end{proof}

Next we define how local substitutions are applied to terms:
\begin{align*}
  x[\delta] &:= \delta(x) \tag{lookup $x$ in $\delta$} \\
  u^{\delta'}[\delta] &:= u^{\delta' \circ \delta}  \\
  \ze[\delta] &:= \ze \\
  \su t [\delta] &:= \su{(t[\delta])} \\
  \boxit t [\delta] &:= \boxit t \\
  \letbox u s t [\delta] &:= \letbox u{s[\delta]}{(t[\delta])}\\
  \matc s \ \vect\branch [\delta] &:= \matc {s[\delta]}\ (\vect\branch[\delta]) \\
  \lambda x. t [\delta] &:= \lambda x. (t[\delta, x/x]) \\
  t\ s [\delta] &:= (t[\delta]) \ (s[\delta]) \\[5pt]
  \cdot \circ \delta &:= \cdot \\
  (\delta', t/x) \circ \delta &:= (\delta' \circ \delta) , t[\delta]/x \\[5pt]
  \var x \STo t[\delta] &:= \var x \STo (t[\delta]) \\
  \ze \STo t [\delta] &:= \ze \STo (t[\delta]) \\
  \su{?u} \STo t [\delta] &:= \su{?u} \STo (t[\delta]) \\
  \lambda x. ?u \STo t [\delta] &:= \lambda x. ?u \STo (t[\delta]) \\
  ?u~?u' \STo t[\delta] &:= ?u~?u' \STo (t[\delta])
\end{align*}

The local substitution lemma is a bit more cumbersome to establish than in
\Cref{sec:st}, because now it is possible to apply local substitutions at layer $0$ in
the process of global substitutions (we will see how this happens in the next
section). %
Moreover, global variables store local substitutions, so the composition must also be
mutually defined at the same time. %
It becomes a bit more complex at layer $1$ where we need to go under pattern matching. %

Let us examine the properties layer by layer.
\begin{theorem}[Local substitution]
  If $\ltyping[\Psi][\Delta] 0 t T$ and $\ltyping 0 \delta \Delta$, then $\ltyping 0{t[\delta]} T$.
\end{theorem}
\begin{theorem}[Local composition]
  If $\ltyping[\Psi][\Delta] 0 {\delta'}{\Delta''}$ and
  $\ltyping[\Psi][\Delta'] 0 \delta \Delta$, then
  $\ltyping[\Psi][\Delta'] 0 {\delta' \circ \delta}{\Delta''}$.
\end{theorem}
\begin{proof}
  We do a mutual induction on $\ltyping[\Psi][\Delta] 0 t T$ and
  $\ltyping[\Psi][\Delta] 0 {\delta'}{\Delta''}$.

  It is easier to consider the local composition theorem first, where we simply invoke
  local substitution for all terms. %
  Now we consider the cases of $t$.
  \begin{itemize}[label=Case]
  \item $t = x$, and we know $x : T \in \Delta$. %
    Then the goal is immediate.
  \item $t = u^{\delta'}$ and $u : (\judge[\Delta'] T) \in \Phi$,
    $\ltyping[\Psi][\Delta]0{\delta'}{\Delta'}$. %
    By using the local composition theorem on $\delta'$ as the smaller argument, we
    have
    \begin{align*}
      \ltyping 0{\delta' \circ \delta}{\Delta'}
    \end{align*}
    which concludes the goal.
  \item $t = \ze$ and $t = \su t'$ are immediate.
    
  \item $t = \lambda x. t'$ and $t = t'~s$ are folklore and identical to the proof in
    \Cref{sec:st}. 
  \end{itemize}
\end{proof}

The case for layer $1$ is slightly more complex because we need to consider pattern
matching. %
Luckily, pattern matching on code only grows the global contexts so we simply
propagate the local substitutions recursively. %
\begin{theorem}[Local substitution]
  If $\ltyping[\Psi][\Delta] 1 t T$ and $\ltyping 1 \delta \Delta$, then $\ltyping 1{t[\delta]} T$.
\end{theorem}
\begin{theorem}[Local composition]
  If $\ltyping[\Psi][\Delta] 1 {\delta'}{\Delta''}$ and
  $\ltyping[\Psi][\Delta'] 1 \delta \Delta$, then
  $\ltyping[\Psi][\Delta'] 1 {\delta' \circ \delta}{\Delta''}$.
\end{theorem}
\begin{theorem}
  If $\ltyping[\Psi][\Delta] 1 \branch {\judge[\Delta'] T \STo T'}$ and
  $\ltyping 1 \delta \Delta$, then
  $\ltyping 1 {\branch[\delta]} {\judge[\Delta'] T \STo T'}$. 
\end{theorem}
\begin{theorem}
  If $\ltyping[\Psi][\Delta] 1 {\vect\branch} {\judge[\Delta'] T \STo T'}$ and
  $\ltyping 1 \delta \Delta$, then
  $\ltyping 1 {\vect\branch[\delta]} {\judge[\Delta'] T \STo T'}$. 
\end{theorem}
\begin{proof}
  We also do mutual induction here. %
  For the local substitutions, the extra cases are $\tbox$ and pattern matching. %
  The other cases are pretty much identical to the previous theorems. %
  The introduction rule of $\square$ is easy because we do not even propagate the
  local substitution. %
  The case of pattern matching is also simple because we just push the local
  substitutions downward and eventually hit the cases for branches. %
  \begin{itemize}[label=Case]
  \item $\branch = \var x \STo t$, then we know $\ltyping[\Psi][\Delta] 1 t T'$. %
    By IH, we have $\ltyping 1 {t[\delta]} T'$ which concludes the goal. 
  \item $\branch = \ze \STo t$, then the same argument as above applies. 
  \item $\branch = \su ?u \STo t$, then $\ltyping[\Psi, u :
    (\judge[\Delta']\Nat)][\Delta] 1 t T'$. %
    Since the local context does not change, we can still apply $\delta$ directly to
    $t$ by IH and obtain the goal. 
  \item $\branch = \lambda x. ?u \STo t$ and $\branch = ?u~?u' \STo t$ are similar to the
    previous case. 
  \end{itemize}
\end{proof}

\subsection{Global Simultaneous Substitutions}

The final missing piece in our syntactic development of this type theory is global
simultaneous substitutions. %
We simply adapt our previous work here and fix the differences.
\begin{align*}
  \sigma := \cdot \sep \sigma, t/u \tag{Global substitutions}
\end{align*}
\begin{mathpar}
  \inferrule*
  {\iscore \Psi}
  {\typing[\Psi]{\cdot}{\cdot}}

  \inferrule*
  {\typing[\Psi]{\sigma}{\Phi} \\ \ltyping 0 {t}{T}}
  {\typing[\Psi]{\sigma, t/u}{\Phi, u : (\judge T)}}
\end{mathpar}
The step case of typing rules now allow $t$ to be typed in a specified local context
according to the binding. %
The equivalence follows similarly.
\begin{mathpar}
  \inferrule*
  {\iscore \Psi}
  {\tyequiv[\Psi]{\cdot}{\cdot}{\cdot}}

  \inferrule*
  {\tyequiv[\Psi]{\sigma}{\sigma'}{\Phi} \\ \ltyequiv 0 {t}{t'}{T}}
  {\tyequiv[\Psi]{\sigma, t/u}{\sigma', t'/u}{\Phi, u : (\judge T)}}
\end{mathpar}

Next we define the application of global substitutions:
\begin{align*}
  x[\sigma] &:= x \\
  u^\delta[\sigma] &:= \sigma(u)[\delta[\sigma]] \tag{lookup $u$ in $\sigma$} \\
  \ze[\sigma] &:= \ze \\
  \su t [\sigma] &:= \su{(t[\sigma])} \\
  \boxit t [\sigma] &:= \boxit{(t[\sigma])} \\
  \letbox u s t[\sigma] &:= \letbox u {s[\sigma]}{(t[\sigma])} \\
  \matc s \ \vect\branch [\sigma] &:= \matc {s[\sigma]} \ (\vect\branch[\sigma]) \\
  \lambda x. t [\sigma] &:= \lambda x. (t[\sigma]) \\
  t~s [\sigma] &:= (t[\sigma])~(s[\sigma]) \\[5pt]
  \cdot[\sigma] &:= \cdot \\
  (\delta', t/x) [\sigma] &:= \delta'[\sigma] , t[\sigma]/x \\[5pt]
  \var x \STo t[\sigma] &:= \var x \STo (t[\sigma]) \\
  \ze \STo t [\sigma] &:= \ze \STo (t[\sigma]) \\
  \su{?u} \STo t [\sigma] &:= \su{?u} \STo (t[\sigma, u/u]) \\
  \lambda x. ?u \STo t [\sigma] &:= \lambda x. ?u \STo (t[\sigma, u/u]) \\
  ?u~?u' \STo t[\sigma] &:= ?u~?u' \STo (t[\sigma, u/u, u'/u'])
\end{align*}

Most cases proceed in expected ways. %
The case of global variables is probably the most interesting one. %
We first lookup the global substitution and obtain an open term $\sigma(u)$. %
In order to obtain a legal term in the current context, we must apply the local
substitution $\delta$. %
But before that, we need to recursively apply $\sigma$ to all terms in $\delta$ so
that $\delta$ lives in the correct global context. %
At last, we apply $\delta[\sigma]$ following the previous definition and obtain a term
in the right global and local contexts. %
This intuition can be justified by the following theorems:
\begin{theorem}[Global substitution] $ $
  \begin{itemize}
  \item If $\ltyping[\Phi] i t T$ and $\typing[\Psi]\sigma\Phi$, then $\ltyping
    i{t[\delta]} T$.
  \item If $\ltyping[\Phi] i {\delta}{\Delta}$ and
    $\typing[\Psi]\sigma\Phi$, then
    $\ltyping i {\delta[\sigma]}{\Delta}$.
  \item If $\ltyping[\Phi] 1 \branch {\judge[\Delta] T \STo T'}$ and
  $\typing[\Psi]\sigma\Phi$, then
  $\ltyping 1 {\branch[\delta]} {\judge[\Delta] T \STo T'}$. 
  \item If $\ltyping[\Phi] 1 {\vect\branch} {\judge[\Delta] T \STo T'}$ and
  $\typing[\Psi]\sigma\Phi$, then
  $\ltyping 1 {\vect\branch[\delta]} {\judge[\Delta] T \STo T'}$. 
  \end{itemize}
\end{theorem}
\begin{proof}
  We do mutual induction. %
  \begin{itemize}[label=Case]
  \item $t = u^\delta$
    \begin{mathpar}
      \inferrule
      {\ltyping[\Phi] i \delta \Delta \\ u : (\judge[\Delta] T) \in \Phi}
      {\ltyping[\Phi]{i}{u^\delta}{T}}
    \end{mathpar}
    \begin{align*}
      & \ltyping[\Psi][\Delta] 0 {\sigma(u)} T \tag{lookup of $\sigma$} \\
      & \ltyping[\Phi] i{\delta[\sigma]} \Delta
        \byIH \\
      & \ltyping 0 {\sigma(u)[\delta[\sigma]]} T
        \tag{by local substitution}
    \end{align*}
    
  \item $t = \boxit{t'}$
    \begin{mathpar}
      \inferrule
      {\istype \Gamma \\ \ltyping[\Phi][\Delta] 0{t'}T'}
      {\ltyping[\Phi]{1}{\boxit {t'}}{\cont[\Delta]{T'}}}
    \end{mathpar}
    \begin{align*}
      & \ltyping[\Psi][\Delta] 0{t'[\sigma]}T'
        \byIH \\
      & \ltyping{1}{\boxit {t'[\sigma]}}{\cont[\Delta]{T'}}
    \end{align*}
    
  \item $t = \matc s \ \vect\branch$, then we simply push $\sigma$ inwards by IHs.
    
  \item $\branch = \var x \STo t$,
    \begin{mathpar}
      \inferrule
      {\iscore \Delta \\ x : T \in \Delta \\ \ltyping[\Phi] 1 t T'}
      {\ltyping[\Phi] 1 {\var x \STo t}{\judge[\Delta] T \STo T'}}
    \end{mathpar}
    \begin{align*}
      & \ltyping 1{t[\sigma]} T'
        \byIH \\
      & \ltyping 1 {\var x \STo (t[\sigma])}{\judge[\Delta] T \STo T'}
    \end{align*}
    
  \item $\branch = \lambda x. ?u \STo t$, 
    \begin{mathpar}
      \inferrule
      {\ltyping[\Phi, u : (\judge[\Delta, x : S]T)] 1 t T'}
      {\ltyping[\Phi] 1 {\lambda x. ?u \STo t}{\judge[\Delta] S \func T \STo T'}}
    \end{mathpar}
    \begin{align*}
      & \typing[\Psi, u : (\judge[\Delta, x : S]T)]{\sigma,u/u}{\Phi, u :
        (\judge[\Delta, x : S]T)}
      \\
      & \ltyping[\Psi, u : (\judge[\Delta, x : S]T)] 1{t[\sigma,u/u]}T'
        \byIH \\
      & \ltyping 1 {\lambda x. ?u \STo (t[\sigma,u/u])}{\judge[\Delta] S \func T \STo T'}
    \end{align*}
    
  \item $\branch = ?u~?u' \STo t$,
    \begin{mathpar}
      \inferrule
      {\forall \iscore S.~ \ltyping[\Phi, u : (\judge[\Delta]S \func T), u' : (\judge[\Delta]S)] 1 t T'}
      {\ltyping[\Phi] 1 {?u~?u' \STo t}{\judge[\Delta] T \STo T'}}
    \end{mathpar}
    \begin{align*}
      & \typing[\Psi, u : (\judge[\Delta]S \func T), u' :
        (\judge[\Delta]S)]{\sigma,u/u,u'/u'}{\Phi, u : (\judge[\Delta]S \func T), u' :
        (\judge[\Delta]S)}
      \\
      & \ltyping[\Psi, u : (\judge[\Delta]S \func T), u' : (\judge[\Delta]S)]
        1{t[\sigma,u/u, u'/u']}T'
        \byIH \\
      & \ltyping 1 {?u~?u' \STo (t[\sigma,u/u,u'/u'])}{\judge[\Delta] T \STo T'}
    \end{align*}
  \end{itemize}
\end{proof}

We could also define the substitutions between dual contexts, but they are define in
an identical way to the development in \Cref{sec:st}, so we omit the details here for
brevity. %
In the next section, we move on to the adaptation of the presheaf model, and show that
the 2-layered contextual modal type theory is normalizing by evaluation.

\section{Presheaf Model for Contextual Types}\labeledit{sec:prescont}

Having developed the syntactic theory of the 2-layered contextual modal type theory,
the next step following the previous development is to show its normalization by using
a presheaf model. %
The adaptation is largely moderate. %
We only need to pay more attention to the interpretation of pattern matching and make
adjustments to the models in the completeness and soundness proofs.

\subsection{Weakenings and Presheaves}

Due to the significant changes in global contexts, our definition of global weakenings
should also be adjusted:
\begin{mathpar}
  \inferrule
  { }
  {\varepsilon: \cdot \To_g \cdot}

  \inferrule
  {\gamma : \Psi \To_g \Phi \\ \iscore \Gamma \\ \iscore T}
  {q(\gamma) : \Psi, u : (\judge T) \To_g \Phi, u : (\judge T)}

  \inferrule
  {\gamma : \Psi \To_g \Phi \\ \iscore \Gamma \\ \iscore T}
  {p(\gamma) : \Psi, u : (\judge T) \To_g \Phi}
\end{mathpar}

We only add the extra premises of $\iscore \Gamma$. %
We still refer to the category of dual contexts and weakenings as $\WC$. 

Next we move on to fix the interpretations:
\begin{align*}
  \intp{\_} &: \Typ \to \WC^{op} \To \SetC \\
  \intp{\Nat} &:= \Nf^\Nat \\
  \intp{\cont T} &:= \Nf^{\cont T} \\
  \intp{S \func T} &:= \intp{S} \hfunc \intp{T}
\end{align*}
Again the fix is moderate. %
The interpretation of a contextual type is the presheaf to the set of normal forms of
that contextual type. %
Since $\Nf^T$ in general is a presheaf, so all properties about this interpretation
remain true. %
The interpretations of global, local and dual contexts remain the same, as well as
their properties including functoriality. %
The definitions of reification and reflection only require minimal case to handle
contextual types instead:
\begin{align*}
  \downarrow^T &: \intp{T} \To \Nf^T \\
  \downarrow^{\cont[\Delta] T}_{\;\Psi;\Gamma}(a)
               &:= a \\
  \uparrow^T &: \Ne\ T \To \intp{T} \\
  \uparrow^{\cont[\Delta] T}_{\;\Psi;\Gamma}(v)
               &:= v
\end{align*}
Additionally, we need to a version of reification for local substitutions, because our
neutral form $u^\theta$ requires $\theta$ to be the normal local substitutions,
i.e. to only contain normal forms.
\begin{align*}
  \downarrow^{\cdot}_{\;\Psi;\Gamma}(*)
  &:= \cdot \\
  \downarrow^{\Delta, x : T}_{\;\Psi;\Gamma}(\rho, a)
  &:= \downarrow^{\Delta}_{\;\Psi;\Gamma}(\rho), \downarrow^{T}_{\;\Psi;\Gamma}(a)/x
\end{align*}

\begin{lemma}\labeledit{lem:ct:reif-lsubst}
  With $\rho \in \intp{\Delta}_{\;\Psi; \Gamma}$, let
  $\theta := \uparrow^\Delta_{\;\Psi; \Gamma}(\rho)$. %
  Then we have $\ltyping 1 {\theta}\Delta$ and $\theta$ is a normal local substitution.
\end{lemma}
\begin{proof}
  We do induction on $\Delta$. %
  Note that $\downarrow^{T}_{\;\Psi;\Gamma}(a)$ always returns normal forms.
\end{proof}

Finally, we fix the interpretation of terms as natural transformations.
\begin{align*}
  \intp{\_}
  &: \forall i \to \ltyping[\Phi][\Delta] i t T \to \intp{\Phi;\Delta}^i \To \intp{T} \\
  \intp{u^\delta}^0_{\;\Psi;\Gamma}(\gamma; \rho)
  &:= \uparrow^{T}_{\;\Psi; \Gamma}(u[\gamma]^{\theta})
    \tag{where $\ltyping[\Phi][\Delta]0\delta{\Delta'}$ and $\theta := \downarrow^{\Delta'}_{\;\Psi; \Gamma}(\intp{\delta}^0_{\;\Psi;
    \Gamma}(\gamma; \rho))$}\\
  \intp{u^\delta}^1_{\;\Psi;\Gamma}(\sigma; \rho)
  &:= \intp{\sigma(u)}^0_{\;\Psi; \Gamma}(\id; \intp{\delta}^1_{\;\Psi;\Gamma}(\sigma;
    \rho)) \\
  \intp{\boxit t}^1_{\;\Psi;\Gamma}(\sigma; \rho)
  &:= \boxit {(t[\sigma])} \\
  \intp{\matc{t}\ \vect\branch}^1_{\;\Psi;\Gamma}(\sigma; \rho)
  &:= \tmatc(s, \vect\branch)_{\;\Psi; \Gamma}(\sigma; \rho)
    \tag{if $\intp{t}^1_{\;\Psi;\Gamma}(\sigma; \rho) = \boxit s$} \\
  \intp{\matc{t}\ \vect\branch}^1_{\;\Psi;\Gamma}(\sigma; \rho)
  &:= \uparrow^{T}_{\;\Psi; \Gamma}(\matc v\ \vect\nbranch)
    \tag{if $\intp{t}^1_{\;\Psi;\Gamma}(\sigma; \rho) = v : \square S$
    and $\vect\nbranch :=\tnfbranch(\vect\branch)_{\;\Psi; \Gamma}(\sigma;\rho)$}
  \\[5pt]
  \intp{\_}
  &: \forall i \to \ltyping[\Phi][\Delta] i \delta \Delta' \to \intp{\Phi;\Delta}^i
    \To \intp{\Delta'} \\
  \intp{\cdot}^i_{\;\Psi;\Gamma}(\_)
  &:= * \\
  \intp{\delta, t/x}^i_{\;\Psi;\Gamma}(\sigma;\rho)
  &:= (\intp{\delta}^i_{\;\Psi;\Gamma}(\sigma;\rho), \intp{t}^i_{\;\Psi;\Gamma}(\sigma;\rho))
  \\[5pt]
  \tmatc
  &: \ltyping[\Psi][\Delta'] 0 t T \to
    \ltyping[\Phi][\Delta] 1{\vect\branch}{\judge[\Delta'] T \STo T'} \to
    \intp{\Phi;\Delta}^1_{\;\Psi; \Gamma} \to \intp{T}_{\;\Psi; \Gamma} \\
  \tmatc(x, \vect\branch)_{\;\Psi; \Gamma}(\sigma; \rho)
  &:=
    \intp{t}^1_{\;\Psi; \Gamma}(\sigma; \rho)
    \tag{where $\var x \STo t := \vect\branch(x)$} \\
  \tmatc(\ze, \vect\branch)_{\;\Psi; \Gamma}(\sigma; \rho)
  &:=
    \intp{t}^1_{\;\Psi; \Gamma}(\sigma; \rho)
    \tag{where $\ze \STo t := \vect\branch(\ze)$} \\
  \tmatc(\su s, \vect\branch)_{\;\Psi; \Gamma}(\sigma; \rho)
  &:=
    \intp{t}^1_{\;\Psi; \Gamma}(\sigma, s/u; \rho)
    \tag{where $\su ?u \STo t := \vect\branch(\su s)$} \\
  \tmatc(\lambda x. s, \vect\branch)_{\;\Psi; \Gamma}(\sigma; \rho)
  &:=
    \intp{t}^1_{\;\Psi; \Gamma}(\sigma, s/u; \rho)
    \tag{where $\lambda x. ?u \STo t := \vect\branch(\lambda x. s)$} \\
  \tmatc(t'~s, \vect\branch)_{\;\Psi; \Gamma}(\sigma; \rho)
  &:=
    \intp{t}^1_{\;\Psi; \Gamma}(\sigma, t'/u, s/u; \rho)
    \tag{where $?u~?u' \STo t := \vect\branch(t'~s)$} \\
  \tmatc(u^\delta, \vect\branch)_{\;\Psi; \Gamma}(\sigma; \rho)
  &:= \uparrow^{T'}_{\;\Psi; \Gamma}(\matc {\boxit{u^\delta}}\ \vect\nbranch)
    \tag{where $\vect\nbranch :=\tnfbranch(\vect\branch)_{\;\Psi; \Gamma}(\sigma;\rho)$} \\[5pt]
  \tnfbranch &: \ltyping[\Phi][\Delta] 1{\branch}{\judge[\Delta'] T \STo T'} \to
               \intp{\Phi;\Delta}^1_{\;\Psi; \Gamma} \to
               \ltyping 1{\nbranch}{\judge[\Delta'] T \STo T'} \\
  \tnfbranch(\var x \STo t)_{\;\Psi; \Gamma}(\sigma; \rho)
  &:= \var x \STo \downarrow^{T'}_{\;\Psi; \Gamma}(\intp{t}^1_{\;\Psi; \Gamma}(\sigma; \rho)) \\
  \tnfbranch(\ze \STo t)_{\;\Psi; \Gamma}(\sigma; \rho)
  &:= \ze \STo \downarrow^{T'}_{\;\Psi; \Gamma}(\intp{t}^1_{\;\Psi; \Gamma}(\sigma;
    \rho)) \\
  \tnfbranch(\su ?u \STo t)_{\;\Psi; \Gamma}(\sigma; \rho)
  &:= \su ?u \STo \downarrow^{T'}_{\;\Psi, u: (\judge[\Delta']\Nat);
    \Gamma}(\intp{t}^1_{\;\Psi, u: (\judge[\Delta']\Nat); \Gamma}(\sigma', u/u;
    \rho'))
    \tag{where $p(\id); \id : \Psi, u: (\judge[\Delta']\Nat);\Gamma \To \Psi;\Gamma$ and $(\sigma'; \rho')
    := (\sigma; \rho)[p(\id); \id]$} \\
  \tnfbranch(\lambda x. ?u \STo t)_{\;\Psi; \Gamma}(\sigma; \rho)
  &:= \lambda x. ?u \STo \downarrow^{T'}_{\;\Psi, u: (\judge[\Delta', x : S]T);
    \Gamma}(\intp{t}^1_{\;\Psi, u: (\judge[\Delta', x : S]T); \Gamma}(\sigma', u/u;
    \rho'))
    \tag{where $S$ is the type of $x$, $p(\id); \id : \Psi, u:
    (\judge[\Delta', x : S]T);\Gamma \To \Psi;\Gamma$ and $(\sigma'; \rho') := (\sigma; \rho)[p(\id); \id]$} \\
  \tnfbranch(?u~?u' \STo t)_{\;\Psi; \Gamma}(\sigma; \rho)
  &:= ?u~?u' \STo \downarrow^{T'}_{\;\Psi, u: (\judge[\Delta']{S \func T'}), u':
    (\judge[\Delta']S); \Gamma}(\intp{t}^1_{\;\Psi, u: (\judge[\Delta']{S \func T'}),
    u': (\judge[\Delta']S); \Gamma}(\sigma', u/u, u'/u'; \rho')) 
    \tag{where $S$ is a parameterized core type, }  \\
  \tag{$p(p(\id)); \id : \Psi, u: (\judge[\Delta']{S \func T'}), u': (\judge[\Delta']S);\Gamma \To \Psi;\Gamma$ and
    $(\sigma'; \rho') := (\sigma; \rho)[p(p(\id)); \id]$}
\end{align*}

The interpretation of terms has grown substantially more complex than the previous
development. %
We only write down the modal cases because the other cases are folklore and remain
unchanged. %
Let us break down what is changed after the introduction of contextual types and
pattern matching by the following lemma. %
\begin{lemma}[Well-definedness] $ $
  \begin{itemize}
  \item If $\ltyping[\Phi][\Delta] i t T$ and $(\sigma;\rho) \in \intp{\Phi;\Delta}^i_{\;\Psi; \Gamma}$,
    then
    $\intp{t}^i_{\;\Psi; \Gamma}(\sigma;\rho) \in \intp{T}_{\;\Psi; \Gamma}$.
  \item If $\ltyping[\Phi][\Delta] i \delta \Delta'$ and
    $(\sigma;\rho) \in \intp{\Phi;\Delta}^i_{\;\Psi; \Gamma}$, then
    $\intp{\delta}^i_{\;\Psi; \Gamma}(\sigma;\rho) \in \intp{\Delta'}_{\;\Psi; \Gamma}$.
  \end{itemize}
\end{lemma}
\begin{proof}
  
  \begin{itemize}[label=Case]
  \item $t = u^\delta$ and $i = 0$,
    first, due to contextual types, all global variables must keep track of local
    substitutions.
    \begin{mathpar}
      \inferrule
      {\ltyping[\Phi][\Delta] 0 \delta \Delta' \\ u : (\judge[\Delta'] T) \in \Phi}
      {\ltyping[\Phi][\Delta]{0}{u^\delta}{T}}
    \end{mathpar}

    Since we are at layer $0$, we know $\gamma : \Psi \To_g \Phi$, hence
    \begin{align*}
      & \intp{\delta}^0_{\;\Psi; \Gamma}(\gamma;\rho) \in \intp{\Delta'}_{\;\Psi;
        \Gamma}
        \byIH \\
      & \ltyping 0 {\downarrow^{\Delta'}_{\;\Psi; \Gamma}(\intp{\delta}^0_{\;\Psi;
        \Gamma}(\gamma;\rho))} \Delta'
        \tag{by \Cref{lem:ct:reif-lsubst}} \\
      & \ltyping 0 {u[\gamma]^\theta} T
        \tag{where ${u[\gamma]^\theta}$ is neutral} \\
      & \uparrow^T_{\;\Psi; \Gamma}({u[\gamma]^\theta}) \in \intp{T}_{\;\Psi; \Gamma}
    \end{align*}
    
  \item $t = u^\delta$ and $i = 1$,
    \begin{mathpar}
      \inferrule
      {\ltyping[\Phi][\Delta] 1 \delta \Delta' \\ u : (\judge[\Delta'] T) \in \Phi}
      {\ltyping[\Phi][\Delta]{1}{u^\delta}{T}}
    \end{mathpar}
    Similar to before, we decrease the layer from $1$ to $0$ to perform further
    interpretation of terms. %
    However, the recursive interpretation is done in a non-empty environment.
    \begin{align*}
      & \intp{\delta}^1_{\;\Psi; \Gamma}(\gamma;\rho) \in \intp{\Delta'}_{\;\Psi;
        \Gamma}
        \byIH \\
      & \ltyping[\Psi][\Delta'] 0{\sigma(u)}T \\
      & (\id; \intp{\delta}^1_{\;\Psi; \Gamma}(\gamma;\rho)) \in \intp{\Psi; \Delta'}^0_{\;\Psi;
        \Gamma} \\
      & \intp{\sigma(u)}^0_{\;\Psi; \Gamma}(\id; \intp{\delta}^1_{\;\Psi;\Gamma}(\sigma;
    \rho)) \in \intp{T}_{\;\Psi; \Gamma}
    \end{align*}
    
  \item $t = \boxit t'$, we do not have to check much as the global substitution
    property directly applies.
    
  \item $t = \matc s \ \vect\branch$, this case is rather simple by following the types
    of $\tmatc$ and $\tnfbranch$ functions. 
  \end{itemize}
\end{proof}

Now we have finished the adjustments to the presheaf model. %
We restate the definition of the NbE algorithm and its completeness and soundness
theorems. %
Next, we move on to fixing the proofs of various properties, as well as the
completeness and soundness proofs.
\begin{definition}
  A normalization by evaluation algorithm given $\ltyping 1 t T$ is
  \begin{align*}
    \nbe^T_{\;\Psi;\Gamma}(t) &:= \downarrow^T_{\;\Psi;\Gamma} (\intp{t}^1_{\;\Psi;\Gamma}(\uparrow^{\Psi;\Gamma}))
  \end{align*}
\end{definition}

\begin{theorem}[Completeness]\labeledit{thm:ct:compl}
  If $\ltyequiv 1 t {t'} T$, then $\nbe^T_{\;\Psi;\Gamma}(t) = \nbe^T_{\;\Psi;\Gamma}(t')$.
\end{theorem}
\begin{theorem}[Soundness]\labeledit{thm:ct:sound}
  If $\ltyping 1 t T$, then $\ltyequiv 1 t {\nbe^T_{\;\Psi;\Gamma}(t)} T$.
\end{theorem}

\subsection{Properties of Interpretations}

Following the steps before, we reestablish several important properties that are
useful for later proofs. %
\begin{lemma}
  $\uparrow^{\Delta'}$ is natural, i.e. for
  $\gamma; \tau : \Phi; \Delta \To \Psi;\Gamma$, we have
  $\uparrow^{\Delta'}_{\;\Psi;\Gamma}(\rho)[\gamma; \tau] =
  \uparrow^{\Delta'}_{\;\Phi;\Delta}(\rho[\gamma; \tau])$. 
\end{lemma}
\begin{proof}
  This is a natural consequence of the naturality of $\uparrow^T$. 
\end{proof}

\begin{lemma}\labeledit{lem:ct:gsubst-nat-gen} $ $
  \begin{itemize}
  \item If $\ltyping[\Phi][\Delta] 0 t T$, $\gamma' : \Psi' \To_g \Psi$,
    $\gamma : \Psi \To_g \Phi$ and $\rho \in \intp{\Delta}_{\;\Psi';\Gamma}$, then
    $\intp{t}^0_{\;\Psi';\Gamma}(\gamma \circ \gamma';\rho) =
    \intp{t[\gamma]}^0_{\;\Psi';\Gamma}(\gamma';\rho)$.
    
  \item If $\ltyping[\Phi][\Delta] 0 \delta \Delta'$, $\gamma' : \Psi' \To_g \Psi$,
    $\gamma : \Psi \To_g \Phi$ and $\rho \in \intp{\Delta}_{\;\Psi';\Gamma}$, then
    $\intp{\delta}^0_{\;\Psi';\Gamma}(\gamma \circ \gamma';\rho) =
    \intp{\delta[\gamma]}^0_{\;\Psi';\Gamma}(\gamma';\rho)$.
  \end{itemize}
\end{lemma}
\begin{proof}
  Only the case of global variables is worth looking into. %
  If $t = u^\delta$ and $\ltyping[\Phi][\Delta] 0 \delta \Delta'$. 
  \begin{align*}
    \intp{u^\delta}^0_{\;\Psi';\Gamma}(\gamma \circ \gamma';\rho)
    &= \uparrow_{\;\Psi';\Gamma}(u[\gamma \circ \gamma']^\theta)
      \tag{where $\theta := \downarrow^{\Delta'}_{\;\Psi';\Gamma}(\intp{\delta}^0_{\;\Psi';\Gamma}(\gamma \circ \gamma';\rho))$} \\
    &= \uparrow_{\;\Psi';\Gamma}(u[\gamma][\gamma']^\theta)
      \tag{notice $\theta =
      \downarrow^{\Delta'}_{\;\Psi';\Gamma}(\intp{\delta[\gamma]}^0_{\;\Psi';\Gamma}(\gamma';\rho))$
      by IH} \\
    &= \intp{u[\gamma]^{\delta[\gamma]}}^0_{\;\Psi';\Gamma}(\gamma';\rho) \\
    &= \intp{u^\delta[\gamma]}^0_{\;\Psi';\Gamma}(\gamma';\rho)
  \end{align*}
\end{proof}

\begin{corollary}\labeledit{lem:ct:gsubst-nat}
  If $\ltyping[\Phi][\Delta] 0 t T$, $\gamma : \Psi \To_g \Phi$ and
  $\rho \in \intp{\Delta}_{\;\Psi;\Gamma}$, then
  $\intp{t}^0_{\;\Psi;\Gamma}(\gamma;\rho) = \intp{t[\gamma]}^0_{\;\Psi;\Gamma}(\id;\rho)$.
\end{corollary}

Now we reestablish the naturality of interpretations:
\begin{lemma}[Naturality]\labeledit{lem:st:t-intp-nat} $ $
  \begin{itemize}
  \item If $\ltyping[\Phi][\Delta] i t T$, $\gamma;\tau : \Psi';\Gamma' \To \Psi;\Gamma$ and
    $\sigma;\rho \in \intp{\Phi;\Delta}^i_{\;\Psi;\Gamma}$, then
    $\intp{t}^i_{\;\Psi;\Gamma}(\sigma;\rho)[\gamma;\tau] =
    \intp{t}^i_{\;\Psi';\Gamma'}((\sigma;\rho)[\gamma;\tau])$.
  \item If $\ltyping[\Phi][\Delta] i \delta \Delta'$, $\gamma;\tau : \Psi';\Gamma' \To \Psi;\Gamma$ and
    $\sigma;\rho \in \intp{\Phi;\Delta}^i_{\;\Psi;\Gamma}$, then
    $\intp{\delta}^i_{\;\Psi;\Gamma}(\sigma;\rho)[\gamma;\tau] =
    \intp{\delta}^i_{\;\Psi';\Gamma'}((\sigma;\rho)[\gamma;\tau])$.
  \item If $\ltyping[\Psi][\Delta'] 0 t T$ and
    $\ltyping[\Phi][\Delta] 1{\vect\branch}{\judge[\Delta'] T \STo T'}$,
    $\sigma;\rho \in \intp{\Phi;\Delta}^1_{\;\Psi;\Gamma}$, and
    $\gamma;\tau : \Psi';\Gamma' \To \Psi;\Gamma$, then
    $\tmatc(t, \vect\branch)_{\;\Psi;\Gamma}(\sigma; \rho)[\gamma; \tau]
    = \tmatc(t[\gamma], \vect\branch)_{\;\Psi';\Gamma'}((\sigma; \rho)[\gamma; \tau])$.
  \item If $\ltyping[\Phi][\Delta] 1{\branch}{\judge[\Delta'] T \STo T'}$,
    $\sigma;\rho \in \intp{\Phi;\Delta}^1_{\;\Psi;\Gamma}$, and
    $\gamma;\tau : \Psi';\Gamma' \To \Psi;\Gamma$, then
    $\tnfbranch(\branch)_{\;\Psi;\Gamma}(\sigma; \rho)[\gamma; \tau] =
    \tnfbranch(\branch)_{\;\Psi';\Gamma'}((\sigma; \rho)[\gamma; \tau])$. 
  \end{itemize}
\end{lemma}
The statement of naturality is much more verbose than before because of both
contextual types and pattern matching on code.
\begin{proof}
  We do mutual induction on the first statements of all four statements. %
  \begin{itemize}[label=Case]
  \item $i = 0$ and $t = u^\delta$ and $\ltyping[\Phi][\Delta] 0 \delta \Delta'$. %
    In this case,
    \begin{align*}
      \intp{u^\delta}^0_{\;\Psi;\Gamma}(\gamma';\rho)[\gamma;\tau]
      &= \uparrow_{\;\Psi;\Gamma}(u[\gamma']^\theta)[\gamma;\tau]
        \tag{where $\theta := \downarrow^{\Delta'}_{\;\Psi;\Gamma}(\intp{\delta}^0_{\;\Psi;\Gamma}(\gamma';\rho))$}\\
      &= \uparrow_{\;\Psi';\Gamma'}(u[\gamma']^\theta[\gamma;\tau])
        \tag{naturality of reflection} \\
      &= \uparrow_{\;\Psi';\Gamma'}(u[\gamma'][\gamma]^{\theta[\gamma;\tau]}) \\
      &= \uparrow_{\;\Psi';\Gamma'}(u[\gamma' \circ \gamma]^{\theta'})
        \tag{where $\theta' :=
        \downarrow^{\Delta'}_{\;\Psi';\Gamma'}(\intp{\delta}^0_{\;\Psi;\Gamma}((\gamma';\rho)[\gamma;\tau]))
        = \theta[\gamma; \tau]$ by IH} \\
      &= \intp{u^\delta}^0_{\;\Psi';\Gamma'}((\gamma';\rho)[\gamma;\tau])
    \end{align*}

  \item $i = 1$ and $t = u$ and $\ltyping[\Phi][\Delta] 0 \delta \Delta'$. %
    In this case,
    \begin{align*}
      \intp{u^\delta}^1_{\;\Psi;\Gamma}(\sigma;\rho)[\gamma;\tau]
      &= \intp{\sigma(u)}^0_{\;\Psi;\Gamma}(\id; \intp{\delta}^1_{\;\Psi;\Gamma}(\sigma;\rho))[\gamma;\tau] \\
      &= \intp{\sigma(u)}^0_{\;\Psi;\Gamma}((\id; \intp{\delta}^1_{\;\Psi;\Gamma}(\sigma;\rho))[\gamma;\tau]) \\
        \byIH \\
      &= \intp{\sigma(u)}^0_{\;\Psi;\Gamma}(\gamma; \intp{\delta}^1_{\;\Psi;\Gamma}(\sigma;\rho)[\gamma;\tau]) \\
      &= \intp{\sigma(u)}^0_{\;\Psi;\Gamma}(\gamma;
        \intp{\delta}^1_{\;\Psi;\Gamma}((\sigma;\rho)[\gamma;\tau]))
      \tag{by IH on $\delta$} \\
      &= \intp{\sigma(u)[\gamma]}^0_{\;\Psi;\Gamma}(\id;
        \intp{\delta}^1_{\;\Psi;\Gamma}((\sigma;\rho)[\gamma;\tau]))
      \tag{by \Cref{lem:ct:gsubst-nat}} \\
      &= \intp{u^\delta}^1_{\;\Psi';\Gamma'}((\sigma;\rho)[\gamma;\tau])
    \end{align*}
    
  \item
    $i = 1$ and $t = \boxit t'$. %
    This case is unchanged.
    
  \item 
    $i = 1$ and $t = \matc {t'} \ \vect\branch$ and $\intp{t'}^1_{\;\Psi;\Gamma}(\sigma;
    \rho) = v$. %
    Let $\vect\nbranch := \tnfbranch(\vect\branch)_{\;\Psi; \Gamma}(\sigma;\rho)$. 
    By IH, we have
    \begin{align*}
      \vect\nbranch[\gamma; \tau] &= \tnfbranch(\vect\branch)_{\;\Psi; \Gamma}(\sigma;\rho)[\gamma; \tau]
      = \tnfbranch(\vect\branch)_{\;\Psi'; \Gamma'}((\sigma;\rho)[\gamma; \tau]) \\
      v[\gamma; \tau] &= \intp{t'}^1_{\;\Psi;\Gamma}(\sigma;
      \rho)[\gamma; \tau] = \intp{t'}^1_{\;\Psi';\Gamma'}((\sigma;
      \rho)[\gamma; \tau])
    \end{align*}
    Therefore, due to
    \begin{align*}
      \uparrow^T_{\;\Phi;\Gamma}(\matc v \ \vect\nbranch)[\gamma; \tau]
      &= \uparrow^T_{\;\Phi;\Gamma}(\matc {v[\gamma; \tau]} \ (\vect\nbranch[\gamma;
        \tau]))
        \tag{by naturality of reflection}
    \end{align*}
    we conclude the goal. 
    
  \item 
    $i = 1$ and $t = \matc {t'} \ \vect\branch$ and $\intp{t'}^1_{\;\Psi;\Gamma}(\sigma;
    \rho) = \boxit s$. %
    Notice that by IH, we have
    \begin{align*}
      \boxit {(s[\gamma])} = \boxit s[\gamma; \tau]
      = \intp{t'}^1_{\;\Psi;\Gamma}(\sigma; \rho)[\gamma; \tau]
      = \intp{t'}^1_{\;\Psi';\Gamma'}((\sigma; \rho)[\gamma; \tau])
    \end{align*}
    We use invoke IH to conclude
    \begin{align*}
      \tmatc(s, \vect\branch)_{\;\Psi;\Gamma}(\sigma; \rho)[\gamma; \tau]
      = \tmatc(s[\gamma], \vect\branch)_{\;\Psi;\Gamma}((\sigma; \rho)[\gamma; \tau])
    \end{align*}
    Effectively this requires the naturality of $\tmatc$ to be proved, which we will work on next.
  \end{itemize}

  For the naturality of $\tmatc$, we know that $\ltyping[\Psi][\Delta'] 0 t T$. %
  By doing induction, we know that there are only cases of terms in the core type
  theory and the case of global variables. %
  Effectively we are relying on the exhaustiveness of pattern matching on code here.
  \begin{itemize}
  \item $t = x$ and $\vect\branch(x) = \var x \STo t'$,
    \begin{mathpar}
      \inferrule
      {\iscore \Psi \\ \iscore{\Delta'} \\ x : T \in \Delta'}
      {\ltyping[\Psi][\Delta']{0}{x}{T}}
    \end{mathpar}
    Then $x[\gamma] = x$, then the goal is effectively to prove
    \begin{align*}
      \intp{t'}^1_{\;\Psi;\Gamma}(\sigma; \rho)[\gamma; \tau]
      = \intp{t'}^1_{\;\Psi';\Gamma'}(\sigma; \rho[\gamma; \tau])
    \end{align*}
    which can be discharged by IH.
    
  \item $t = \ze$ and $\vect\branch(\ze) = \ze \STo t'$,
    \begin{mathpar}
      \inferrule
      {\iscore \Psi \\ \iscore{\Delta'}}
      {\ltyping[\Psi][\Delta']{0}{\ze}{\Nat}}
    \end{mathpar}
    This case is the same as above.
    
  \item $t = \su s$ and $\vect\branch(\su s) = \su ?u \STo t'$,
    \begin{mathpar}
      \inferrule
      {\ltyping[\Psi][\Delta']{0}{s}{\Nat}}
      {\ltyping[\Psi][\Delta']{0}{\su s}{\Nat}}
    \end{mathpar}
    Then we prove
    \begin{align*}
      \intp{t'}^1_{\;\Psi;\Gamma}(\sigma, s/u; \rho)[\gamma; \tau]
      = \intp{t'}^1_{\;\Psi';\Gamma'}(\sigma, s/u; \rho[\gamma; \tau])
    \end{align*}
    which can be discharged by IH.
    
  \item $t = \lambda x. s$ and $\vect\branch(\lambda x. s) = \lambda x. ?u \STo t'$,
    \begin{mathpar}
      \inferrule
      {\ltyping[\Psi][\Delta', x : S]{0}{s}{T}}
      {\ltyping[\Psi][\Delta']{0}{\lambda x. s}{S \func T}}
    \end{mathpar}
    This case is the same as above. %
    The extension of local context does not affect the argument.
    
  \item $t = s~s'$ and $\vect\branch(s~s') = ?u~?u' \STo t'$ and for some $\iscore S$,
    \begin{mathpar}
      \inferrule
      {\ltyping[\Psi][\Delta']{0}{s}{S \func T} \\ \ltyping[\Psi][\Delta']{0}{s'}{S}}
      {\ltyping[\Psi][\Delta']{0}{s\ s'}{T}}
    \end{mathpar}
    This case is still the same expect that there are two terms inserted to the global
    substitution.
    
  \item $t = u^\delta$,
    \begin{mathpar}
      \inferrule
      {\ltyping[\Psi][\Delta'] 0 \delta {\Delta''} \\ u : (\judge[\Delta''] T) \in \Psi}
      {\ltyping[\Psi][\Delta']{0}{u^\delta}{T}}
    \end{mathpar}
    In this case, we do not look up $\vect\branch$ at all. %
    Instead, we just construct a neutral form and use reflection to convert it to a
    natural transformation. %
    Let $\vect\nbranch := \tnfbranch(\vect\branch)_{\;\Psi; \Gamma}(\sigma;\rho)$. 
    By IH, we have
    \begin{align*}
      \vect\nbranch[\gamma; \tau] &= \tnfbranch(\vect\branch)_{\;\Psi; \Gamma}(\sigma;\rho)[\gamma; \tau]
      = \tnfbranch(\vect\branch)_{\;\Psi'; \Gamma'}((\sigma;\rho)[\gamma; \tau]) \\
      \boxit{u^\delta}[\gamma; \tau] &= \boxit{u[\gamma]^{\delta[\gamma]}}
    \end{align*}
    Therefore, the goal is concluded by the naturality of reflection and the two
    equations above. %
  \end{itemize}

  At last, we have the naturality of $\tnfbranch$ to prove. %
  This in fact is very similar to the neutral case in the $\tletbox$ expression in our
  previous development. %
  Here we only pick the $\lambda$ case for demonstration. %
  If $\branch = \lambda x. ?u \STo t$, 
  \begin{mathpar}
    \inferrule
    {\ltyping[\Phi, u : (\judge[\Delta', x : S]T)] 1 t T'}
    {\ltyping[\Phi][\Delta] 1 {\lambda x. ?u \STo t}{\judge[\Delta'] S \func T \STo T'}}
  \end{mathpar}
  Then the goal effectively becomes
  \begin{align*}
    & \downarrow^{T'}_{\;\Psi, u : (\judge[\Delta', x : S]T);\Gamma}(\intp{t}^1_{\;\Psi,
      u : (\judge[\Delta', x : S]T);\Gamma}(\sigma[p(\id)], u/u; \rho[p(\id);
      \id]))[q(\gamma); \tau] \\
    =~& \downarrow^{T'}_{\;\Psi', u : (\judge[\Delta', x :
        S]T);\Gamma'}(\intp{t}^1_{\;\Psi', u : (\judge[\Delta', x :
        S]T);\Gamma'}((\sigma[p(\id)], u/u; \rho[p(\id); \id])[q(\gamma); \tau]))
  \end{align*}
  This equation is further due to
  \begin{align*}
    \intp{t}^1_{\;\Psi, u : (\judge[\Delta', x : S]T);\Gamma}(\sigma[p(\id)], u/u;
    \rho[p(\id); \id])[q(\gamma);\tau] = \intp{t}^1_{\;\Psi', u : (\judge[\Delta', x : S]T);\Gamma'}((\sigma[p(\id)], u/u;
    \rho[p(\id); \id])[q(\gamma); \tau])
  \end{align*}
  This is proved by IH on $t$. 
\end{proof}
This lemma justifies the claim that the interpretations of terms and local
substitutions return natural transformations. %

\begin{lemma}\labeledit{lem:ct:intp-lsubst-0} $ $
  \begin{itemize}
  \item If $\ltyping[\Phi][\Delta'] 0 t T$, $\ltyping[\Phi][\Delta] 0 \delta \Delta'$,
    $\gamma : \Psi \To_g \Phi$ and $\rho \in \intp{\Delta}_{\;\Psi;\Gamma}$, then
    $\intp{t[\delta]}^0_{\;\Psi;\Gamma}(\gamma; \rho) =
    \intp{t}^0_{\;\Psi;\Gamma}(\gamma; \intp{\delta}^0_{\;\Psi;\Gamma}(\gamma;
    \rho))$.
  \item If $\ltyping[\Phi][\Delta'] 0{\delta'}{\Delta''}$, $\ltyping[\Phi][\Delta] 0 \delta \Delta'$,
    $\gamma : \Psi \To_g \Phi$ and $\rho \in \intp{\Delta}_{\;\Psi;\Gamma}$, then
    $\intp{\delta' \circ \delta}^0_{\;\Psi;\Gamma}(\gamma; \rho) =
    \intp{\delta'}^0_{\;\Psi;\Gamma}(\gamma; \intp{\delta}^0_{\;\Psi;\Gamma}(\gamma;
    \rho))$.
  \end{itemize}
\end{lemma}
\begin{proof}
  We perform mutual induction.
  \begin{itemize}
  \item $t = x$, $t = \ze$, $t = \su t'$, and $t = s~s'$, immediate with possible use
    of IHs.

  \item $t = \lambda x. s$,
    \begin{mathpar}
      \inferrule
      {\ltyping[\Psi][\Delta', x : S]{0}{s}{T}}
      {\ltyping[\Psi][\Delta']{0}{\lambda x. s}{S \func T}}
    \end{mathpar}
    We compute
    \begin{align*}
      \intp{\lambda x. s[\delta]}^0_{\;\Psi;\Gamma}(\gamma; \rho)
      &= \intp{\lambda x. (s[\delta[\id;p(\id)], x/x])}^0_{\;\Psi;\Gamma}(\gamma; \rho)
      \\
      &= (\gamma'; \tau' : \Psi';\Gamma' \To \Psi; \Gamma) (a) \mapsto
        \intp{s[\delta[\id;p(\id)], x/x]}^0_{\;\Psi;\Gamma}(\gamma \circ \gamma';
        (\rho[\gamma'; \tau'], a)) \\
      &= (\gamma'; \tau') (a) \mapsto
        \intp{s}^0_{\;\Psi;\Gamma}(\gamma \circ \gamma';
        \intp{\delta[\id;p(\id)], x/x}^0_{\;\Psi;\Gamma}(\gamma \circ \gamma';
        (\rho[\gamma'; \tau'], a)))
        \byIH \\
      &= (\gamma'; \tau') (a) \mapsto
        \intp{s}^0_{\;\Psi;\Gamma}(\gamma \circ \gamma';
        (\intp{\delta}^0_{\;\Psi;\Gamma}(\gamma \circ \gamma';\rho[\gamma'; \tau']),
        a)) \\
      &= (\gamma'; \tau') (a) \mapsto
        \intp{s}^0_{\;\Psi;\Gamma}(\gamma \circ \gamma';
        (\intp{\delta}^0_{\;\Psi;\Gamma}(\gamma;\rho)[\gamma'; \tau'], a))
        \tag{by naturality of $\intp{\delta}^0$} \\
      &= \intp{\lambda x. s}^0_{\;\Psi;\Gamma}(\gamma;
        \intp{\delta}^0_{\;\Psi;\Gamma}(\gamma; \rho)) 
    \end{align*}

  \item $t = u^{\delta'}$,
    \begin{mathpar}
      \inferrule
      {\ltyping[\Psi][\Delta'] 0{\delta'}{\Delta''} \\ u : (\judge[\Delta''] T) \in \Psi}
      {\ltyping[\Psi][\Delta']{0}{u^\delta}{T}}
    \end{mathpar}
    We compute as follows
    \begin{align*}
      \intp{u^{\delta'}[\delta]}^0_{\;\Psi;\Gamma}(\gamma; \rho)
      &= \intp{u^{\delta' \circ \delta}}^0_{\;\Psi;\Gamma}(\gamma; \rho) \\
      &= \uparrow^T_{\;\Psi;\Gamma}(u[\gamma]^\theta)
        \tag{where $\theta := \downarrow^{\Delta''}_{\;\Psi;\Gamma}(\intp{\delta' \circ \delta}^0_{\;\Psi;\Gamma}(\gamma; \rho))$}
    \end{align*}
    Notice that
    \begin{align*}
      \intp{\delta' \circ \delta}^0_{\;\Psi;\Gamma}(\gamma; \rho)
      &= \intp{\delta'}^0_{\;\Psi;\Gamma}(\gamma;
        \intp{\delta}^0_{\;\Psi;\Gamma}(\gamma; \rho))
        \byIH
    \end{align*}
    Moreover,
    \begin{align*}
      \intp{u^{\delta'}}^0_{\;\Psi;\Gamma}(\gamma;
      \intp{\delta}^0_{\;\Psi;\Gamma}(\gamma; \rho))
      &= \uparrow^T_{\;\Psi;\Gamma}(u[\gamma]^{\theta'})
        \tag{where $\theta' := \downarrow^{\Delta''}_{\;\Psi;\Gamma}(\intp{\delta'}^0_{\;\Psi;\Gamma}(\gamma;
        \intp{\delta}^0_{\;\Psi;\Gamma}(\gamma; \rho)))$}
    \end{align*}
    The goal is seen by having $\theta = \theta'$. 
  \end{itemize}
\end{proof}

\begin{lemma}\labeledit{lem:ct:intp-0-1} $ $
  \begin{itemize}
  \item If $\ltyping[\Phi][\Delta] 0 t T$, $\typing[\Psi]\sigma\Phi$, $\gamma : \Psi' \To_g \Psi$ and $\rho \in
    \intp{\Delta}_{\;\Psi';\Gamma}$, then
    $\intp{t[\sigma]}^0_{\;\Psi';\Gamma}(\gamma; \rho) =
    \intp{t}^1_{\;\Psi';\Gamma}(\sigma[\gamma]; \rho)$.
  \item If $\ltyping[\Phi][\Delta] 0 \delta \Delta'$, $\typing[\Psi]\sigma\Phi$, $\gamma : \Psi' \To_g \Psi$ and $\rho \in
    \intp{\Delta}_{\;\Psi';\Gamma}$, then
    $\intp{\delta[\sigma]}^0_{\;\Psi';\Gamma}(\gamma; \rho) =
    \intp{\delta}^1_{\;\Psi';\Gamma}(\sigma[\gamma]; \rho)$.
  \end{itemize}
\end{lemma}
\begin{proof}
  This lemma connects interpretations at both layers. %
  Since $t$ is from layer $0$, we only concern ourselves with the case of global variables.
  
  If $t = u^\delta$, then
  \begin{align*}
    \intp{u^\delta[\sigma]}^0_{\;\Psi';\Gamma}(\gamma; \rho)
    &= \intp{\sigma(u)[\delta[\sigma]]}^0_{\;\Psi';\Gamma}(\gamma; \rho) \\
    &= \intp{\sigma(u)[\delta[\sigma]][\gamma]}^0_{\;\Psi';\Gamma}(\id; \rho)
      \tag{by \Cref{lem:ct:gsubst-nat}} \\
    &= \intp{\sigma(u)[\gamma][\delta[\sigma][\gamma]]}^0_{\;\Psi';\Gamma}(\id; \rho)
    \\
    &= \intp{\sigma(u)[\gamma]}^0_{\;\Psi';\Gamma}(\id;
      \intp{\delta[\sigma][\gamma]}^0_{\;\Psi';\Gamma}(\id; \rho))
      \tag{by \Cref{lem:ct:intp-lsubst-0}} \\
    &= \intp{\sigma(u)[\gamma]}^0_{\;\Psi';\Gamma}(\id;
      \intp{\delta}^1_{\;\Psi';\Gamma}(\sigma[\gamma]; \rho))
      \byIH \\
    &= \intp{u^\delta}^0_{\;\Psi';\Gamma}(\sigma[\gamma]; \rho)
  \end{align*}
\end{proof}

\begin{lemma}\labeledit{lem:ct:intp-lsubst-1} $ $
  \begin{itemize}
  \item If $\ltyping[\Phi][\Delta'] 1 t T$, $\ltyping[\Phi][\Delta] 1 \delta \Delta'$,
    $\typing[\Psi]{\sigma}\Phi$ and $\rho \in \intp{\Delta}_{\;\Psi;\Gamma}$, then
    $\intp{t[\delta]}^1_{\;\Psi;\Gamma}(\sigma; \rho) =
    \intp{t}^1_{\;\Psi;\Gamma}(\sigma; \intp{\delta}^1_{\;\Psi;\Gamma}(\sigma;
    \rho))$.

  \item If $\ltyping[\Phi][\Delta'] 1{\delta'}{\Delta''}$, $\ltyping[\Phi][\Delta] 1 \delta \Delta'$,
    $\typing[\Psi]{\sigma}\Phi$ and $\rho \in \intp{\Delta}_{\;\Psi;\Gamma}$, then
    $\intp{\delta' \circ \delta}^1_{\;\Psi;\Gamma}(\sigma; \rho) =
    \intp{\delta'}^1_{\;\Psi;\Gamma}(\sigma; \intp{\delta}^1_{\;\Psi;\Gamma}(\sigma;
    \rho))$.

  \item If $\ltyping[\Psi][\Gamma'] 0 t T'$ and
    $\ltyping[\Phi][\Delta'] 1{\vect\branch}{\judge[\Gamma']{T'} \STo T}$,
    $\ltyping[\Phi][\Delta] 1 \delta \Delta'$, $\typing[\Psi]{\sigma}\Phi$ and
    $\rho \in \intp{\Delta}_{\;\Psi;\Gamma}$, then
    $\tmatc(t, \vect\branch[\delta])_{\;\Psi;\Gamma}(\sigma; \rho) = \tmatc(t,
    \vect\branch)_{\;\Psi;\Gamma}(\sigma; \intp{\delta}^1_{\;\Psi;\Gamma}(\sigma;
    \rho))$.
    
  \item If $\ltyping[\Phi][\Delta'] 1{\branch}{\judge[\Gamma']{T'} \STo T}$,
    $\ltyping[\Phi][\Delta] 1 \delta \Delta'$, $\typing[\Psi]{\sigma}\Phi$ and
    $\rho \in \intp{\Delta}_{\;\Psi;\Gamma}$, then
    $\tnfbranch(\branch[\delta])_{\;\Psi;\Gamma}(\sigma; \rho) =
    \tnfbranch(\branch)_{\;\Psi;\Gamma}(\sigma; \intp{\delta}^1_{\;\Psi;\Gamma}(\sigma;
    \rho))$.

  \end{itemize}
\end{lemma}
\begin{proof}
  We do mutual induction. %
  We only consider the cases different from \Cref{lem:ct:intp-lsubst-0}.
  \begin{itemize}
  \item  $t = u^\delta$,
    \begin{mathpar}
      \inferrule
      {\ltyping[\Phi][\Delta'] 1 {\delta'} {\Gamma'} \\ u : (\judge[\Gamma'] T) \in \Phi}
      {\ltyping[\Phi][\Delta']{1}{u^{\delta'}}{T}}
    \end{mathpar}
    we compute
    \begin{align*}
      \intp{u^{\delta'}[\delta]}^1_{\;\Psi;\Gamma}(\sigma; \rho)
      &= \intp{\sigma(u)}^0_{\;\Psi;\Gamma}(\id; \intp{\delta' \circ
        \delta}^1_{\;\Psi;\Gamma}(\sigma; \rho)) \\
      &= \intp{\sigma(u)}^0_{\;\Psi;\Gamma}(\id;
        \intp{\delta'}^1_{\;\Psi;\Gamma}(\sigma;
        \intp{\delta}^1_{\;\Psi;\Gamma}(\sigma; \rho)))
        \byIH \\
      &= \intp{u^{\delta'}}^1_{\;\Psi;\Gamma}(\sigma;
        \intp{\delta}^1_{\;\Psi;\Gamma}(\sigma; \rho))
    \end{align*}
    as required.
  \item $t = \boxit s$,
    \begin{mathpar}
      \inferrule
      {\istype \Delta' \\ \ltyping[\Phi][\Gamma'] 0 s T'}
      {\ltyping[\Phi][\Delta']{1}{\boxit s}{\cont[\Gamma']{T'}}}
    \end{mathpar}
    we compute
    \begin{align*}
      \intp{\boxit s[\delta]}^1_{\;\Psi;\Gamma}(\sigma; \rho)
      &= \intp{\boxit s}^1_{\;\Psi;\Gamma}(\sigma; \rho) \\
      &= \boxit {(s[\sigma])} \\
      &= \intp{\boxit s}^1_{\;\Psi;\Gamma}(\sigma; \intp{\delta}^1_{\;\Psi;\Gamma}(\sigma; \rho))
    \end{align*}
    
  \item $t = \matc s \ \vect\branch$,
    \begin{mathpar}
      \inferrule
      {\ltyping[\Phi][\Delta'] 1 {s}{\cont[\Gamma']{T'}} \\ \ltyping[\Phi][\Delta'] 1 {\vect \branch}{\judge[\Gamma']{T'} \STo T}}
      {\ltyping[\Phi][\Delta'] 1 {\matc s\ \vect\branch} T}
    \end{mathpar}
    we compute
    \begin{align*}
      \intp{\matc s\ \vect\branch[\delta]}^1_{\;\Psi;\Gamma}(\sigma; \rho)
      = \intp{\matc {(s[\delta])}\ (\vect\branch[\delta])}^1_{\;\Psi;\Gamma}(\sigma; \rho)
    \end{align*}
    By IH, we know
    \begin{align*}
      \intp{s[\delta]}^1_{\;\Psi;\Gamma}(\sigma; \rho)
      = \intp{s}^1_{\;\Psi;\Gamma}(\sigma; \intp{\delta}^1_{\;\Psi;\Gamma}(\sigma; \rho))
    \end{align*}
    We split its result by cases:
    \begin{itemize}[label=Subcase]
    \item If $\intp{s[\delta]}^1_{\;\Psi;\Gamma}(\sigma; \rho) = \boxit {s'}$, then we
      just resort this case to the third statement of this lemma by IH.
      
    \item If $\intp{s[\delta]}^1_{\;\Psi;\Gamma}(\sigma; \rho) = v$ for some neutral
      $v$, then we just resort this case to the fourth statement of this lemma by IH.
    \end{itemize}    
  \end{itemize}

  Now we just need to verify the third and the fourth statements. %
  We will only pick representative cases to verify:
  \begin{itemize}[label=Case]
  \item $t = s~s'$ and $\vect\branch(s~s') = ?u~?u' \STo t'$ and for some $\iscore {S'}$,
    \begin{mathpar}
      \inferrule
      {\ltyping[\Psi][\Gamma']{0}{s}{S' \func T'} \\ \ltyping[\Psi][\Gamma']{0}{s'}{S'}}
      {\ltyping[\Psi][\Gamma']{0}{s\ s'}{T'}}
    \end{mathpar}
    we compute
    \begin{align*}
      & \tmatc(s~s', \vect\branch)_{\;\Psi;\Gamma}(\sigma;
        \intp{\delta}^1_{\;\Psi;\Gamma}(\sigma; \rho))
      \\
      =~& \intp{t'}^1_{\;\Psi;\Gamma}(\sigma, s/u, s'/u';
        \intp{\delta}^1_{\;\Psi;\Gamma}(\sigma; \rho)) \\
      =~& \intp{t'}^1_{\;\Psi;\Gamma}(\sigma, s/u, s'/u';
          \intp{\delta[p(p(\id))]}^1_{\;\Psi;\Gamma}(\sigma, s/u, s'/u'; \rho))
      \tag{where $p(p(\id)) : \Phi, u : (\judge[\Gamma']{S' \func T'}), u' : (\judge[\Gamma']{S'}) \To_g \Phi$}\\
      =~& \intp{t'[\delta[p(p(\id))]]}^1_{\;\Psi;\Gamma}(\sigma, s/u, s'/u'; \rho)
          \byIH \\
      =~& \tmatc(s~s', \vect\branch[\delta])_{\;\Psi;\Gamma}(\sigma; \rho)
    \end{align*}
  \item $t = u^{\delta'}$,
    \begin{mathpar}
      \inferrule
      {\ltyping[\Psi][\Gamma'] 0 {\delta'} {\Gamma''} \\ u : (\judge[\Gamma'']{T'}) \in \Psi}
      {\ltyping[\Psi][\Gamma']{0}{u^{\delta'}}{T'}}
    \end{mathpar}
    In this case, we have
    \begin{align*}
      \tmatc(u^{\delta'}, \vect\branch)_{\;\Psi;\Gamma}(\sigma;
      \intp{\delta}^1_{\;\Psi;\Gamma}(\sigma; \rho))
      & = \uparrow^T_{\;\Psi;\Gamma}(\matc {\boxit {u^{\delta'}}}\
      \tnfbranch(\vect\branch)_{\;\Psi;\Gamma}(\sigma;
      \intp{\delta}^1_{\;\Psi;\Gamma}(\sigma; \rho))) \\
      & = \uparrow^T_{\;\Psi;\Gamma}(\matc {\boxit {u^{\delta'}}}\
        \tnfbranch(\vect\branch[\delta])_{\;\Psi;\Gamma}(\sigma; \rho))
        \byIH \\
      &= \tmatc(u^{\delta'}, \vect\branch[\delta])_{\;\Psi;\Gamma}(\sigma; \rho)
    \end{align*}
    
  \item $\branch = \lambda x.?u \To t'$, 
    \begin{mathpar}
      \inferrule
      {\ltyping[\Phi, u : (\judge[\Gamma', x : S]{T'})][\Delta'] 1{t'} T}
      {\ltyping[\Phi][\Delta'] 1 {\lambda x. ?u \STo t'}{\judge[\Gamma'] S \func T' \STo T}}
    \end{mathpar}
    we compute
    \begin{align*}
      & \tnfbranch(\branch[\delta])_{\;\Psi;\Gamma}(\sigma; \rho) \\
      =~& \lambda x.?u \To
          \downarrow^T_{\;\Psi, u : (\judge[\Gamma', x :
          S]{T'});\Gamma}(\intp{t'[\delta[p(\id)]]}^1_{\;\Psi, u : (\judge[\Gamma', x
          : S]{T'});\Gamma}(\sigma[p(\id)], u/u; \rho[p(\id)]))
          \tag{where $p(\id) : \Phi, u : (\judge[\Delta', x : S]{T'}) \To_g \Phi$} \\
      =~& \lambda x.?u \To
          \downarrow^T_{\;\Psi, u : (\judge[\Gamma', x :
          S]{T'});\Gamma}(\intp{t'}^1_{\;\Psi, u : (\judge[\Gamma', x :
          S]{T'});\Gamma}(\sigma[p(\id)], u/u;
          \intp{\delta[p(\id)]}^1_{\;\Psi, u : (\judge[\Gamma', x : S]{T'});\Gamma}(\sigma[p(\id)], u/u; \rho[p(\id)]))) 
          \byIH \\
      =~& \lambda x.?u \To
          \downarrow^T_{\;\Psi, u : (\judge[\Gamma', x : S]{T'});\Gamma}(\intp{t'}^1_{\;\Psi, u : (\judge[\Gamma', x : S]{T'});\Gamma}(\sigma[p(\id)],
          u/u; \intp{\delta}^1_{\;\Psi, u : (\judge[\Gamma', x : S]{T'});\Gamma}(\sigma[p(\id)]; \rho[p(\id)]))) \\
      =~& \lambda x.?u \To
          \downarrow^T_{\;\Psi, u : (\judge[\Gamma', x : S]{T'});\Gamma}(\intp{t'}^1_{\;\Psi, u : (\judge[\Gamma', x : S]{T'});\Gamma}(\sigma[p(\id)],
          u/u; \intp{\delta}^1_{\;\Psi;\Gamma}(\sigma; \rho)[p(\id)]))
          \tag{by naturality} \\
      =~& \tnfbranch(\branch)_{\;\Psi;\Gamma}(\sigma; \intp{\delta}^1_{\;\Psi;\Gamma}(\sigma; \rho))
    \end{align*}
  \end{itemize}
  Hence we conclude the desired lemma.
\end{proof}

\Cref{lem:ct:intp-lsubst-0,lem:ct:intp-lsubst-1} in fact allow us to prove the
following corollary:
\begin{lemma}[Local substitutions]\labeledit{lem:ct:loc-subst}
  Assume $x$ is the topmost local variable, then
  $\intp{t[s/x]}^i_{\;\Psi;\Gamma}(\sigma;\rho) =
  \intp{t}^i_{\;\Psi;\Gamma}(\sigma;(\rho, \intp{s}^i_{\;\Psi;\Gamma}(\sigma;\rho)))$
\end{lemma}

\begin{lemma}\labeledit{lem:ct:glob-subst-gen} $ $
  \begin{itemize}
  \item If $\ltyping[\Phi][\Delta] 1 t T$, $\typing[\Psi']{\sigma'}\Psi$,
    $\typing[\Psi]{\sigma}\Phi$ and $\rho \in \intp{\Delta}_{\;\Psi';\Gamma}$, then
    $\intp{t}^1_{\;\Psi';\Gamma}(\sigma \circ \sigma';\rho) =
    \intp{t[\sigma]}^1_{\;\Psi';\Gamma}(\sigma';\rho)$.
    
  \item If $\ltyping[\Phi][\Delta] 1 \delta \Delta'$, $\typing[\Psi']{\sigma'}\Psi$,
    $\typing[\Psi]{\sigma}\Phi$ and $\rho \in \intp{\Delta}_{\;\Psi';\Gamma}$, then
    $\intp{\delta}^1_{\;\Psi';\Gamma}(\sigma \circ \sigma';\rho) =
    \intp{\delta[\sigma]}^1_{\;\Psi';\Gamma}(\sigma';\rho)$.
    
  \item If $\ltyping[\Psi'][\Delta'] 1 t {T'}$, $\typing[\Psi']{\sigma'}\Psi$,
    $\ltyping[\Phi][\Delta] 1 {\vect \branch}{\judge[\Delta']{T'} \STo T}$,
    $\typing[\Psi]{\sigma}\Phi$ and $\rho \in \intp{\Delta}_{\;\Psi';\Gamma}$, then
    $\tmatc(t, \vect\branch)_{\;\Psi';\Gamma}(\sigma \circ \sigma', \rho)
    = \tmatc(t, \vect\branch[\sigma])_{\;\Psi';\Gamma}(\sigma', \rho)$.

  \item If $\ltyping[\Phi][\Delta] 1 {\branch}{\judge[\Delta']{T'} \STo T}$,
    $\typing[\Psi']{\sigma'}\Psi$, $\typing[\Psi]{\sigma}\Phi$ and
    $\rho \in \intp{\Delta}_{\;\Psi';\Gamma}$, then
    $\tnfbranch(\branch)_{\;\Psi';\Gamma}(\sigma \circ \sigma'; \rho) =
    \tnfbranch(\branch[\sigma])_{\;\Psi';\Gamma}(\sigma'; \rho)$.
\end{itemize}
\end{lemma}
\begin{proof}
  We do mutual induction. %
  We only consider the modal cases because the other cases are identical to those in
  \Cref{lem:ct:gsubst-nat-gen}.
  \begin{itemize}[label=Case]
  \item  $t = u^\delta$,
    \begin{mathpar}
      \inferrule
      {\ltyping[\Phi][\Delta] 0 \delta {\Delta'} \\ u : (\judge[\Delta'] T) \in \Phi}
      {\ltyping[\Phi][\Delta]{0}{u^\delta}{T}}
    \end{mathpar}
    we compute
    \begin{align*}
      \intp{u^\delta}^1_{\;\Psi';\Gamma}(\sigma \circ \sigma';\rho)
      &= \intp{(\sigma \circ \sigma')(u)}^0_{\;\Psi';\Gamma}(\id,
        \intp{\delta}^1_{\;\Psi';\Gamma}(\sigma \circ \sigma';\rho)) \\
      &= \intp{\sigma(u)[\sigma']}^0_{\;\Psi';\Gamma}(\id,
        \intp{\delta[\sigma]}^1_{\;\Psi';\Gamma}(\sigma';\rho))
        \byIH \\
      &= \intp{\sigma(u)}^1_{\;\Psi';\Gamma}(\sigma',
        \intp{\delta[\sigma]}^1_{\;\Psi';\Gamma}(\sigma';\rho))
        \tag{by \Cref{lem:ct:intp-0-1}} \\
      &= \intp{\sigma(u)[\delta[\sigma]]}^1_{\;\Psi';\Gamma}(\sigma';\rho)
      \tag{by \Cref{lem:ct:intp-lsubst-1}} \\
      &= \intp{u^\delta[\sigma]}^1_{\;\Psi';\Gamma}(\sigma';\rho)
    \end{align*}
    
  \item $t = \boxit s$,
    \begin{mathpar}
      \inferrule
      {\istype \Delta \\ \ltyping[\Phi][\Delta'] 0 s T'}
      {\ltyping[\Phi][\Delta]{1}{\boxit s}{\cont[\Delta']{T'}}}
    \end{mathpar}
    we compute
    \begin{align*}
      \intp{\boxit{s}}^1_{\;\Psi';\Gamma}(\sigma \circ \sigma';\rho)
      &= \boxit{(s[\sigma\circ\sigma'])} \\
      &= \boxit{(s[\sigma][\sigma'])} \\
      &= \intp{\boxit{(s[\sigma])}}^1_{\;\Psi';\Gamma}(\sigma';\rho)
    \end{align*}
  \item $t = \matc s \ \vect\branch$,
    \begin{mathpar}
      \inferrule
      {\ltyping[\Phi][\Delta] 1 {s}{\cont[\Delta']{T'}} \\ \ltyping[\Phi][\Delta] 1 {\vect \branch}{\judge[\Delta']{T'} \STo T}}
      {\ltyping[\Phi][\Delta] 1 {\matc s\ \vect\branch} T}
    \end{mathpar}
    By IH
    \begin{align*}
      \intp{s}^1_{\;\Psi';\Gamma}(\sigma \circ \sigma'; \rho)
      = \intp{s[\sigma]}^1_{\;\Psi';\Gamma}(\sigma'; \rho)
    \end{align*}
    so we split its result by cases:
    \begin{itemize}[label=Subcase]
    \item If $\intp{s}^1_{\;\Psi';\Gamma}(\sigma \circ \sigma'; \rho) = \boxit {s'}$,
      then we just resort this case to the third statement of this lemma by IH.
      
    \item If $\intp{s}^1_{\;\Psi';\Gamma}(\sigma \circ \sigma'; \rho) = v$ for some neutral
      $v$, then we just resort this case to the fourth statement of this lemma by IH.
    \end{itemize}
  \end{itemize}

  Again, we only verify the representative cases:
  \begin{itemize}[label=Case]
  \item $t = s~s'$ and $\vect\branch(s~s') = ?u~?u' \STo t'$ and for some $\iscore {S'}$,
    \begin{mathpar}
      \inferrule
      {\ltyping[\Psi'][\Delta']{0}{s}{S' \func T'} \\ \ltyping[\Psi'][\Delta']{0}{s'}{S'}}
      {\ltyping[\Psi'][\Delta']{0}{s\ s'}{T'}}
    \end{mathpar}
    we compute
    \begin{align*}
      & \tmatc(s~s', \vect\branch)_{\;\Psi';\Gamma}(\sigma \circ \sigma'; \rho)
      \\
      =~& \intp{t'}^1_{\;\Psi';\Gamma}((\sigma \circ \sigma'), s/u, s'/u'; \rho) \\
      =~& \intp{t'}^1_{\;\Psi';\Gamma}((\sigma,u/u,u'/u') \circ (\sigma', s/u, s'/u'); \rho) \\
      =~& \intp{t'[\sigma,u/u,u'/u']}^1_{\;\Psi';\Gamma}((\sigma', s/u, s'/u'); \rho)
          \byIH \\
      =~& \tmatc(s~s', \vect\branch[\sigma])_{\;\Psi';\Gamma}(\sigma'; \rho)
    \end{align*}
  \item $t = u^{\delta'}$,
    \begin{mathpar}
      \inferrule
      {\ltyping[\Psi][\Delta'] 0 {\delta'} {\Delta''} \\ u : (\judge[\Delta'']{T'}) \in \Psi}
      {\ltyping[\Psi][\Delta']{0}{u^{\delta'}}{T'}}
    \end{mathpar}
    In this case, we have
    \begin{align*}
      \tmatc(u^{\delta'}, \vect\branch)_{\;\Psi';\Gamma}(\sigma \circ \sigma'; \rho)
      & = \uparrow^T_{\;\Psi';\Gamma}(\matc {\boxit {u^{\delta'}}}\
      \tnfbranch(\vect\branch)_{\;\Psi';\Gamma}(\sigma \circ \sigma'; \rho)) \\
      & = \uparrow^T_{\;\Psi';\Gamma}(\matc {\boxit {u^{\delta'}}}\
        \tnfbranch(\vect\branch[\sigma])_{\;\Psi';\Gamma}(\sigma'; \rho))
        \byIH \\
      & = \tmatc(u^{\delta'}, \vect\branch[\sigma])_{\;\Psi';\Gamma}(\sigma'; \rho)
    \end{align*}
    
  \item $\branch = \lambda x.?u \To t'$, 
    \begin{mathpar}
      \inferrule
      {\ltyping[\Phi, u : (\judge[\Delta', x : S]{T'})][\Delta] 1{t'} T}
      {\ltyping[\Phi][\Delta] 1 {\lambda x. ?u \STo t'}{\judge[\Delta'] S \func T' \STo T}}
    \end{mathpar}
    we compute
    \begin{align*}
      & \tnfbranch(\branch)_{\;\Psi';\Gamma}(\sigma \circ \sigma'; \rho) \\
      =~& \lambda x.?u \To
          \downarrow^T_{\;\Psi', u : (\judge[\Delta', x :
          S]{T'});\Gamma}(\intp{t'}^1_{\;\Psi', u : (\judge[\Delta', x :
          S]{T'});\Gamma}((\sigma \circ \sigma')[p(\id)], u/u; \rho[p(\id)])) 
          \tag{where $p(\id) : \Phi, u : (\judge[\Delta', x : S]{T'}) \To_g \Phi$} \\
      =~& \lambda x.?u \To
          \downarrow^T_{\;\Psi', u : (\judge[\Delta', x :
          S]{T'});\Gamma}(\intp{t'}^1_{\;\Psi', u : (\judge[\Delta', x :
          S]{T'});\Gamma}((\sigma[p(\id)], u/u) \circ (\sigma'[p(\id)], u/u);
          \rho[p(\id)])) \\
      =~& \lambda x.?u \To
          \downarrow^T_{\;\Psi', u : (\judge[\Delta', x : S]{T'});\Gamma}(\intp{t'[\sigma[p(\id)],
          u/u]}^1_{\;\Psi', u : (\judge[\Delta', x : S]{T'});\Gamma}((\sigma'[p(\id)], u/u); \rho[p(\id)]))
          \byIH \\
      =~& \tnfbranch(\branch[\sigma])_{\;\Psi';\Gamma}(\sigma'; \rho)
    \end{align*}
  \end{itemize}
\end{proof}

At this point, we have fully re-verified the necessary properties to establish the
completeness and soundness theorems.

\subsection{Completeness}

The completeness proof simply extends the development before. %
We write down the definitions here to be self-contained:
\begin{mathpar}
  \inferrule*
  {\ltyping 0 s T \\ \ltyping 0 t T \\\\ s = t}
  {\lsemtyeq 0{s}{t}{T}}

  \inferrule*
  {\ltyping 1 s T \\ \ltyping 1 t T \\\\
    \forall (\sigma; \rho) \in \intp{\Psi; \Gamma}^1_{\;\Phi;\Delta}.
    \intp{s}^1_{\;\Phi;\Delta}(\sigma;\rho) = \intp{t}^1_{\;\Phi;\Delta}(\sigma;\rho)}
  {\lsemtyeq 1{s}{t}{T}}

  \inferrule*
  {\lsemtyeq i{t}{t}{T}}
  {\lsemtyp i t T}

  \inferrule*
  {\ltyping 0 \delta \Delta \\ \ltyping 0 {\delta'} \Delta \\\\ \delta = \delta'}
  {\lsemtyeq 0{\delta}{\delta'}{\Delta}}

  \inferrule*
  {\ltyping 1 \delta \Delta \\ \ltyping 1 {\delta'}{\Delta} \\\\
    \forall (\sigma; \rho) \in \intp{\Psi; \Gamma}^1_{\;\Phi;\Delta}.
    \intp{\delta}^1_{\;\Phi;\Delta}(\sigma;\rho) = \intp{\delta'}^1_{\;\Phi;\Delta}(\sigma;\rho)}
  {\lsemtyeq 1{\delta}{\delta'}{\Delta}}

  \inferrule*
  {\lsemtyeq i{\delta}{\delta'}{\Delta}}
  {\lsemtyp i \delta \Delta}
\end{mathpar}

We let the following be definitions directly so we do not have to prove them:
\begin{mathpar}
  \inferrule
  {\iscore \Delta \\ \lsemtyeq 1 t {t'} T'}
  {\lsemtyeq 1 {\var x \STo t}{\var x \STo t'}{\judge[\Delta] T \STo T'}}
  
  \inferrule
  {\iscore \Delta \\ \lsemtyeq 1 t {t'} T'}
  {\lsemtyeq 1 {\ze \STo t} {\ze \STo t'}{\judge[\Delta] \Nat \STo T'}}

  \inferrule
  {\lsemtyeq[\Psi, u : (\judge[\Delta]\Nat)] 1 t {t'} T'}
  {\lsemtyeq 1 {\su ?u \STo t}{\su ?u \STo t'}{\judge[\Delta] \Nat \STo T'}}

  \inferrule
  {\lsemtyeq[\Psi, u : (\judge[\Delta, x : S]T)] 1 t {t'} T'}
  {\lsemtyeq 1 {\lambda x. ?u \STo t}{\lambda x. ?u \STo t'}{\judge[\Delta] S \func T \STo T'}}

  \inferrule
  {\forall \iscore S.~ \lsemtyeq[\Psi, u : (\judge[\Delta]S \func T), u' :
    (\judge[\Delta]S)] 1 t {t'} T'}
  {\lsemtyeq 1 {?u~?u' \STo t}{?u~?u' \STo t'}{\judge[\Delta] T \STo T'}}

  \inferrule
  {\lsemtyeq 1 {\branch}{\branch}{\judge[\Delta] T \STo T'}}
  {\lsemtyp 1 \branch{\judge[\Delta] T \STo T'}}
\end{mathpar}

For the same reason, we only focus on judgments at layer $1$.

\subsubsection{Congruence Rules}

\begin{lemma}
  \begin{mathpar}
    \inferrule
    {\lsemtyeq 1 \delta{\delta'} \Delta \\ u : (\judge[\Delta] T) \in \Psi}
    {\lsemtyeq{1}{u^\delta}{u^{\delta'}}{T}}
  \end{mathpar}
\end{lemma}
\begin{proof}
  Assume $(\sigma; \rho) \in \intp{\Psi; \Gamma}^1_{\;\Phi;\Delta}$,
  we have
  \begin{align*}
    \intp{\delta}^1_{\;\Phi;\Delta}(\sigma; \rho)
    = \intp{\delta'}^1_{\;\Phi;\Delta}(\sigma; \rho)
  \end{align*}
  This allows us to conclude
  \begin{align*}
    \intp{u^\delta}^1_{\;\Phi;\Delta}(\sigma; \rho)
    = \intp{\sigma(u)}^0_{\;\Phi;\Delta}(\id; \intp{\delta}^1_{\;\Phi;\Delta}(\sigma; \rho))
    = \intp{\sigma(u)}^0_{\;\Phi;\Delta}(\id; \intp{\delta'}^1_{\;\Phi;\Delta}(\sigma; \rho))
    = \intp{u^{\delta'}}^1_{\;\Phi;\Delta}(\sigma; \rho)
  \end{align*}
\end{proof}

\begin{lemma}
  \begin{mathpar}
    \inferrule
    {\lsemtyeq 1 {s}{s'}{\cont[\Delta'] T} \\ \lsemtyeq 1 {\vect \branch}{\vect \branch'}{\judge[\Delta'] T \STo T'}}
    {\lsemtyeq 1 {\matc s\ \vect\branch}{\matc s'\ \vect\branch'} T'}
  \end{mathpar}
\end{lemma}
\begin{proof}
  Assume $(\sigma; \rho) \in \intp{\Psi; \Gamma}^1_{\;\Phi;\Delta}$,
  we have
  \begin{align*}
    \intp{s}^1_{\;\Phi;\Delta}(\sigma; \rho)
    = \intp{s'}^1_{\;\Phi;\Delta}(\sigma; \rho)
    \in \Nf^{\cont[\Delta'] T}_{\;\Phi;\Delta}
  \end{align*}
  We perform case analysis:
  \begin{itemize}[label=Case]
  \item $\intp{s}^1_{\;\Phi;\Delta}(\sigma; \rho) = v$ for some neutral $v$. %
    In this case, we need to get into each branch and evaluate the normal form of the
    body. %
    We just pick one body and the others are similar: %
    \begin{mathpar}
      \inferrule
      {\lsemtyeq[\Psi, u : (\judge[\Delta', x : S]T)] 1 t {t'} T'}
      {\lsemtyeq 1 {\lambda x. ?u \STo t}{\lambda x. ?u \STo t'}{\judge[\Delta'] S \func T \STo T'}}
    \end{mathpar}
    We construct
    \begin{align*}
      (\sigma[p(\id)], u/u; \rho) \in \intp{\Psi, u : (\judge[\Delta', x : S]T);
      \Gamma}^1_{\;\Phi, u : (\judge[\Delta', x : S]T);\Delta}
    \end{align*}
    using which we have
    \begin{align*}
      \intp{t}^1_{\;\Phi, u : (\judge[\Delta', x : S]T);\Delta}(\sigma[p(\id)], u/u; \rho) = \intp{t'}^1_{\;\Phi, u : (\judge[\Delta', x : S]T);\Delta}(\sigma[p(\id)], u/u; \rho)
    \end{align*}
    This concludes our goal by congruence over reification. 
  \item $\intp{s}^1_{\;\Phi;\Delta}(\sigma; \rho) = \boxit {u^\delta}$, then this case is
    almost identical to the case above.
    
  \item $\intp{s}^1_{\;\Phi;\Delta}(\sigma; \rho) = \boxit {(t_1~t_2)}$ for some $t_1$
    and $t_2$. %
    Moreover we assume $t_1$ has type $S \func T$. %
    Since $t_1$ lives at layer $0$, we know $\iscore S$. %
    By looking up $\vect\branch$ and $\vect\branch'$, we have
    \begin{mathpar}
      \inferrule
      {\forall \iscore S.~ \lsemtyeq[\Psi, u : (\judge[\Delta]S \func T), u' :
        (\judge[\Delta]S)] 1 t {t'} T'}
      {\lsemtyeq 1 {?u~?u' \STo t}{?u~?u' \STo t'}{\judge[\Delta] T \STo T'}}
    \end{mathpar}
    We construct
    \begin{align*}
      (\sigma, t_1/u, t_2/u'; \rho) \in \intp{\Psi, u : (\judge[\Delta]S \func T), u' :
      (\judge[\Delta]S); \Gamma}^1_{\;\Phi;\Delta}
    \end{align*}
    From this we have
    \begin{align*}
      \intp{t}^1_{\;\Phi;\Delta}(\sigma, t_1/u, t_2/u'; \rho)
      = \intp{t'}^1_{\;\Phi;\Delta}(\sigma, t_1/u, t_2/u'; \rho)
    \end{align*}
    This is exactly the goal.
  \item All rest of the cases are cases for code. %
    They follow similarly to the case above.
  \end{itemize}
\end{proof}

\subsubsection{$\beta$ Rules}

\begin{lemma}
  \begin{mathpar}
    \inferrule
    {x : T \in \Delta \\ \lsemtyp 1 {\vect \branch}{\judge[\Delta] T \STo T'} \\
      \vect\branch(x) = \var x \STo t}
    {\lsemtyeq 1 {\matc {\boxit x}\ \vect\branch}{t} T'}  
  \end{mathpar}
\end{lemma}
\begin{proof}
  We know $\lsemtyp 1 t{T'}$ by premise. %
  Assume $(\sigma; \rho) \in \intp{\Psi; \Gamma}^1_{\;\Phi;\Delta}$,
  we compute
  \begin{align*}
    \intp{\matc {\boxit x}\ \vect\branch}^1_{\;\Phi;\Delta}(\sigma; \rho)
    = \intp{t}^1_{\;\Phi;\Delta}(\sigma; \rho)
  \end{align*}
  as desired.
\end{proof}

\begin{lemma}
  \begin{mathpar}
    \inferrule
    {\lsemtyp[\Psi][\Delta, x : S] 0 s T \\ \lsemtyp 1 {\vect \branch}{\judge[\Delta] S \func T \STo
        T'} \\ \vect\branch(\lambda x. s) = \lambda x. ?u \STo t}
    {\lsemtyeq 1 {\matc {\boxit {(\lambda x. s)}}\ \vect\branch}{t[s/u]}{T'}}
  \end{mathpar}
\end{lemma}
\begin{proof}
  We know $\lsemtyp[\Psi, u : (\judge[\Delta, x : S] T)] 1 t{T'}$ by premise. %
  Assume $(\sigma; \rho) \in \intp{\Psi; \Gamma}^1_{\;\Phi;\Delta}$,
  we compute
  \begin{align*}
    \intp{\matc {\boxit {(\lambda x. s)}}\ \vect\branch}^1_{\;\Phi;\Delta}(\sigma; \rho)
    = \intp{t}^1_{\;\Phi;\Delta}(\sigma, s[\sigma]/u; \rho)
  \end{align*}
  as desired by \Cref{lem:ct:glob-subst-gen}.
\end{proof}

\begin{lemma}
  \begin{mathpar}
    \inferrule
    {\lsemtyp 1 {\vect \branch}{\judge[\Delta] \Nat \STo T'} \\
      \vect\branch(\ze) = \ze \STo t}
    {\lsemtyeq 1 {\matc {\boxit \ze}\ \vect\branch}{t} T'}  

    \inferrule
    {\lsemtyp[\Psi][\Delta] 0 s\Nat \\ \lsemtyp 1 {\vect \branch}{\judge[\Delta] \Nat \STo T'}
      \\ \vect\branch(\su s) = \su ?u \STo t}
    {\lsemtyeq 1 {\matc {\boxit {(\su s)}}\ \vect\branch}{t[s/u]}{T'}}

    \inferrule
    {\lsemtyp[\Psi][\Delta] 0 t{S \func T} \\ \lsemtyp[\Psi][\Delta] 0 s S \\
      \lsemtyp 1 {\vect \branch}{\judge[\Delta] T \STo T'} \\
      \vect\branch(t~s) = ?u~?u' \STo t}
    {\lsemtyeq 1 {\matc {\boxit {(t~s)}}\ \vect\branch}{t[t/u, s/u']}{T'}}
  \end{mathpar}
\end{lemma}
\begin{proof}
  We follow the same pattern as above.
\end{proof}

\subsubsection{Fundamental Theorems}

Having all semantic judgments allows us to conclude the following fundamental
theorems:
\begin{theorem}[Fundamental] $ $
  \begin{itemize}
  \item If $\ltyping i t T$, then $\lsemtyp i t T$.
  \item If $\ltyping i \delta \Delta$, then $\lsemtyp i \delta \Delta$.
  \item If $\ltyping 1 {\branch}{\judge[\Delta] T \STo T'}$, then
    $\lsemtyp 1 {\branch}{\judge[\Delta] T \STo T'}$.
  \item If $\ltyequiv i s t T$, then $\lsemtyeq i s t T$.
  \item If $\ltyequiv i \delta{\delta'}\Delta$, then $\lsemtyeq i
    \delta{\delta'}\Delta$.
  \item If $\ltyequiv 1 {\branch}{\branch'}{\judge[\Delta] T \STo T'}$, then
    $\lsemtyeq 1 {\branch}{\branch'}{\judge[\Delta] T \STo T'}$.
  \end{itemize}
\end{theorem}

From here, we can prove the completeness theorem.
\begin{proof}[Proof of \Cref{thm:ct:compl}]
  Immediate by plugging in $\uparrow^{\Psi;\Gamma} \in \intp{\Psi;\Gamma}^1_{\;\Psi;\Gamma}$. 
\end{proof}

\subsection{Soundness}

The final theorem we would like to establish for contextual types is the soundness
proof. %
To our surprises, the soundness theorem is substantially more complex than the
$\tletbox$ formulation, due to pattern matching on code. %
The complexity comes from the need to maintain both syntactic and semantic information
of layer $0$ because pattern matching might potentially access both. %
This complexity leads to a rarely seen, inductively defined semantic judgments as
shown shortly. %

\subsubsection{Layer-$0$ Semantic Judgments}\labeledit{sec:prescont:semlayer0}

Since our core theory stays the same, we can completely reuse the gluing model at
layer $0$ as defined in~\Cref{sec:st:glue-0} without any adjustment. %
Its properties also hold automatically. %
However, what changes substantially due to pattern matching is the definition of
semantic judgment.
\begin{align*}
  \lSemtypPrime 0 t T &:= \forall \gamma : \Phi \To_g \Psi \tand \delta \sim \rho \in
                   \glu{\Gamma}^0_{\;\Phi;\Delta}. t[\gamma][\delta] \sim
                   \intp{t}^0_{\;\Phi;\Delta}(\gamma;\rho) \in
                   \glu{T}^0_{\;\Phi;\Delta}
\end{align*}
This judgment says that applying related environments from both syntax and semantics
to the same term $t$ at layer $0$ yields related results. %
This judgment is stable under syntactic equivalence due to 
\Cref{lem:st:glue-resp-equiv}. %
This means that dynamics is already introduced in this judgment and we in general have
lost the syntactic information of $t$. %
However, at layer $1$, in the semantic rule for typing pattern matching, we must have
access to the syntactic information and all subterms of the scrutinee must be
semantically well-typed. %
In short, $\lSemtypPrime 0 t T$ does not provide sufficient information to support
pattern matching. %
To remedy this issue, we re-introduce syntactic information back to the semantic
judgment. %
This leads to an inductively defined semantic judgment as follows:
\begin{mathpar}
  \inferrule
  {\iscore \Psi \\ \iscore \Gamma}
  {\lSemtyp 0 {\cdot}{\cdot}}

  \inferrule
  {\lSemtyp 0 {\delta}{\Delta} \\ \lSemtyp 0 {t}{T}}
  {\lSemtyp 0 {\delta, t/x}{\Delta, x : T}}

  \inferrule
  {\iscore \Psi \\ \iscore \Gamma}
  {\lSemtyp 0{\ze}{\Nat}}

  \inferrule
  {\lSemtyp 0{t}{\Nat} \\ \lSemtypPrime 0 {\su t} T}
  {\lSemtyp 0{\su t}{\Nat}}

  \inferrule
  {\lSemtyp 0 \delta \Delta \\ u : (\judge[\Delta] T) \in \Psi \\ \lSemtypPrime 0 {u^\delta} T}
  {\lSemtyp 0{u^\delta}{\cont[\Delta] T}}

  \inferrule
  {\lSemtyp[\Psi][\Gamma, x : S]{0}{t}{T} \\ \lSemtypPrime 0 {\lambda x. t}{S \func T}}
  {\lSemtyp 0{\lambda x. t}{S \func T}}

  \inferrule
  {\lSemtyp 0{t}{S \func T} \\ \lSemtyp 0{s}{S} \\ \lSemtypPrime 0 {t~s} T}
  {\lSemtyp 0{t\ s}{T}}
\end{mathpar}
This new definition clearly keeps track of all syntactic information that can be used
later in pattern matching. %
The following lemma states that $\lSemtyp 0 {t} T$ is a richer notion than $\lSemtypPrime 0
t T$.
\begin{lemma}
  If $\lSemtyp 0 {t} T$, then $\lSemtypPrime 0 t T$.
\end{lemma}
\begin{proof}
  By case analysis. 
\end{proof}


Finally, we adjust the definition of semantic judgment for global substitutions:
\begin{mathpar}
  \inferrule*
  {\iscore \Psi}
  {\Semtyp[\Psi]{\cdot}{\cdot}}

  \inferrule*
  {\Semtyp[\Psi]{\sigma}{\Phi} \\ \lSemtyp 0 {t}{T}}
  {\Semtyp[\Psi]{\sigma, t/u}{\Phi, u : (\judge T)}}
\end{mathpar}

Monotonicity holds:
\begin{lemma}[Monotonicity] $ $
  \begin{itemize}
  \item If $\lSemtyp[\Psi] 0 {t}{T}$ and $\gamma : \Psi' \To_g \Psi$, then
    $\lSemtyp[\Psi'][\Phi] 0 {t[\gamma]}{T}$.
  \item If $\lSemtyp[\Psi] 0 {\delta}{\Delta}$ and $\gamma : \Psi' \To_g \Psi$, then
    $\lSemtyp[\Psi'] 0 {\delta[\gamma]}{T}$.
  \end{itemize}
\end{lemma}

\begin{lemma}[Monotonicity]
  If $\Semtyp[\Psi]{\sigma}{\Phi}$ and $\gamma : \Psi' \To_g \Psi$, then
  $\Semtyp[\Psi']{\sigma[\gamma]}{\Phi}$.
\end{lemma}

Fundamental theorem also holds because our core language does not have substantial
change. %
The case of global variables is very simple to check due to \Cref{lem:st:ne-glue}. 
\begin{theorem}[Fundamental] $ $
  \begin{itemize}
  \item If $\ltyping 0 t T$, then $\lSemtyp 0 t T$.
  \item If $\ltyping 0 \delta \Delta$, then $\lSemtyp 0 \delta \Delta$.
  \end{itemize}
\end{theorem}
\begin{proof}
  This is effectively the same to prove that all $\lSemtypPrime {0} {t} T$ premises in the
  definition of $\lSemtyp 0 t T$ are redundant, but this is just the same as the
  fundamental theorems at layer $0$ in our previous system.
\end{proof}

\subsubsection{Layer-$1$ Gluing Model}

With the adjustments at layer $0$, things fall into one piece automatically at layer
$1$. %
We follow our noses to make adjustments at layer $1$. %
First, we generalize the gluing of contextual types:
\begin{mathpar}
  \inferrule*
  {\ltyequiv 1 t{\boxit s}{\cont[\Delta]T} \\ \lSemtyp[\Psi][\Delta] 0 {s} T}
  {t \sim {\boxit s} \in {\cont[\Delta]T}_{\;\Psi;\Gamma}}

  \inferrule*
  {\ltyequiv 1 t v{\cont[\Delta]T}}
  {t \sim v \in {\cont[\Delta]T}_{\;\Psi;\Gamma}}
\end{mathpar}
Notice here in the first case, $\lSemtyp[\Psi][\Delta] 0 {t} T$ is the inductively
defined semantic judgment at layer $0$, so the information of $t$ is preserved
also in the semantics. %
Then we define the gluing of types at layer $1$:
\begin{align*}
  \glu{T}^1_{\;\Psi;\Gamma} &\subseteq \Exp \times \intp{T}_{\;\Psi;\Gamma} \\
  \glu{\Nat}^1_{\;\Psi;\Gamma} &:= \Nat_{\;\Psi;\Gamma} \\
  \glu{\cont[\Delta] T}^1_{\;\Psi;\Gamma} &:= {\cont[\Delta] T}_{\;\Psi;\Gamma} \\ 
  \glu{S \func T}^1_{\;\Psi;\Gamma} &:= \{(t, a) \sep \forall
                              \gamma; \tau : \Phi; \Delta \To \Psi;\Gamma, s \sim b \in
                              \glu{S}^1_{\;\Phi;\Delta}. t[\gamma; \tau]\ s \sim a(\gamma; \tau, b) \in
                              \glu{T}^1_{\;\Phi;\Delta} \} 
\end{align*}

Finally we define the semantic judgments for both types and local substitutions at
layer $1$:
\begin{definition}
  \begin{align*}
    \lSemtyp 1 t T &:= \forall \Semtyp[\Phi]\sigma\Psi \tand \delta \sim \rho \in
                     \glu{\Gamma}^1_{\;\Phi;\Delta}. t[\sigma;\delta] \sim
                     \intp{t}^1_{\;\Phi;\Delta}(\sigma;\rho) \in
                     \glu{T}^1_{\;\Phi;\Delta} \\
    \lSemtyp 1 {\delta'} {\Delta'} &:= \forall \Semtyp[\Phi]\sigma\Psi \tand \delta \sim \rho \in
                     \glu{\Gamma}^1_{\;\Phi;\Delta}. \delta'[\sigma] \circ \delta \sim
                                     \intp{\delta'}^1_{\;\Phi;\Delta}(\sigma;\rho) \in
                                     \glu{\Delta'}^1_{\;\Phi;\Delta} 
  \end{align*}
\end{definition}

\subsubsection{Semantic Rules at Layer $1$}

Finally we move on to the semantic rules at layer $1$. %
We skip those terms from core type theory because the proof would be identical to the
previous development.
\begin{lemma}
  \begin{mathpar}
    \inferrule
    {\iscore \Psi \\ \iscore \Gamma}
    {\lSemtyp 1 {\cdot}{\cdot}}

    \inferrule
    {\lSemtyp 1 {\delta}{\Delta} \\ \lSemtyp 1 {t}{T}}
    {\lSemtyp 1 {\delta, t/x}{\Delta, x : T}}

    \inferrule
    {\istype \Gamma \\ \lSemtyp[\Psi][\Delta] 0 t T}
    {\lSemtyp{1}{\boxit t}{\cont[\Delta] T}}
\end{mathpar}
\end{lemma}
\begin{proof}
  Immediate by following the definition. 
\end{proof}

\begin{lemma}
  \begin{mathpar}
    \inferrule
    {\lSemtyp 1{\delta'}{\Delta'} \\ u : (\judge[\Delta'] T) \in \Psi}
    {\lSemtyp 1{u^{\delta'}}{\cont[\Delta'] T}}
  \end{mathpar}
\end{lemma}
\begin{proof}
  Assume $\Semtyp[\Phi]\sigma\Psi$ and
  $\delta \sim \rho \in \glu{\Gamma}^1_{\;\Phi;\Delta}$. %
  Then we should prove
  \begin{align*}
    \sigma(u)[\delta'[\sigma] \circ \delta] &\sim
    \intp{\sigma(u)}^0_{\;\Phi;\Delta}(\id; \intp{\delta'}^1_{\;\Phi;\Delta}(\sigma;
                                              \rho)) \in \glu{T}^0_{\;\Phi;\Delta}
                                              =\glu{T}^1_{\;\Phi;\Delta}
                                              \tag{by \Cref{lem:st:glu-0-to-1}}
  \end{align*}
  Moreover, we know $\lSemtyp[\Psi][\Delta']0{\sigma(u)}T$, $\id : \Psi \To_g \Psi$
  and $\delta'[\sigma] \circ \delta \sim \intp{\delta'}^1_{\;\Phi;\Delta}(\sigma;
  \rho)  \in \glu{\Delta'}^1_{\;\Phi;\Delta}$. %
  This is sufficient to prove the goal.
\end{proof}

Finally, we prove the semantic rule for pattern matching. %
We define the following judgments for the semantic well-typedness of branches:
\begin{mathpar}
  \inferrule
  {\iscore \Delta \\ \lSemtyp 1 t T'}
  {\lSemtyp 1 {\var x \STo t}{\judge[\Delta] T \STo T'}}
  
  \inferrule
  {\iscore \Delta \\ \lSemtyp 1 t T'}
  {\lSemtyp 1 {\ze \STo t}{\judge[\Delta] \Nat \STo T'}}

  \inferrule
  {\lSemtyp[\Psi, u : (\judge[\Delta]\Nat)] 1 t T'}
  {\lSemtyp 1 {\su ?u \STo t}{\judge[\Delta] \Nat \STo T'}}

  \inferrule
  {\lSemtyp[\Psi, u : (\judge[\Delta, x : S]T)] 1 t T'}
  {\lSemtyp 1 {\lambda x. ?u \STo t}{\judge[\Delta] S \func T \STo T'}}

  \inferrule
  {\forall \iscore S.~ \lSemtyp[\Psi, u : (\judge[\Delta]S \func T), u' : (\judge[\Delta]S)] 1 t T'}
  {\lSemtyp 1 {?u~?u' \STo t}{\judge[\Delta] T \STo T'}}
\end{mathpar}

\begin{lemma}
  \begin{mathpar}
    \inferrule
    {\lSemtyp 1 {s}{\cont[\Delta'] T} \\ \lSemtyp 1 {\vect \branch}{\judge[\Delta'] T \STo T'}}
    {\lSemtyp 1 {\matc s\ \vect\branch} T'}
  \end{mathpar}
\end{lemma}
\begin{proof}
  Assume $\Semtyp[\Phi]\sigma\Psi$ and
  $\delta \sim \rho \in \glu{\Gamma}^1_{\;\Phi;\Delta}$. %
  We have
  \begin{align*}
    s[\sigma;\delta] \sim \intp{s}^1_{\;\Phi;\Delta}(\sigma;\rho)
    \in {\cont[\Delta'] T}_{\;\Phi;\Delta}
  \end{align*}
  By \Cref{lem:st:glue-nf}, we have
  \begin{align*}
    \ltyequiv[\Phi][\Delta] 1 {s[\sigma;\delta]}{\downarrow^{\cont[\Delta']T}_{\;\Phi;\Delta}(\intp{s}^1_{\;\Phi;\Delta}(\sigma;\rho))}{\cont[\Delta']T}
  \end{align*}
  
  Then we do case analysis:
  \begin{itemize}[label=Case]
  \item $\intp{s}^1_{\;\Phi;\Delta}(\sigma;\rho) = v$ for some neutral $v$. %
    Then
    \begin{align*}
      \ltyequiv[\Phi][\Delta] 1 {s[\sigma;\delta]}{\downarrow^{\cont[\Delta']T}_{\;\Phi;\Delta}(v)}{\cont[\Delta']T}
    \end{align*}
    Then we evaluate the normal forms of all branches. %
    We just pick one branch to verify and other branches are similar.
    \begin{mathpar}
      \inferrule
      {\lSemtyp[\Psi, u : (\judge[\Delta', x : S]{T''})] 1 t T'}
      {\lSemtyp 1 {\lambda x. ?u \STo t}{\judge[\Delta'] S \func T'' \STo T'}}
    \end{mathpar}
    We construct
    \begin{align*}
      \Semtyp[\Phi, u : (\judge[\Delta', x : S]{T''})]{\sigma[p(\id)], u/u}{\Psi, u : (\judge[\Delta', x : S]{T''})}
    \end{align*}
    and
    $\delta[p(\id)] \sim \rho[p(\id)] \in \glu{\Gamma}^1_{\;\Phi, u : (\judge[\Delta', x : S]{T''});\Delta}$. %
    We have
    \begin{align*}
      t[\sigma[p(\id)], u/u;\delta[p(\id)]] \sim \intp{t}^1_{\;\Phi;\Delta}(\sigma[p(\id)], u/u;\rho[p(\id)])
      \in \glu{T'}^1_{\;\Phi, u : (\judge[\Delta', x : S]T);\Delta} 
    \end{align*}
    By \Cref{lem:st:glue-nf}, we have
    \begin{align*}
      \ltyequiv[\Phi, u : (\judge[\Delta', x : S]{T''})][\Delta] 1 {t[\sigma[p(\id)],
      u/u;\delta[p(\id)]]}{\downarrow^{T'}_{\;\Phi, u : (\judge[\Delta', x :
      S]T);\Delta}(\intp{t}^1_{\;\Phi, u : (\judge[\Delta', x :
      S]T);\Delta}(\sigma[p(\id)], u/u;\rho[p(\id)]))}{T'}
    \end{align*}

    We apply similar proofs to other branches, yielding
    \begin{align*}
      \ltyequiv[\Phi][\Delta] 1 {\vect \branch[\sigma;\delta]}{\vect \nbranch}{\judge[\Delta'] T \STo T'}
    \end{align*}
    where $\vect\nbranch := \tnfbranch(\vect\branch)_{\;\Phi;\Delta}(\sigma; \rho)$. %
    Since $\matc{v}\ \vect\nbranch$ is neutral, we apply
    \Cref{lem:st:ne-glue,lem:st:glue-resp-equiv} to obtain the goal.
    
  \item $\intp{s}^1_{\;\Phi;\Delta}(\sigma;\rho) = \boxit {s'}$ and
    $\lSemtyp[\Phi][\Delta'] 0 {s'}T$. %
    We further analyze $\lSemtyp[\Phi][\Delta'] 0 {s'}T$.

    \begin{itemize}[label=Subcase]
    \item $s' = u^{\delta'}$, then we hit the neutral case. %
      We follow the previous case and construct a neutral again.
      
    \item $s' = \lambda x. s''$ and
      $\lambda x. ?u \STo t := \vect\branch(\lambda x. s'')$,
      \begin{mathpar}
        \inferrule
        {\lSemtyp[\Phi][\Delta', x : S]{0}{s''}{T''} \\ \lSemtyp[\Phi][\Delta'] 0 {' \lambda x. s''}{S \func T''}}
        {\lSemtyp[\Phi][\Delta'] 0{\lambda x. s''}{S \func T''}}
      \end{mathpar}
      We evaluate as follows
      \begin{align*}
        \intp{\matc s \ \vect\branch}^1_{\;\Phi;\Delta}(\sigma; \rho)
        &= \tmatc(\lambda x. s'', \vect\branch)_{\;\Phi;\Delta}(\sigma; \rho) \\
        &= \intp{t}^1_{\;\Phi;\Delta}(\sigma, s''/u; \rho)
      \end{align*}
      Notice that
      \begin{align*}
        \Semtyp[\Phi]{\sigma, s''/u}{\Psi, u : (\judge[\Delta', x : S]{T''})}
      \end{align*}
      Hence
      \begin{align*}
        t[\sigma, s''/u; \delta]
        \sim \intp{t}^1_{\;\Phi;\Delta}(\sigma, s''/u; \rho) \in \glu{T'}^1_{\;\Phi;\Delta}
      \end{align*}
      Moreover,
      \begin{align*}
        t[\sigma[p(\id)],u/u; \delta[p(\id)]][\id, s''/u]
        &= t[(\sigma[p(\id)],u/u) \circ (\id, s''/u); \delta[p(\id)]] \\
        &= t[\sigma, s''/u; \delta]
      \end{align*}
      Therefore
      \begin{align*}
        \ltyequiv[\Phi][\Delta] 1 {\matc s\ \vect\branch[\sigma;\delta]}{t[\sigma, s''/u; \delta]}T'
      \end{align*}
      and
      \begin{align*}
        \intp{\matc s\ \vect\branch[\sigma;\delta]}^1_{\;\Phi;\Delta}(\sigma; \rho) = \intp{t}^1_{\;\Phi;\Delta}(\sigma, s''/u; \rho)
      \end{align*}
      Combining both and \Cref{lem:st:glue-resp-equiv} gives the desired goal.
      
    \item The other cases of $s'$ follow the same principle by using
      \Cref{lem:st:glue-resp-equiv} and prove the goals.
    \end{itemize}
  \end{itemize}
\end{proof}

\subsubsection{Fundamental Theorems at Layer $1$}

\begin{theorem}[Fundamental] $ $
  \begin{itemize}
  \item If $\ltyping 1 t T$, then $\lSemtyp 1 t T$.
  \item If $\ltyping 1 \delta \Delta$, then $\lSemtyp 1 \delta \Delta$.
  \end{itemize}
\end{theorem}

This allows us to conclude the soundness theorem.
\begin{proof}[Proof of \Cref{thm:ct:sound}]
  This is a direct consequence of fundamental theorems, the same as before.
\end{proof}

Now we have finished our development of the NbE proof of 2-layered contextual modal
type theory.

\section{Adding Recursor for Natural Numbers}

For completeness, in this section, we add recursors for natural numbers to the type
theory in \Cref{sec:contextual}. %
This addition is rather straightforward and is essentially identical to a standard
presheaf model of STLC with weak natural numbers, so readers uninterested in this
section might safely skip it. %
We include this section in order to better match up with our published version.

\subsection{Adjustments to Syntax}

In this subsection, we first update the syntactic definitions of the system. %
First, we update the syntax as follows:
\begin{align*}
  s, t &:= \cdots \sep \recn T s {x~y. s'} t \tag*{(Terms, $\Exp$)} \\
  v &:= \cdots \sep \recn T w {x~y. w'} v \tag*{(Neutral form, $\Ne$)}
\end{align*}
In the syntax, $\recn T s {x~y. s'} t$ eliminates the natural number $t$. %
Here, $s$ is the base case and $s'$ is the step case. %
$x$ is the predecessor in $s'$ and $y$ is the recursive call, i.e. inductive
hypothesis in the proof theoretic sense. 

Then, we add the following typing rules for the recursor:
\begin{mathpar}
  \inferrule
  {\ltyping i s T \\ \ltyping[\Psi][\Gamma, x : \Nat, y : T]i{s'}T \\ \ltyping i t \Nat}
  {\ltyping i{\recn T s {x~y. s'} t}T}

  \inferrule
  {\ltyping[\Psi, u : (\judge[\Delta]T), u' : (\judge[\Delta, x : \Nat, y : T]T),
    u'' : (\judge[\Delta]\Nat)] 1 t T'}
  {\ltyping 1 {\recn{T}{?u}{x~y.?u'}{?u''} \STo t}{\judge[\Delta] T \STo T'}}
\end{mathpar}
Since the recursor is extended to STLC, i.e. it is available at both layers including
layer $0$, we need to add a case for pattern matching to handle code of the recursor
as in the second rule. 

Next, we extend $\ltyping 1{\vect\branch}{\judge[\Delta] T \STo T'}$ to ensure that
the branch for the recursor is included properly so that this judgment remains \emph{covering}:
\begin{mathpar}
  \inferrule
  {\forall b \in \vect \branch ~.~ \ltyping 1 {\branch}{\judge[\Delta] \Nat \STo T'}
    \\
    \branch_\ze = \ze \STo t \text{ for some } t \\
    \branch_\tsucc = \su ?u \STo t \text{ for some } t \\
    \mhighlight{\branch_\trec = \recn{T'}{?u}{x~y.?u'}{?u''} \STo t \text{ for some } t} \\
    \branch_\tapp = ?u~?u' \STo t \text{ for some } t \\
    \forall x : \Nat \in \Delta ~.~ b_x = \var x \STo t \text{ for some } t \\
  \vect\branch \text{ is a permutation of } \{ \branch_\ze, \branch_\tsucc,
  \mhighlight{\branch_\trec}, \branch_\tapp, \branch_x \text{ for all } x : \Nat \in \Delta \}}
  {\ltyping 1{\vect\branch}{\judge[\Delta] \Nat \STo T'}}

  \inferrule
  {\forall b \in \vect \branch ~.~ \ltyping 1 {\branch}{\judge[\Delta] S \func T \STo T'}
    \\
    \branch_\lambda = \lambda x. ?u \STo t \text{ for some } t \\
    \mhighlight{\branch_\trec = \recn{T'}{?u}{x~y.?u'}{?u''} \STo t \text{ for some } t} \\
    \branch_\tapp = ?u~?u' \STo t \text{ for some } t \\
    \forall x : S \func T \in \Delta ~.~ b_x = \var x \STo t \text{ for some } t \\
  \vect\branch \text{ is a permutation of } \{ \branch_\lambda,
  \mhighlight{\branch_\trec}, \branch_\tapp, \branch_x \text{ for all } x : S \func T \in \Delta \}}
  {\ltyping 1{\vect\branch}{\judge[\Delta] S \func T \STo T'}}
\end{mathpar}
This case must be included in both rules because it is possible to eliminate a natural
number into any type.

At last, we include the equivalence rules. %
There are three kinds of rules that we must add. %
Since we add another term construct, we must include its congruence rule. %
Then we specify the $\beta$ equivalence when encountering $\ze$ and $\tsucc$,
respectively. %
At last, we specify the $\beta$ equivalence for $\tmatc$. 
\begin{mathpar}
  \inferrule
  {\ltyequiv i {s_0}{s_0'} T \\ \ltyequiv[\Psi][\Gamma, x : \Nat, y : T]i{s_1}{s_1'}T
    \\ \ltyequiv i t {t'}\Nat}
  {\ltyequiv i{\recn T{s_0}{x~y. s_1}t}{\recn T{s_0'}{x~y. s_1'}{t'}}T}

  \inferrule
  {\ltyping 1 s T \\ \ltyping[\Psi][\Gamma, x : \Nat, y : T]i{s'}T \\ \ltyping 1 t \Nat}
  {\ltyequiv 1{\recn T s {x~y. s'} \ze}{s}T}

  \inferrule
  {\ltyping 1 s T \\ \ltyping[\Psi][\Gamma, x : \Nat, y : T]i{s'}T \\ \ltyping 1 t \Nat}
  {\ltyequiv 1{\recn T s {x~y. s'}{\su t}}{s'[t/x,\recn T s {x~y. s'}{t}/y]}T}

  \inferrule
  {\ltyping 0 s T \\ \ltyping[\Psi][\Gamma, x : \Nat, y : T]0{s'}T \\ \ltyping 1 {t'} \Nat \\
    \ltyping 1 {\vect \branch}{\judge[\Delta] T \STo T'} \\
    \vect\branch(\recn T s {x~y.s'}{t'}) = \recn{T}{?u}{x~y.?u'}{?u''} \STo t}
  {\ltyequiv 1 {\matc {\boxit{(\recn T s {x~y.s'}{t'})}}\ \vect\branch}{t[s/u,s'/u',t'/u'']} T'} 
\end{mathpar}

\subsection{Semantics of Recursor}

The presheaf model in \Cref{sec:prescont} can be extended with the recursor on natural
numbers so that it gives a normalization algorithm with the recursor.

First, we extend the evaluation function:
\begin{align*}
  \intp{\recn T s{x~y. s'} t}^i_{\;\Psi;\Gamma}(\sigma;\rho) :=
  \trec^i_T(s,{x~y. s'}, \intp{t}^i_{\;\Psi;\Gamma}(\sigma;\rho))_{\;\Psi;\Gamma}(\sigma;\rho)
\end{align*}
where $\trec$ is the recursor implemented in the semantics. %
Note that by typing, we know that $\intp{t}^i_{\;\Psi;\Gamma}(\sigma;\rho) \in
\intp{\Nat}_{\;\Psi;\Gamma} = \Nf^\Nat_{\;\Psi;\Gamma}$.
The $\trec$ semantic function does recursion intuitively:
\begin{align*}
    \trec^i_T(\_)_{\;\Psi;\Gamma}
    &:
      \ltyping[\Phi][\Delta] i s T \to \ltyping[\Phi][\Delta, x : \Nat, y : T] i{s'} T \to
   \Nf^\Nat_{\;\Psi;\Gamma} \to \intp{\Phi;\Delta}^i_{\;\Psi;\Gamma} \to
      \intp{T}_{\;\Psi;\Gamma}
      \tag*{\textbf{(Semantic Recursor for \Nat)}} \\  
    \trec^i_T(s,{x~y. s'}, \ze)_{\;\Psi;\Gamma}(\sigma;\rho)
    &:= \intp{s}^i_{\;\Psi;\Gamma}(\sigma;\rho) \\
    \trec^i_T(s,{x~y. s'}, \su w)_{\;\Psi;\Gamma}(\sigma;\rho)
    &:= \intp{s'}^i_{\;\Psi;\Gamma}(\sigma; (\rho, w, \trec^i_T(s,{x~y. s'},
      w)_{\;\Psi;\Gamma}(\sigma;\rho))) \\
    \trec^i_T(s,{x~y. s'}, v)_{\;\Psi;\Gamma}(\sigma;\rho)
    &:= \uparrow^T_{\;\Psi;\Gamma}(\recn T w {x~y. w'} v)
      \tag*{(where $w :=
    \downarrow^T_{\;\Psi;\Gamma}(\intp{s}^i_{\;\Psi;\Gamma}(\sigma;\rho))$, and
    $w' := \downarrow^T_{\;\Psi;\Gamma, x : \Nat, y : T}(\intp{s'}^i_{\;\Psi;\Gamma, x
      : \Nat, y : T}(\sigma';\rho''))$,)} \\
  \tag*{(with $(\sigma'; \rho') := (\sigma; \rho)[\id; p(p(\id))]:
    \intp{\Phi;\Delta}^i_{\;\Psi;\Gamma, x : \Nat, y : T}$, and $\rho'' := \rho', x, \uparrow^T_{\;\Psi;\Gamma}(y)$)}
\end{align*}
The $\trec$ function operates on a normal form of a \Nat. %
It acts on both layers. %
When $i = 0$, it means that this function is running some generated code, i.e. code
that are lifted by $\tletbox$. %
Checking the typing rules, there are three possibilities. %
First, it could be a $\ze$. %
Then the first case matches. %
Second, it could be a $\tsucc$ of something. %
Then the second case applies. %
Notice that the evaluation environment is extended with the predecessor and the result
of the recursive call. %
Finally, it could be some neutral form. %
Then the third case applies. %
In this case, we must apply reflection ($\uparrow^T_{\;\Psi;\Gamma}$) so we must
compose a neutral form. %
This is done by composing a neutral form of the recursor. %
To do that, we need the normalize $s$ and $s'$, which is what the details are about. %

After describing the recursor itself, let us describe how the \tmatc function acts on
the recursor. %
This is done by extending the \tmatc function with another case for the recursor:
\begin{align*}
      \tmatc(\recn T s{x~y. s'} t, \vect\branch)_{\;\Psi;\Gamma}(\rho; \rho)
    &:= \intp{t}^1_{\;\Psi;\Gamma}(\sigma, s/u, s'/u', t/u''; \rho)
      \tag*{(where ${\recn{T}{?u}{x~y.?u'}{?u''} \STo t} := \vect\branch(\recn T s{x~y. s'} t)$)}
\end{align*}
Much as expected, the \tmatc function only looks up the right branch and then evaluate
the body. %

Similarly, we also need to extend the \tnfbranch function. %
This is simple as we just need to follow other cases:
\[
  \begin{array}{r@{~}l}
    \tnfbranch(\recn{T}{?u}{x~y.?u'}{?u''} \STo t)_{\;\Psi; \Gamma}(\sigma; \rho)
    &:= \\
    \multicolumn{2}{r}{\qquad\qquad\recn{T}{?u}{x~y.?u'}{?u''} \STo \downarrow^{T'}_{\;\Psi';
    \Gamma}(\intp{t}^1_{\;\Psi'; \Gamma}(\sigma', u^\id/u, u'^\id/u', u''^\id/u''; \rho'))} \\
    \multicolumn{2}{r}{\mbox{(where $\Psi' := u : (\judge[\Delta]T), u' : (\judge[\Delta, x : \Nat, y : T]T),
    u'' : (\judge[\Delta]\Nat)$,)}}  \\
    \multicolumn{2}{r}{\mbox{(and $(\sigma'; \rho') := (\sigma; \rho)[p(p(p(\id))); \id]\in \intp{\Phi;\Delta}^1_{\;\Psi';\Gamma}$)}}
  \end{array}
\]
The fix for completeness and soundness proofs is identical to those for STLC with weak
natural numbers. %
In particular, we do not need to amend the definitions of semantic judgments as we
only add an elimination principle for an existing type. %
At this point, we have finished extending our normalization algorithm with a recursor
for \Nat.






\end{document}